\documentclass[11pt]{elsarticle}
\usepackage{amssymb,amsmath}
\usepackage{graphicx}
\usepackage{appendix} 
\usepackage{algorithm}
\usepackage{algpseudocode} 
\usepackage{comment}
\usepackage{tikz}
\usepackage{rotating} 
\usepackage{multirow}
\usepackage{booktabs}
\usetikzlibrary{shapes.geometric}
\usetikzlibrary{shapes,arrows}
\usepackage{lineno}
\usepackage[colorlinks]{hyperref}
\usepackage{amsthm}
\usepackage{bm}
\newtheorem{thm}{Theorem}
\newtheorem{lem}[thm]{Lemma}

\newtheorem{defin}[thm]{Definition}
\theoremstyle{remark}
\newtheorem{rem}{Remark}[section]

\date{\today}
\journal{Journal of Computational Physics}
\begin{document}
\begin{frontmatter}

\title{Phase Transitions, Non-Extremality (Reconstruction), and Markov Entropy Rate for the Mixed Spin-$(s,\tfrac12)$ Ising Model on a Cayley Tree of Order Three}
\author[Harran]{Hasan Ak\i n \corref{cor1}}
\ead{akinhasan@harran.edu.tr;
akinhasan25@gmail.com}\fntext[fn1]{ORCID: \href{https://orcid.org/0000-0001-6447-4035}{0000-0001-6447-4035}}
\address[Harran]{Department of Mathematics, Faculty of Arts and Sciences, Harran
University, TR63050, Sanliurfa, Turkey}

\date{\today}
\begin{abstract}
We investigate the mixed spin-$(s,\tfrac12)$ Ising model on a Cayley tree of order three ($k=3$), extending the approach of \cite{Akin2024}. For the representative case $s=5$, the associated recursion leads to an 11-dimensional dynamical system, and phase-transition regions are examined via the local stability of the disordered (symmetric) fixed point, detected through the condition $|\lambda_{\max}|\ge 1$ for the Jacobian matrix. To study extremality (non-reconstruction) of the disordered phase, we represent translation-invariant splitting Gibbs measures by tree-indexed Markov chains and compute the relevant Dobrushin coefficients. At the symmetric fixed point we obtain explicit transition kernels and the induced two-step kernel on the spin-$\tfrac12$ layer; its second eigenvalue $\lambda_2$ yields a spectral reconstruction test consistent with the Kesten--Stigum condition $3|\lambda_2|^2>1$. In addition, we introduce the Markov entropy rate as a computable thermodynamic/information-theoretic observable and derive closed-form expressions for this entropy rate at the symmetric fixed point for an arbitrary spin $s$. Numerical illustrations for $s=1,2,\ldots,5$ are provided to compare the entropy-rate behavior with the spectral criteria. Our results connect naturally to themes in information theory and statistical physics and are also relevant to reconstruction problems on trees that appear in biology/phylogenetics.
\\[4pt]
\\[2pt]

\end{abstract}

\end{frontmatter}
\tableofcontents

\section{Introduction}\label{sec:intro}

On a rooted $d$-ary tree, the broadcasting process sends a root state down the edges through a fixed irreducible, aperiodic channel $M$ \cite{Mossel2001}. The reconstruction problem asks whether depth-$n$ boundary data still distinguish different root states as $n\to\infty$ (e.g., via non-vanishing total variation or mutual information) \cite{Mossel2001}. While $d|\lambda_2(M)|^2>1$ is a classical sufficient condition, it is not always sharp; reconstruction can hold even when $d|\lambda_2(M)|^2<1$ for some channels \cite{Mossel2001}. Mossel further characterizes reconstructability by the existence of a nontrivial tail $\sigma$-field, that is, long-range memory on the tree \cite{Mossel2001}.

In statistical mechanics \cite{Mossel2004}, information theory \cite{MosselPeres2003},  and communication theory \cite{CoverThomas2006}, processes on a Cayley tree (CT) are often used to model how information (if any) propagates from the root to distant generations. In the standard setting, each vertex receives a noisy copy of its parent’s state, with identical transmission statistics on every edge and independent noise across edges \cite{Cavender1978}, and related formulations also appear in statistical physics \cite{Preston1974,Spitzer1975}. The central question is whether the configuration observed at level $n$ retains non-negligible information about the root as $n$ increases. This reconstruction-type problem arises in biology \cite{Muffato2023}, information theory \cite{Mossel2001}, and statistical physics. In our study, when computing the entropy rate, we quantify information transfer from the grandparent to the grandchild by incorporating the intermediate generation \cite{Reza1994}, rather than measuring the flow only between consecutive generations.

The general framework of Gibbs measures and phase transitions on trees and lattices
is discussed in detail in Georgii~\cite{Georgii1988} and Rozikov~\cite{Rozikov2013}, 
while the entropy rate of stationary Markov chains follows the classical
information-theoretic treatment of Cover and Thomas~\cite{CoverThomas2006}.
For related results on Ising models on tree-like graphs and reconstruction
phenomena; see, for example, Dembo and Montanari~\cite{DemboMontanari2010}.

The single-spin Ising model has been extensively investigated since the 1920s,
particularly in the context of statistical mechanics. Its analysis becomes more
tractable on the CT than on the cubic lattice, owing to the hierarchical
structure of the tree. This setting has enabled detailed studies not only of phase
transitions but also of thermodynamic functions; see, for instance,
\cite{Gandolfo-Haydar-2013,Gandolfo-Rahmet-2013}. Subsequent developments introduced
mixed-spin Ising models, in which different spin types coexist within a single
configuration. Such models have been analysed both on cubic lattices and on Bethe
lattices and CTs. In recent years, particular attention has been paid
to $(s,(2t-1)/2)$ mixed-spin Ising models on CTs~\cite{EAP16,EBM18}, where $s$ and
$t$ are positive integers. These studies show that mixed-spin systems exhibit
distinctive properties and potential applications that do not arise in
single-spin variants.

Mixed-spin Ising models have thus attracted considerable interest across several
areas. For example, Stre\v{c}ka~\cite{strecka2006} investigates the mixed
spin-$(\tfrac12,S)$ Ising model on the Union Jack lattice (a centered square lattice)
via mapping to the eight-vertex model. A key result of that work concerns
the weak-universal critical exponents along the bicritical line: these exponents
are enhanced by a factor of two for systems with integer spin-$S$ atoms compared
to systems with half-odd-integer-spin-$S$ atoms. In~\cite{Akin2024}, we constructed
translation-invariant splitting Gibbs measures (TISGMs) for a mixed-spin model on a
CT of order two. It was shown there that increasing the weight of the
$s$-spin value enlarges the basin of repulsion of the fixed point
$\ell_0^{(s)}$, which characterises the TISGM, thereby broadening the phase transition
region. In~\cite{akin2024-CJP}, the cavity method was employed to study the phase
transition behaviour of the $(2,\tfrac12)$ mixed-spin Ising model on a third-order
CT via stability analysis of the associated dynamical system.

In our earlier work~\cite{Akin-Muk-24JSTAT}, we carried out a stability analysis of
the $(1,\tfrac12)$ mixed-spin Ising model on a CT of arbitrary order $k$.
It was shown in~\cite{Akin-Muk-24JSTAT} that, for this model, the parameter region
in which the fixed point
\[
\left(
\left(\frac{\theta^2+1}{2\theta}\right)^k,
\left(\frac{\theta^2+1}{2\theta}\right)^k,
1
\right)
\]
is repelling increases as the tree order $k$ increases. Furthermore, in~\cite{Akin2024}
we determined the phase transition regions for the $(s,\tfrac12)$ mixed-spin Ising
model with arbitrary spin $s$ and $k=2$ by analyzing the corresponding Gibbs measures.

Phase transitions in spin models on the CT depend sensitively on both the interaction regime and underlying spin structure. Single-spin Potts models typically exhibit phase transitions only in the ferromagnetic region \cite{Rozikov2013,Akin-Ulusoy-2023,Akin-Ulusoy-2022}, whereas we recently showed that a modified $q$-state Potts model undergoes a transition exclusively in the antiferromagnetic regime \cite{Akin2025Chaos}. In contrast, mixed-spin Ising models display richer behavior, with phase transitions occurring in both ferromagnetic and antiferromagnetic regimes \cite{Akin2024,Akin-Muk-24JSTAT,Akin2023Chaos,Has-Far-2021}. Despite substantial progress in identifying phase boundaries, a unified understanding that links thermodynamic disorder with the probabilistic structure of these recursive lattice models remains incomplete \cite{Külske2009a}.

Beyond the mere existence of phase transitions, a central question on trees is the
\emph{extremality} of Gibbs measures, in particular for the disordered (free boundary)
phase \cite{ioffe1996}. Extremality is closely related to the decay of correlations, the
information-theoretic reconstruction problem on trees \cite{MartinelliSinclairWeitz2004}, and to the entropy rate of the
underlying the Markov chain representation. In our setting, the disordered phase is
described by a TISGMs measure whose structure is encoded
by a finite-state Markov chain along each ray. The spectral properties of
its transition matrix---and in particular the second largest eigenvalue---simultaneously
govern (i) the extremal/non-extremal nature of the free boundary measure,
(ii) the reconstruction threshold, and (iii) the entropy rate of the induced Markov
process. This provides a natural bridge between thermodynamic quantities such as the
specific entropy, and probabilistic notions such as extremality. 

The question of whether Gibbs measures associated with lattice models on Cayley-type lattices are extremal or non-extremal has been a long-standing subject of study \cite{Mossel2001,MosselPeres2003,KuelskeRozikov2017,Martin2003,MartinelliSinclairWeitz2007,
Moraal1978,RR-2021,Rozikov2018}. In \cite{KuelskeRozikov2015}, Külske and Rozikov examine whether these states, whenever they exist, are extremal or non-extremal within the set of all Gibbs measures corresponding to the solid-on-solid model with spins $0,1,2$ on the CT, as the coupling strength varies. 

Recent investigations into the extremality of disordered phases in mixed-spin Ising \cite{Akin2024,akin2024-CJP,Akin-Muk-24JSTAT,Akin2025Chaos,Akin2023Chaos,Has-Far-2021} and asymmetric Potts models \cite{Ak-Mukh-2025-PS,Qawasmeh-Ak-Mukh-2025-PS} have revealed that, unlike single-spin systems, these models typically require the analysis of multiple stochastic matrices. Consequently, a comprehensive assessment now necessitates both the evaluation of the Kesten-Stigum condition and the calculation of the Dobrushin coefficients. This missing link is naturally framed by the \textbf{Reconstruction Problem}: whether one can recover the root's original spin state from observations taken far out at the leaves. On trees, this question aligns exactly with \textbf{extremality/uniqueness} of the limiting Gibbs measure: reconstruction is possible precisely when the limiting Gibbs measure is non-extremal and impossible when it is extremal.

The analysis of entropy in lattice models has evolved through various formalisms, from thermodynamic derivations on the Bethe lattice \cite{Seino1992} to topological and measure-theoretic explorations of tree structures \cite{Petersen-2020}. While recent studies have established the existence of relative entropy density for generalized Gibbs measures \cite{Kuelske2004} and utilized entropy rates to quantify stochastic complexity \cite{Feutrill2021}, conditional entropies have emerged as a robust tool for measuring disorder, where global mean-field measures prove insufficient \cite{Brandani2013}. In our previous research \cite{Akin2025IJMPB}, we investigated the $q$-state Akin model through a classical thermodynamic lens and calculated the entropy via the temperature derivative of the free energy. However, the mixed-spin $(s, 1/2)$ architecture introduces structural asymmetries that require a more nuanced information-theoretic perspective. Unlike single-spin systems, this model necessitates the evaluation of multiple stochastic matrices to define the states of the system. By integrating the Kesten-Stigum and Dobrushin criteria within an entropy-rate framework, we rigorously quantify the transition between paramagnetic disorder and extremal order. This dual analytical approach allows us to pinpoint the boundaries of information preservation that traditional thermodynamic derivatives may overlook.

Rahmatullaev and Burxonova \cite{Rahmatullaev2025FreeEnergyEntropy} investigated the free energy and entropy of constructive Gibbs measures for the Ising model on a CT tree of order three using the Kolmogorov consistency condition. Their study provided a detailed analytical framework for understanding the thermodynamic properties of a hierarchical lattice structure. While our previous research focused on the thermodynamic entropy derived from the free energy function relative to temperature \cite{Akin2022,Akin2022CJP,Akin23CJP}, utilizing cavity methods, the present study shifts the analytical framework toward an information-theoretic approach (Shannon/relative entropy) for entropy calculation. Our contribution is to close this gap by explicitly connecting the \textbf{Markov entropy rate} \( h_{\mathrm{MC}}^{(s)}(\varphi) \)---a stable, long-range measure of disorder---to the \textbf{Dobrushin condition}, the classical criterion controlling extremality in the disordered regime. Using a \textbf{two-stage Markov chain} representation, we compute \( h_{\mathrm{MC}}^{(s)}(\varphi) \) as a joint function of the spin magnitude \( s \) and thermal parameter \( \varphi \). This yields a single framework where the emergence of non-extremality (equivalently, reconstruction)---the point at which the system retains a 'memory' of boundary information---can be read off quantitatively via the Kesten-Stigum threshold and its temperature dependence. In this context, Külske and Formentin \cite{Külske2009a} studied the reconstruction problem for Markov chains on trees, introducing a convex entropy criterion to determine the boundaries of non-reconstructibility.

Thus, the entropy rate disorder and extremality-based uniqueness become two complementary lenses on the same phenomenon, providing a dual quantitative view of the chaoticity and phase stability of the model. In the present work, we focus on the phase transition region and extremality properties
of the $(s,\tfrac12)$ mixed-spin Ising model on a CT of order $k=3$. Our approach combines a stability analysis of the associated dynamical system with a
spectral and entropy analyses of the corresponding Markov chain representation. In this way, we extend our previous results to a regime where both the spin magnitude $s$ and the branching number $k$ are relatively large, and we clarify how the interplay
between these parameters affects the structure of the phase transitions, the extremality of the disordered phase, and the behavior of the entropy. 

The Kesten--Stigum (KS) criterion does not \emph{guarantee} non-extremality; rather, it provides a threshold that separates the regimes of extremality and possible loss of extremality of the disordered Gibbs measure. In particular, non-extremality becomes relevant when the KS bound is exceeded, whereas below the KS threshold, the disordered phase is expected to remain extremal. In our framework, this transition can be tracked through the joint $(s,\varphi)$ dependence of the Markov entropy rate $h_{\mathrm{MC}}^{(s)}(\varphi)$.

The remainder of this paper is structured as follows. In Section~\ref{sec:TIGMs}, we construct TISGMs for arbitrary $k$ utilizing Kolmogorov's consistency condition. Section~\ref{sec:Dynamical-sys} is devoted to the study of phase transition phenomena for the mixed spin-$\mathbf{(5, \tfrac{1}{2})}$ Ising model on a CT of order three via the stability analysis of an 11-dimensional dynamical system. In Section~\ref{sec:Extremality}, we investigate the disordered phase to determine its extremality and non-extremality in the mixed-spin Ising model using the tree-indexed Markov chain. Specifically, Section \ref{subsec:extremality} identifies the region of the extremal disordered phase using the Dobrushin criterion based on the coefficients of the stochastic matrices $\mathbb{P}$ and $\mathbb{Q}$. Subsection~\ref{subsec:non-extremality} determines the region where the disordered phases are nonextremal using the KS approach. Finally, in Section~\ref{sec:entropy}, we compute the Markov entropy rate and thermodynamic disorder—offering a dual perspective—by employing a two-stage Markov chain model.

\section{Preliminaries}\label{sec-Preliminaries}

In this section, we recall the construction of Gibbs measures (GMs) for the
$(s,\tfrac12)$ mixed spin Ising model (MSIM) on an $n$-shell CT of
arbitrary order, and subject to prescribed boundary field conditions (BFCs).

\subsection{Mixed-spin configurations and finite-volume Gibbs measures}
Consider a semi-infinite Cayley tree $\Gamma^k = (V, E)$ of order $k$ with root $x^0$. For $n \ge 1$, let $V_n$ denote the ball of radius $n$ centered at the root, and let $W_n = V_n \setminus V_{n-1}$ represent the $n$th shell (level) of the tree. We partition the vertex set $V$ based on the parity of the distance from the root, as follows: Let
\begin{align*}
\Gamma^+ &= \{x \in V : d(x^0, x) \text{ is even}\} = \bigcup_{m=0}^{\infty} W_{2m}, \\
\Gamma^- &= \{x \in V : d(x^0, x) \text{ is odd}\} = \bigcup_{m=0}^{\infty} W_{2m+1}.
\end{align*}

The Cayley tree is bipartite such that $V = \Gamma^+ \cup \Gamma^-$ and $\Gamma^+ \cap \Gamma^- = \varnothing$. We assign different spin spaces to these partitions: vertices in $\Gamma^+$ carry spins from the set 
\[
\Phi = \{-s, -(s-1), \dots, s-1, s\},
\]
while vertices in $\Gamma^-$ carry spins from 
\[
\Psi = \left\{-\tfrac{1}{2}, \tfrac{1}{2}\right\}.
\]

A mixed-spin configuration $\xi$ on $V$ is defined as
\begin{align}\label{eq:spin-configuration-def}
\xi(x) = 
\begin{cases} 
\sigma(x) \in \Phi, & \text{if } x \in \Gamma^+, \\
\eta(x) \in \Psi, & \text{if } x \in \Gamma^-.
\end{cases} 
\end{align}

The total configuration space is denoted by $\Xi = \Phi^{\Gamma^+} \times \Psi^{\Gamma^-}$. For a finite volume $V_n$, the restricted configuration spaces are given by $\Omega_n^+ = \Phi^{\Gamma^+ \cap V_n}$ and $\Omega_n^- = \Psi^{\Gamma^- \cap V_n}$.

The Hamiltonian of the model on $V$ is given by
\begin{equation}\label{Hamil1}
H(\xi)
=
-\,J \sum_{\langle x,y \rangle \in E} \xi(x)\,\xi(y),
\end{equation}
with coupling constant $J\in\mathbb{R}$. Its restriction to $V_n$ will be denoted
by $H_n$.

For each $n\ge 1$, the finite-volume Gibbs measure $\mu_n^{\mathbf{h}}$ on
$\Xi^{V_n}$ with boundary field $\mathbf{h}$ supported on $W_n$ is defined by
\begin{equation}\label{Gibbs1}
\mu_n^{\mathbf{h}}(\xi_n)
=
\frac{1}{Z_n(\beta,\mathbf{h})}
\exp\Bigl\{
-\beta H_n(\xi_n)
+
\sum_{x \in W_n} h_{\xi(x)}(x)
\Bigr\},
\end{equation}
where $\xi_n\in\Xi^{V_n}$, $\beta=1/T$ is the inverse temperature, and
$Z_n(\beta,\mathbf{h})$ is the corresponding partition function.

The boundary fields are described by the vector-valued functions
\[
\mathbf{h}(x) = (h_{-s}(x),\dots,h_s(x)), 
\qquad 
\widetilde{\mathbf{h}}(x) = \bigl(\widetilde{h}_{-1/2}(x),\widetilde{h}_{1/2}(x)\bigr),
\]
so that for $x\in W_n$,
\[
h_{\xi(x)}(x) = 
\begin{cases} 
h_{\sigma(x)}(x), & x \in \Omega^+, \\[2pt]
\widetilde{h}_{\eta(x)}(x), & x \in \Omega^-.
\end{cases}
\]
In contrast to the single-spin Ising model, the mixed-spin setting naturally
requires both types of fields, $\mathbf{h}$ and $\widetilde{\mathbf{h}}$, see
e.g.~\cite{Akin2023Chaos,Has-Far-2021}.

\subsection{Compatibility and splitting Gibbs measures}

The family of finite-volume measures $\{\mu_n^{\mathbf{h}}\}_{n\ge 1}$ is said to
be \emph{compatible} (or to satisfy the Kolmogorov consistency condition) if, for
every $n\ge 1$ and every configuration $\xi_{n-1}\in\Xi^{V_{n-1}}$, one has
\begin{equation}\label{compatile1}
\mu_{n-1}^{\mathbf{h}}(\xi_{n-1})
=
\sum_{w \in \Xi^{W_n}}
\mu_n^{\mathbf{h}}(\xi_{n-1} \vee w),
\end{equation}
where $\xi_{n-1}\vee w$ denotes the configuration on $V_n$ that coincides with
$\xi_{n-1}$ on $V_{n-1}$ and with $w$ on $W_n$. Here
\[
\Xi^{W_n} =
\begin{cases}
\Psi^{W_n}, & \text{if } n \text{ is even}, \\[2pt]
\Phi^{W_n}, & \text{if } n \text{ is odd},
\end{cases}
\]
reflecting the tree’s bipartite structure.

Condition~\eqref{compatile1} is precisely the Kolmogorov consistency condition:
the $(n-1)$-volume measure $\mu_{n-1}^{\mathbf{h}}$ is the marginal of the
$n$-volume measure $\mu_n^{\mathbf{h}}$ obtained by summing over all spin
configurations on the outer shell $W_n$. By the Kolmogorov extension theorem,
any such compatible family $\{\mu_n^{\mathbf{h}}\}_{n\ge 1}$ defines a unique
probability measure $\mu$ on the space of infinite configurations
\[
\Omega = \lim_{\longleftarrow} \Xi^{V_n},
\]
such that for all $n\ge 1$ and $\xi_n\in\Xi^{V_n}$,
\[
\mu\bigl(\{\xi\in\Omega : \xi|_{V_n} = \xi_n\}\bigr)
=
\mu_n^{\mathbf{h}}(\xi_n).
\]

A probability measure $\mu$ obtained in this way is called a
\emph{splitting Gibbs measure} (SGM) for the $(s,\tfrac12)$-MSIM on the Cayley
tree. The term “splitting” reflects the recursive factorisation of the measure
along the levels of the tree, which is induced by its hierarchical structure. The main
task in what follows is to identify conditions on the boundary fields
$\mathbf{h}$ and $\widetilde{\mathbf{h}}$ under which the compatibility
condition~\eqref{compatile1} holds, thereby yielding well-defined infinite-volume
Gibbs measure for the model.

\section{Construction of the TISGMs for arbitrary \texorpdfstring{$k$}{k}}\label{sec:TIGMs}

In the setting described above, Kolmogorov's consistency theorem
(see, e.g.,~\cite{Shiryaev-1980}) implies that, for every configuration
$\xi_n \in \Xi^{V_n}$, there exists a unique probability measure $\mu$ on the
infinite-volume configuration space $\Xi^V$ such that
\begin{equation}\label{SGM11}
\mu\bigl(\{\xi \in \Xi^V : \xi|_{V_n} = \xi_n\}\bigr)
=
\mu_n^{\pmb{h}}(\xi_n).
\end{equation}
Thus, any family $\{\mu_n^{\pmb{h}}\}_{n\ge 1}$ satisfying the compatibility
condition~\eqref{compatile1} determines a unique infinite-volume Gibbs measure.
Following our earlier work
\cite{Akin2024,Akin-Muk-24JSTAT,Akin2023Chaos,Has-Far-2021}, we now recall a
theorem which characterises this compatibility in terms of the boundary
fields, thereby providing a constructive description of the Gibbs measures.
The analysis carried out in the present paper is based on the dynamical system
encoded by these functional relationships.

\begin{thm}[{\cite[Theorem 2.1]{Akin2024}}]\label{thm-comp0}
Let $n = 1,2,\ldots$ and $\varphi = e^{\frac{\beta J}{2}}$. The family of
finite-volume Gibbs measures $\{\mu_n^{\pmb{h}}\}$ defined by~\eqref{Gibbs1}
is compatible (i.e., satisfies~\eqref{compatile1}) if and only if, for every
$x \in V$, the following conditions are satisfied:

\textbf{(i)} For $i \in \{-s, -(s-1), \ldots, -2, -1\}$,
\begin{align}\label{consistent1}
\exp\bigl(h_i(x) - h_0(x)\bigr)
=
\prod_{y \in S(x)} 
\left(\frac{\varphi^{2|i|} + \exp\bigl(\widetilde{h}_{\frac12}(y) - \widetilde{h}_{-\frac12}(y)\bigr)}
{\varphi^{|i|}\bigl[1 + \exp\bigl(\widetilde{h}_{\frac12}(y) - \widetilde{h}_{-\frac12}(y)\bigr)\bigr]}\right).
\end{align}

\textbf{(ii)} For $i \in \{1, 2, \ldots, s-1, s\}$,
\begin{align}\label{consistent2}
\exp\bigl(h_i(x) - h_0(x)\bigr)
=
\prod_{y \in S(x)} 
\left(\frac{1 + \varphi^{2|i|}\exp\bigl(\widetilde{h}_{\frac12}(y) - \widetilde{h}_{-\frac12}(y)\bigr)}
{\varphi^{|i|}\bigl[1 + \exp\bigl(\widetilde{h}_{\frac12}(y) - \widetilde{h}_{-\frac12}(y)\bigr)\bigr]}\right).
\end{align}

\textbf{(iii)} For the spin-$\tfrac12$ component,
\begin{align}\label{consistent3}
\exp\bigl(\widetilde{h}_{\frac12}(x) - \widetilde{h}_{-\frac12}(x)\bigr) 
=
\prod_{y \in S(x)} 
\left(
\frac{
\sum_{\sigma(y) = -s}^{s} \varphi^{\sigma(y)} e^{h_{\sigma(y)}(y)}}
{\sum_{\sigma(y) = -s}^{s} \varphi^{-\sigma(y)} e^{h_{\sigma(y)}(y)}}
\right).
\end{align}
\end{thm}

It is well known that
\[
\mathbf{h} = \bigl\{h_x \in \mathbb{R} : x \in V\bigr\}
\quad \text{and} \quad 
\widetilde{\mathbf{h}} = \bigl\{\widetilde{h}_x \in \mathbb{R} : x \in V\bigr\}
\]
represent the (generalized) boundary conditions for the model. For any choice of
boundary fields $(\mathbf{h},\widetilde{\mathbf{h}})$ satisfying the functional
equations~\eqref{consistent1}–\eqref{consistent3}, the corresponding family
$\{\mu_n^{\pmb{h}}\}$ is compatible and therefore determines a unique infinite-volume
Gibbs measure $\mu$. Conversely, every (splitting) Gibbs measure for the
$(s,\tfrac12)$-MSIM on the CT arises from a boundary condition solving
\eqref{consistent1}–\eqref{consistent3}, yielding a one-to-one correspondence
between such boundary fields and Gibbs measures; see, for example,
\cite{Gandolfo-Rahmet-2013}.

\begin{defin}\label{def-TI1}
Let $x\in V$, $i\in\bigl\{-\tfrac{1}{2}, \tfrac{1}{2}\bigr\}$ and
$j\in\{-s,\ldots,-1,0,1,\ldots,s\}$. We say that the boundary fields
are \emph{translation-invariant} if
\[
\widetilde{h}_{i}(x) = \widetilde{h}_{i},
\qquad
h_{j}(x) = h_{j}
\]
for all $x\in V$. In this case the vector-valued boundary fields
\[
\widetilde{\mathbf{h}}
=
\bigl(\widetilde{h}_{-\frac12},\widetilde{h}_{\frac12}\bigr),
\qquad
\mathbf{h}
=
\bigl(h_{-s},\ldots,h_{-1},h_0,h_1,\ldots,h_s\bigr)
\]
are called \emph{translation-invariant}, and the corresponding Gibbs
measures are called \emph{translation-invariant splitting Gibbs measures} (TISGMs).
\end{defin}

In the remainder of this section we analyse the functional equations
\eqref{consistent1}–\eqref{consistent3} under the translation-invariance
assumption in order to determine the existence of TISGMs for the mixed
$(s,\tfrac12)$ Ising model on the CT $\Gamma^k$ (in particular,
for $k=3$).

For $i\in\Phi\setminus\{0\}$ we introduce the (field) differences
\[
U_i(x) = h_i(x) - h_0(x),
\qquad
R(x) = \widetilde{h}_{\frac12}(x) - \widetilde{h}_{-\frac12}(x).
\]
Under translation invariance, these become site-independent:
$U_i(x)\equiv U_i$ and $R(x)\equiv R$. Substituting into
\eqref{consistent1}–\eqref{consistent3} and using the fact that
$|S(x)| = k$ for every vertex $x$, we obtain the following system of equations:

For $i = -s, -s+1, \ldots, -2, -1$, equation~\eqref{consistent1} yields
\begin{align}\label{TI-1}
e^{U_i}
&=\left(
\frac{\varphi^{2|i|} + e^{R}}{\varphi^{|i|}\bigl(1 + e^{R}\bigr)}
\right)^{k}.
\end{align}
For $i = 1, 2, \ldots, s-1, s$, equation~\eqref{consistent2} gives
\begin{align}\label{TI-2}
e^{U_i}
&=\left(
\frac{1 + \varphi^{2i} e^{R}}{\varphi^{i}\bigl(1 + e^{R}\bigr)}
\right)^{k}.
\end{align}
Finally, from equation~\eqref{consistent3} we obtain
\begin{align}\label{TI-3}
e^{R}
&=\left(
\frac{\displaystyle\sum_{i=-s}^{s} \varphi^{s+i} e^{U_i}}
{\displaystyle\sum_{i=-s}^{s} \varphi^{s-i} e^{U_i}}
\right)^{k}.
\end{align}
Equations~\eqref{TI-1}–\eqref{TI-3} therefore describe the TISGMs for the mixed
$(s,\tfrac12)$ Ising model on the CT of order $k$.

\section{Construction and Stability Analysis of the 11-Dimensional Dynamical System for the 
Mixed Spin-\texorpdfstring{$(5,\tfrac12)$}{(5,1/2)} Ising Model on a CT of Order 
\texorpdfstring{$3$}{3}}
\label{sec:Dynamical-sys}

In this section, we carry out a detailed stability analysis of the dynamical system
defined by Equations~\eqref{TI-1}--\eqref{TI-3} in the case $s = 5$. Throughout, we
consider the spin sets
\[
\Phi = \{-5,-4,-3,-2,-1,0,1,2,3,4,5\}, 
\qquad
\Psi = \Bigl\{-\tfrac{1}{2},\,\tfrac{1}{2}\Bigr\},
\]
corresponding to the mixed spin-$(5,\tfrac12)$ configuration on the CT.
The dynamical system arises from the consistency equations for TISGMs and encodes the evolution of the effective boundary fields
across the successive shells of the tree. Our main objective is to determine, within
the relevant parameter space, those regions where the distinguished fixed point
describing the disordered phase changes its character: in particular, we identify
how and when this fixed point passes from being locally stable (attractive) to
repelling as the order $k$ of the CT increases, thereby signalling the
onset and evolution of the phase transitions in the model.

By introducing the change of variables \( Z = e^{R(y)} \) and \(
X_i = e^{U_i(y)} \), the system of equations
\eqref{TI-1}--\eqref{TI-3} can be rewritten as an 11-dimensional
system:

\begin{align}
f_i(X_{-5}, \ldots, X_{-1}, X_1, \ldots, X_5, Z) &:= \left(
\frac{\varphi^{2|i|} + Z}{\varphi^{|i|}(1 + Z)} \right)^3=X_i,
&& \text{for } -5 \leq i \leq -1, \label{dynamicals3aa1} \\
f_i(X_{-5}, \ldots, X_{-1}, X_1, \ldots, X_5, Z) &:= \left( \frac{1
+ \varphi^{2i} Z}{\varphi^{i}(1 + Z)} \right)^3=X_i,
&& \text{for } 1 \leq i \leq 5, \label{dynamicals3aa2} \\
f_6(X_{-5}, \ldots, X_{-1}, X_1, \ldots, X_5, Z) &:= \left(
\frac{\sum_{i=-5}^{5} \varphi^{5 + i} X_i}{\sum_{i=-5}^{5} \varphi^{5 -
i} X_i} \right)^3=Z:=F(Z), && \text{where } X_0 = 1.
\label{dynamicals3aa3}
\end{align}
To find an obvious fixed point, the following lemma can be applied:

\begin{lem}\label{eq:IP-Fixed-Point}
A particular solution of the system \eqref{dynamicals3aa1}--\eqref{dynamicals3aa3} is
\begin{align}\label{eq:FP-initial}
\ell_0^{(5)}:\quad
Z^{(0)}=1,\quad
X_i^{(0)}=\left(\frac{\varphi^{2|i|}+1}{2\varphi^{|i|}}\right)^3
=\left(\cosh\!\left(\frac{\beta J|i|}{2}\right)\right)^3,
\ \ 1\le |i|\le 5.
\end{align}
\end{lem}

\begin{proof}
\begin{itemize}
\item[\textbf{(1)}] \textbf{Reduction to a one-variable fixed-point equation.}
Substituting \eqref{dynamicals3aa1} and \eqref{dynamicals3aa2} into \eqref{dynamicals3aa3}
(for $-5\le s\le 5$) yields the scalar fixed-point equation
\begin{equation}\label{eq:FZ=Z}
F(Z)=\left(\frac{N(Z)}{D(Z)}\right)^3=Z,
\end{equation}
where
\[
N(Z):=A(\varphi)+33\varphi^{20}Z+3A(\varphi)Z^2+C(\varphi)Z^3,\qquad
D(Z):=C(\varphi)+3A(\varphi)Z+33\varphi^{20}Z^2+A(\varphi)Z^3,
\]
and the auxiliary coefficients are defined as
\begin{align*}
A(\varphi) &:= \varphi^{10}\big(1 + \varphi^{2} + \varphi^{4} + \varphi^{6} + \varphi^{8} + \varphi^{10} 
+ \varphi^{12} + \varphi^{14} + \varphi^{16} + \varphi^{18} + \varphi^{20}\big), \\[4pt]
C(\varphi) &:= 1 + \varphi^{4} + \varphi^{8} + \varphi^{12} + \varphi^{16} + \varphi^{20} 
+ \varphi^{24} + \varphi^{28} + \varphi^{32} + \varphi^{36} + \varphi^{40}.
\end{align*}

\item[\textbf{(2)}] \textbf{Verification that $Z=1$ is a fixed point.}
Evaluating $N$ and $D$ at $Z=1$ gives
\[
N(1)=A(\varphi)+33\varphi^{20}+3A(\varphi)+C(\varphi)=4A(\varphi)+C(\varphi)+33\varphi^{20},
\]
\[
D(1)=C(\varphi)+3A(\varphi)+33\varphi^{20}+A(\varphi)=4A(\varphi)+C(\varphi)+33\varphi^{20}.
\]
Hence $N(1)=D(1)$, so $F(1)=\bigl(N(1)/D(1)\bigr)^3=1$. Therefore $Z^{(0)}=1$ satisfies
\eqref{eq:FZ=Z}.

\item[\textbf{(3)}] \textbf{Closed form for $X_i^{(0)}$ once $Z^{(0)}=1$.}
Substituting $Z^{(0)}=1$ back into \eqref{dynamicals3aa1}--\eqref{dynamicals3aa2} (for
$1\le |i|\le 5$) yields
\begin{align}\label{IP-FP-1a}
X_i^{(0)}=\left(\frac{\varphi^{|i|}+\varphi^{-|i|}}{2}\right)^3
=\left(\frac{\varphi^{2|i|}+1}{2\varphi^{|i|}}\right)^3.  
\end{align}

Finally, if $\varphi=e^{\beta J/2}$ then
\[
\frac{\varphi^{|i|}+\varphi^{-|i|}}{2}
=\cosh\!\left(\frac{\beta J|i|}{2}\right),
\]
which gives the hyperbolic cosine form in \eqref{eq:FP-initial}.
\end{itemize}
This proves that $\ell_0^{(5)}$ is the particular solution.
\end{proof}

The Split Gibbs Measure (SGM) is said to be in the disordered phase (DP) when $Z=1$ satisfies equation \eqref{dynamicals3aa3}. This condition is equivalent to the system being insensitive to the boundary conditions \cite{Mukhamedov-2022}. In our mixed-spin $(5, 1/2)$ Ising model, the TISGM associated with the fixed point $\ell_0^{(5)}$ represents this phase transition. Hence, we denote this DP by $\mu_0^{(5)}$, where the superscript refers to the spin value $s=5$.The evolution of the system is governed by the 11-dimensional dynamical mapping defined in \eqref{dynamicals3aa1}--\eqref{dynamicals3aa3}. We determine the local stability of any fixed point $(X_{-5}^*, \dots, X_5^*, Z^*)$ using linear stability analysis. Specifically, the fixed point is stable if all eigenvalues of the Jacobian matrix evaluated at that point have a magnitude that is strictly less than unity ($|\lambda_i| < 1$).

We consider the following dynamical system:
\[
\begin{cases}
X_i^{(n+1)} = f_i\!\left(X_{-5}^{(n)}, \dots, X_5^{(n)}, Z^{(n)}\right), & i = -5, \dots, -1, 1, \dots, 5, \\[4pt]
Z^{(n+1)} = f_6\!\left(X_{-5}^{(n)}, \dots, X_5^{(n)}, Z^{(n)}\right),
\end{cases}
\]
where the functions $f_i$ govern the update rules for variables.

\subsection{Stability Analysis via the Jacobian Matrix}\label{subsec:SA-a1} 
Here, we investigate the stability of the fixed point $\ell_0^{(5)}$ (defined in \eqref{eq:FP-initial}) for the dynamical system given by \eqref{dynamicals3aa1}--\eqref{dynamicals3aa3}. To this end, we derive the Jacobian matrix associated with the system and evaluate it at $\ell_0^{(5)}$. By examining the eigenvalues of the resulting Jacobian, we identify the parameter regimes in which the fixed point changes stability, thereby determining the conditions for a phase transition in the $(5,\tfrac{1}{2})$-mixed spin Ising model.

Now let the vector
\begin{align*}
\mathbf{X}&= \left(X_{-5}, X_{-4}, X_{-3}, X_{-2}, X_{-1}, X_{1}, X_{2}, X_{3}, X_{4}, X_{5}, Z\right)  
\end{align*}
be a fixed point of the equations \eqref{dynamicals3aa1}, \eqref{dynamicals3aa2}, and \eqref{dynamicals3aa3}.
The Jacobian matrix $J(\mathbf{X})$ at the fixed point $\mathbf{X}$ is defined 
\begin{align}\label{matrix:JAC_0}
J(\mathbf{X}) =
\begin{bmatrix}
\frac{\partial f_{-5}}{\partial X_{-5}} & \cdots & \frac{\partial f_{-5}}{\partial X_{5}} & \frac{\partial f_{-5}}{\partial Z} \\
\vdots & \ddots & \vdots & \vdots \\
\frac{\partial f_{5}}{\partial X_{-5}} & \cdots & \frac{\partial f_{5}}{\partial X_{5}} & \frac{\partial f_{5}}{\partial Z} \\
\frac{\partial f_{6}}{\partial X_{-5}} & \cdots & \frac{\partial f_{6}}{\partial X_{5}} & \frac{\partial f_{6}}{\partial Z}
\end{bmatrix},  
\end{align}
where $f_{-5}, f_{-4}, \dots, f_{5}, f_{6}$ denote the right-hand sides of the dynamical system
\eqref{dynamicals3aa1}--\eqref{dynamicals3aa3}.

The expressions for derivatives are:\\
For \( -5 \leq i \leq -1 \):
\begin{align}\label{eq:PDE-1}
\frac{\partial f_i}{\partial Z} = - \frac{3 \left( \varphi^{2|i|}-1
\right)\left( Z + \varphi^{2|i|} \right)^2}{\varphi^{3|i|}(1 + Z)^4},
\quad \frac{\partial f_i}{\partial X_i}=0.
\end{align}

For \( 1 \leq i \leq 5 \):
\begin{align}\label{eq:PDE-2}
\frac{\partial f_i}{\partial Z} = \frac{3 \left(\varphi
^{2i}-1\right)\left(1+Z \varphi ^{2 i}\right)^2}{\varphi ^{3 i}
(1+Z)^4}, \quad \frac{\partial f_i}{\partial X_i}=0.
\end{align}

For \( f_6 \) and \(i\neq 0\), let
\[
N = \sum_{i=-5}^5 \varphi^{5+i} X_i, \quad D = \sum_{i=-5}^5
\varphi^{5-i} X_i,
\]
then
\begin{align}\nonumber
\frac{\partial f_6}{\partial X_j}&= \frac{\partial}{\partial X_j} \left[ \left( \frac{\sum\limits_{k=-5}^{5} \varphi^{5+k} X_k }{ \sum\limits_{k=-5}^{5}
\varphi^{5-k} X_k }\right)^3 \right] = 3 \left( \frac{N}{D} \right)^2 \cdot \frac{\varphi^{5+j} D - \varphi^{5-j} N}{D^2} \\\label{eq:PDE-3aa}
&= 3 \cdot \frac{ \left( \sum\limits_{k=-5}^{5} \varphi^{5+k} X_k\right)^2 } { \left( \sum\limits_{k=-5}^{5} \varphi^{5-k} X_k\right)^4 } \cdot \left( \varphi^{5+j} \sum\limits_{k=-5}^{5}\varphi^{5-k} X_k - \varphi^{5-j} \sum\limits_{k=-5}^{5} \varphi^{5+k} X_k
\right).
\end{align}

Since the function $f_6$ is independent of $Z$, one can easily see that
\[
\frac{\partial f_6}{\partial Z} = 0.
\]

Let the initial state vector from Lemma~\ref{eq:IP-Fixed-Point} be
\[
\mathbf{X}^{(0)}=\bigl(X_{-5}^{(0)},X_{-4}^{(0)},X_{-3}^{(0)},X_{-2}^{(0)},X_{-1}^{(0)},
X_{1}^{(0)},X_{2}^{(0)},X_{3}^{(0)},X_{4}^{(0)},X_{5}^{(0)},Z^{(0)}\bigr).
\]

In our earlier work~\cite{Akin2024}, we proposed a method for explicitly determining the eigenvalues of the Jacobian matrix associated with a nonlinear dynamical system of dimension $(2s+1)$. For completeness, we restate the mathematical results below.

\begin{lem}\cite[Lemma B.1]{Akin2024}\label{lemma:Eigenvalues-S}
Consider a $(2s+1) \times (2s+1)$ matrix $J$ with entries
\begin{align}\label{eq:matrix-Stability-S}
J_{ij} = 
\begin{cases}
-a_{s+1-i}, & 1 \leq i \leq s, \quad j = 2s+1; \\
a_{i-s}, & s+1 \leq i \leq 2s, \quad j = 2s+1; \\
0, & 1 \leq i \leq 2s, \quad 1 \leq j \leq 2s; \\
-b_{s+1-j}, & i = 2s+1, \quad 1 \leq j \leq s; \\
b_{j-s}, & i = 2s+1, \quad s+1 \leq j \leq 2s; \\
0, & i = 2s+1, \quad j = 2s+1.
\end{cases}
\end{align}
where $\{a_k\}$ and $\{b_k\}$ are the given sequences.
Then the eigenvalues of $J$ in \eqref{eq:matrix-Stability-S} are
\[
\left\{0,0,\ldots,0,\,-\sqrt{2\sum_{i=1}^s a_i b_i},\,\sqrt{2\sum_{i=1}^s a_i b_i}\right\}.
\]
\end{lem}

From \eqref{eq:PDE-1}, \eqref{eq:PDE-2}, and \eqref{eq:PDE-3aa}, with the variables ordered as
\[
(X_{-5},X_{-4},X_{-3},X_{-2},X_{-1},X_{1},X_{2},X_{3},X_{4},X_{5},Z),
\]
the Jacobian matrix at \(\mathbf{X}^{(0)}\) takes the block form
\[
J\bigl(\mathbf{X}^{(0)}\bigr)
=
\begin{pmatrix}
\mathbf{0}_{10\times 10} & \mathbf{g}\\[4pt]
\mathbf{h}^{\!\top} & 0
\end{pmatrix},
\qquad
\mathbf{g}
=
\begin{pmatrix}
\displaystyle \frac{\partial f_{-5}}{\partial Z}\\[6pt]
\displaystyle \frac{\partial f_{-4}}{\partial Z}\\
\vdots\\[2pt]
\displaystyle \frac{\partial f_{5}}{\partial Z}
\end{pmatrix}_{\!\big|\;\mathbf{X}^{(0)}},
\quad
\mathbf{h}
=
\begin{pmatrix}
\displaystyle \frac{\partial f_{6}}{\partial X_{-5}}\\[6pt]
\displaystyle \frac{\partial f_{6}}{\partial X_{-4}}\\
\vdots\\[2pt]
\displaystyle \frac{\partial f_{6}}{\partial X_{5}}
\end{pmatrix}_{\!\big|\;\mathbf{X}^{(0)}}.
\]

Equivalently, in entry-wise form,
\[
J\bigl(\mathbf{X}^{(0)}\bigr)=
\begin{pmatrix}
0 & \cdots & 0 & \dfrac{\partial f_{-5}}{\partial Z}\\
\vdots & \ddots & \vdots & \vdots\\
0 & \cdots & 0 & \dfrac{\partial f_{5}}{\partial Z}\\
\dfrac{\partial f_{6}}{\partial X_{-5}} & \cdots & \dfrac{\partial f_{6}}{\partial X_{5}} & 0
\end{pmatrix}_{\big|\;\mathbf{X}^{(0)}}.
\]
Since the \(10\times 10\) block is zero, the characteristic polynomial is
\[
\chi_{J}(\lambda)=\lambda^{9}\bigl(\lambda^{2}-\mathbf{h}^{\!\top}\mathbf{g}\bigr),
\]
so the nonzero eigenvalues are
\(
\lambda_{\pm}=\pm\sqrt{\mathbf{h}^{\!\top}\mathbf{g}}\,
\)
(evaluated at \(\mathbf{X}^{(0)}\)) (for \cite{Akin2024,Akin-Muk-24JSTAT,Akin2023Chaos} for details).

For the discrete-time map, the fixed point $\mathbf{X}^{(0)}$ is linearly stable if and only if
\(
|\lambda_{\pm}|<1
\),
which is equivalent to
\(
\mathbf{h}^{\!\top}\mathbf{g}<1
\).
In this work, however, our main interest lies in the regime where
\(
|\lambda_{\pm}|\geq 1
\).
Therefore, we concentrate on the parameter domain in which the fixed point becomes repelling.
As noted in our earlier studies \cite{Akin2024,Akin-Muk-24JSTAT,Akin2023Chaos}, the condition
\(
|\lambda_{\pm}|\geq 1
\)
is associated with the emergence of multiple fixed points for the system in \eqref{dynamicals3aa1}--\eqref{dynamicals3aa3}.
Hence, the appearance of more than one fixed point indicates that the model exhibits a phase transition.

Now let us give the following Lemma.

\begin{lem}\label{eq:JM}
The entries of the Jacobian matrix
\(
J(\mathbf{X}^{(0)})= \bigl(J_{ij}\bigr)_{11 \times 11}
\)
at the fixed point \(\ell_0\) given in \eqref{eq:FP-initial} for the system of
Eqs. \eqref{dynamicals3aa1}--\eqref{dynamicals3aa3} are as follows:

\[
J_{ij} =
\begin{cases}
-\frac{3(\varphi^{4|i|} - 1) (1 + \varphi^{2|i|})}{16\varphi^{3|i|} }, & -5 \leq i \leq -1, \quad j = 11, \\[10pt]
\frac{3(\varphi^{4i} - 1) (1 + \varphi^{2i})}{16\varphi^{3i} }, & 1 \leq i \leq 5, \quad j = 11, \\[10pt]
0, & 1 \leq i \leq 10, \quad 1 \leq j \leq 10, \\[10pt]
\left.\dfrac{\partial f_6(X_{-5}, \ldots, X_{-1}, X_1, \ldots, X_5)}{\partial X_{j - 6}} \right|_{\ell_0}, & i = 11, \quad 1 \leq j \leq 5, \\[10pt]
\left.\dfrac{\partial f_6(X_{-5}, \ldots, X_{-1}, X_1, \ldots, X_5)}{\partial X_{j - 5}} \right|_{\ell_0}, & i = 11, \quad 6 \leq j \leq 10, \\[10pt]
0, & i = 11, \quad j = 11.
\end{cases}
\]  
\end{lem}

Because the proof of the lemma is based on elementary calculus, we will not provide it again here. From \eqref{matrix:JAC_0}, based on Lemma \ref{eq:JM}, the Jacobian matrix for the dynamical system described in \eqref{dynamicals3aa1}--\eqref{dynamicals3aa3} is computed as follows:
\begin{align}\label{eq:JM-1a}
J&=\left(
\begin{array}{ccccccccccc}
 0 & 0 & 0 & 0 & 0 & 0 & 0 & 0 & 0 & 0 & -\frac{3(\varphi^{20}-1)(1+\varphi^{10})}{16\varphi^{15}}\\
 0 & 0 & 0 & 0 & 0 & 0 & 0 & 0 & 0 & 0 & -\frac{3(\varphi ^{16}-1)(1+\varphi ^8)}{16\varphi ^{12}} \\
 0 & 0 & 0 & 0 & 0 & 0 & 0 & 0 & 0 & 0 & -\frac{3(\varphi ^{12}-1)(1+\varphi ^6)}{16\varphi ^9} \\
 0 & 0 & 0 & 0 & 0 & 0 & 0 & 0 & 0 & 0 & -\frac{3(\varphi ^8-1)(1+\varphi ^4)}{16\varphi ^6} \\
 0 & 0 & 0 & 0 & 0 & 0 & 0 & 0 & 0 & 0 & -\frac{3(\varphi ^4-1)(1+\varphi ^2)}{16\varphi ^3} \\
 0 & 0 & 0 & 0 & 0 & 0 & 0 & 0 & 0 & 0 & \frac{3(\varphi ^4-1)(1+\varphi ^2)}{16\varphi ^3} \\
 0 & 0 & 0 & 0 & 0 & 0 & 0 & 0 & 0 & 0 & \frac{3(\varphi ^8-1)(1+\varphi ^4)}{16\varphi ^6} \\
 0 & 0 & 0 & 0 & 0 & 0 & 0 & 0 & 0 & 0 & \frac{3(\varphi ^{12}-1)(1+\varphi ^6)}{16\varphi ^9} \\
 0 & 0 & 0 & 0 & 0 & 0 & 0 & 0 & 0 & 0 & \frac{3(\varphi ^{16}-1)(1+\varphi ^8)}{16\varphi ^{12}} \\
 0 & 0 & 0 & 0 & 0 & 0 & 0 & 0 & 0 & 0 & \frac{3(\varphi ^{20}-1)(1+\varphi ^{10})}{16\varphi ^{15}} \\
 -A & -B & -C & -F & -G & G & F & C & B & A & 0
\end{array}
\right),   
\end{align}
where \[
\begin{aligned}
A &= \left.\frac{\partial f_6}{\partial X_{-5}}\right|_{\ell_0}=\frac{24 \varphi ^{15} \left(\varphi ^{10}-1\right)}{K(\varphi)}, &
B &= \left.\frac{\partial f_6}{\partial X_{-4}}\right|_{\ell_0}=\frac{24 \varphi ^{16} \left(\varphi ^8-1\right)}{K(\varphi)}, \\
C &= \left.\frac{\partial f_6}{\partial X_{-3}}\right|_{\ell_0}=\frac{24 \varphi ^{17} \left(\varphi ^6-1\right)}{K(\varphi)}, &
F &= \left.\frac{\partial f_6}{\partial X_{-2}}\right|_{\ell_0}=\frac{24 \varphi ^{18} \left(\varphi ^4-1\right)}{K(\varphi)}, \\
G &= \left.\frac{\partial f_6}{\partial X_{-1}}\right|_{\ell_0}=\frac{24\varphi ^{19} \left(\varphi ^2-1\right)}{K(\varphi)}. &
\end{aligned}
\]
\[K(\varphi)=1+\varphi ^4+\varphi ^8+4 \varphi ^{10}+5 \varphi ^{12}+4 \varphi ^{14}+5 \varphi ^{16}+4 \varphi ^{18}+38 \varphi ^{20}+4 \varphi ^{22}+5 \varphi ^{24}+4 \varphi ^{26}+5 \varphi ^{28}+4 \varphi ^{30}+\varphi ^{32}+\varphi ^{36}+\varphi ^{40}.
\]

Based on Lemma \ref{lemma:Eigenvalues-S}, the eigenvalues \(\lambda_i\) of matrix \(J\) \eqref{eq:JM-1a} are given by:
\[
\lambda_i =
\begin{cases}
0, & \text{(multiplicity 9)} \\[4pt]
\pm\sqrt{s(\varphi)}, & \text{(each simple)}
\end{cases},
\]
where the function \(s(\varphi)\) is defined as
\begin{align*}
s(\varphi) &= \frac{9 \left(1 + \varphi^4 + \varphi^8 + \varphi^{12} + \varphi^{16} - 10\varphi^{20} + \varphi^{24} + \varphi^{28} + \varphi^{32} + \varphi^{36} + \varphi^{40}\right)}{K(\varphi)}.
\end{align*}
This formulation reveals that the largest eigenvalue is \(\lambda_{\max}(\varphi) = \sqrt{\frac{9P(\varphi)}{K(\varphi)}}\). To characterize the phase transition, we identify the region where the fixed point \(\ell_0^{(5)}\) becomes repelled. This requires the condition \(|\lambda_{\max}(\varphi)| > 1\), leading us to solve the equation \(\sqrt{\frac{9P(\varphi)}{K(\varphi)}} - 1 = 0\). Within the domain where the square root is defined, this equation is equivalent to
\[
\frac{9P(\varphi)}{K(\varphi)} = 1 \quad \Longleftrightarrow \quad 9P(\varphi) - K(\varphi) = 0.
\]
Solving the resulting polynomial equation \(9P(\varphi) - K(\varphi) = 0\) for positive real roots using symbolic computation (Mathematica) yields the following critical values:
\[
\varphi \approx 0.901258081777163, \qquad \varphi \approx 1.10956009185308.
\]
Thus, we obtain the following results.

\begin{thm}
\label{thm:repelling-l0}
For $\varphi \in (0,\, 0.901258081777163) \,\cup\, (1.10956009185308,\, \infty)$, the fixed point $\ell_0^{(5)}$ is repelling. 
At the critical values $\varphi \approx 0.901258081777163$ and $\varphi \approx 1.10956009185308$, one has
$|\lambda_{\max}(\varphi)|=1$, so $\ell_0^{(5)}$ is non-hyperbolic (neutral).
\end{thm}

\begin{rem}
Our previous studies~\cite{Akin2024,akin2024-CJP,Akin-Muk-24JSTAT,Akin2023Chaos} established that in $d$-dimensional dynamical systems corresponding to lattice models on a CT, the existence of a repelling fixed point implies the presence of additional fixed points. Consequently, the detection of a repelling fixed point serves as an indicator of phase transition in the model.
\end{rem}

\subsection{Stability of the one-dimensional dynamical system}

To assess the local stability of a fixed point in the one-dimensional system given in \eqref{dynamicals3aa3}, we use the linearization criterion \(|F'(Z^*)|\), where \(Z^*\) denotes the fixed point. This scalar derivative test replaces the more general Jacobian-based analysis that is applicable to higher-dimensional systems.

The relevant dynamical map is given by
\begin{align}\label{eq:1D-Stability1a}
  F(Z):=
\frac{\bigl(A(\varphi)+33\varphi^{20}Z+3A(\varphi)Z^{2}+C(\varphi)Z^{3}\bigr)^{3}}
     {\bigl(C(\varphi)+3A(\varphi)Z+33\varphi^{20}Z^{2}+A(\varphi)Z^{3}\bigr)^{3}}.  
\end{align}

A direct computation shows that \(F(1)=1\), so \(Z^*=1\) is indeed the fixed point.

Applying the quotient rule yields the derivative at this fixed point as follows:
\[
F'(1)=\frac{9\bigl(C(\varphi)-11\varphi^{20}\bigr)}{4A(\varphi)+C(\varphi)+33\varphi^{20}}.
\]
Here the auxiliary coefficients are
\begin{align*}
A(\varphi)&:=\varphi^{10}\bigl(1+\varphi^{2}+\varphi^{4}+\varphi^{6}+\varphi^{8}+\varphi^{10}
          +\varphi^{12}+\varphi^{14}+\varphi^{16}+\varphi^{18}+\varphi^{20}\bigr),\\[4pt]
C(\varphi)&:=1+\varphi^{4}+\varphi^{8}+\varphi^{12}+\varphi^{16}+\varphi^{20}
          +\varphi^{24}+\varphi^{28}+\varphi^{32}+\varphi^{36}+\varphi^{40}.
\end{align*}
The derivative admits the factorisation
\[
F'(1)=\frac{9(\varphi-1)^{2}(\varphi+1)^{2}(\varphi^{2}+1)^{2}\,Q(\varphi)}{P(\varphi)},
\]
with
\[
\begin{aligned}
Q(\varphi)&=\varphi^{32}+3\varphi^{28}+6\varphi^{24}+10\varphi^{20}+15\varphi^{16}
         +10\varphi^{12}+6\varphi^{8}+3\varphi^{4}+1,\\
P(\varphi)&=\varphi^{40}+\varphi^{36}+\varphi^{32}+4\varphi^{30}+5\varphi^{28}+4\varphi^{26}
         +5\varphi^{24}+4\varphi^{22}+38\varphi^{20}+4\varphi^{18}+5\varphi^{16}\\
       &\quad+4\varphi^{14}+5\varphi^{12}+4\varphi^{10}+\varphi^{8}+\varphi^{4}+1.
\end{aligned}
\]
Both polynomials \(Q(\varphi)\) and \(P(\varphi)\) are strictly positive for every real \(\varphi\). Consequently,
\[
F'(1)\geq 0\quad\text{for all }\varphi,\qquad
F'(1)=0\iff\varphi=\pm1.
\]

For a one-dimensional map, the fixed point \(Z=1\) is
\begin{itemize}
    \item \textbf{repelling} when \(|F'(1)|>1\),
    \item \textbf{attracting} when \(|F'(1)|<1\),
    \item \textbf{neutral} when \(|F'(1)|=1\).
\end{itemize}
Since \(F'(1)\ge 0\), the repelling condition reduces to \(F'(1)>1\), i.e.
\[
\frac{9\bigl(C(\varphi)-11\varphi^{20}\bigr)}{4A(\varphi)+C(\varphi)+33\varphi^{20}}>1
\quad\Longleftrightarrow\quad
A(\varphi)-2C(\varphi)+33\varphi^{20}<0.
\]
\begin{figure}[!htbp]
\centering
\includegraphics[width=65mm]{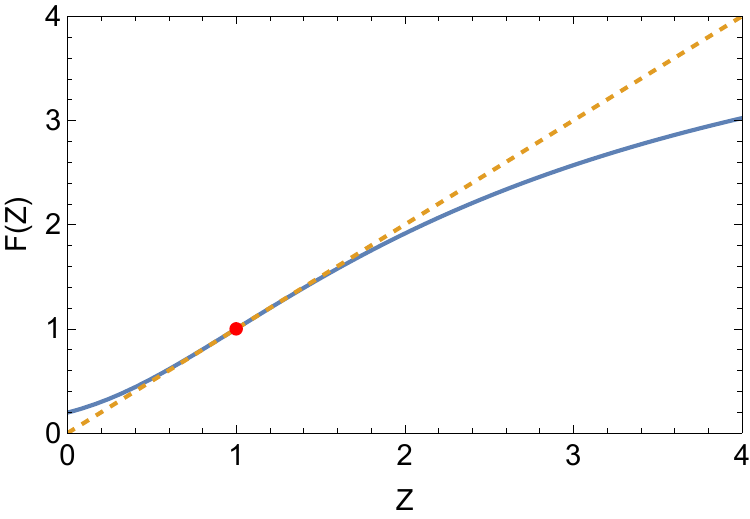}
\hspace{0.5cm}
\includegraphics[width=65mm]{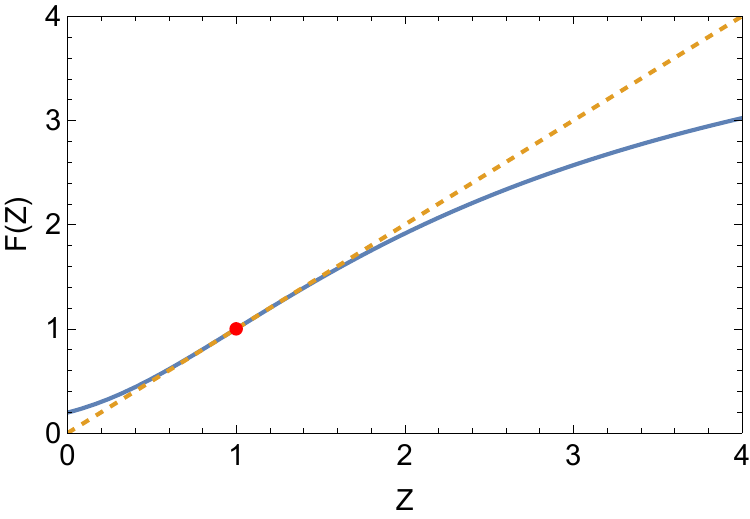}
\caption{(Color online) Plots of the one-dimensional map $F(Z)$ (solid curve) and the diagonal $y = Z$ (dashed line) at the two critical parameter values: 
$\varphi = 0.901258081777163$ (left) and $\varphi = 1.10956009185308$ (right). 
At these values, $|F'(1)| = 1$, indicating a change in the stability of the fixed point $Z = 1$. 
For $0 < \varphi < 0.901258081777163$ and $\varphi > 1.10956009185308$, the fixed point is repelling; 
for $0.901258081777163 < \varphi < 1.10956009185308$, it is attracting.}
\label{F(Z)-Stability-1}
\end{figure}

\paragraph{Consistency with the Jacobian analysis}
From the explicit formula for \(F'(1)\) one verifies that \(F'(1)\) is continuous
for \(\varphi>0\), with the equation \(F'(1)=1\) having exactly two positive solutions
\[
\varphi_1 \approx 0.901258081777163,
\qquad
\varphi_2 \approx 1.10956009185308.
\]

Moreover,
\[
F'(1)>1 \quad\text{for}\quad 0<\varphi<\varphi_1 \ \text{and}\ \varphi>\varphi_2,
\qquad
F'(1)<1 \quad\text{for}\quad \varphi_1<\varphi<\varphi_2.
\]
Since \(Z=1\) is repelling if and only if \(|F'(1)|>1\), this yields the repelling
parameter set
\[
\varphi \in (0,\varphi_1)\,\cup\,(\varphi_2,\infty).
\]
This coincides exactly with the repelling intervals stated in
Theorem~\ref{thm:repelling-l0} for the equilibrium \(\ell_0^{(5)}\) of the full system.
Thus the one-dimensional reduction and the Jacobian eigenvalue analysis provide
consistent descriptions of local stability of \(\ell_0^{(5)}\).
\subsection{Numerical results}

The qualitative behavior of the map $F$ at the two critical values $\varphi = \varphi_1$ and $\varphi = \varphi_2$ is illustrated in Figure~\ref{F(Z)-Stability-1}. The figure shows $F(Z)$ together with the diagonal $y = Z$.

The phase transition occurs precisely within the parameter intervals where the fixed point is repelled. This is demonstrated numerically in Fig. ~\ref{Stability-1}.
\begin{figure}[!htbp]
\centering
\includegraphics[width=65mm]{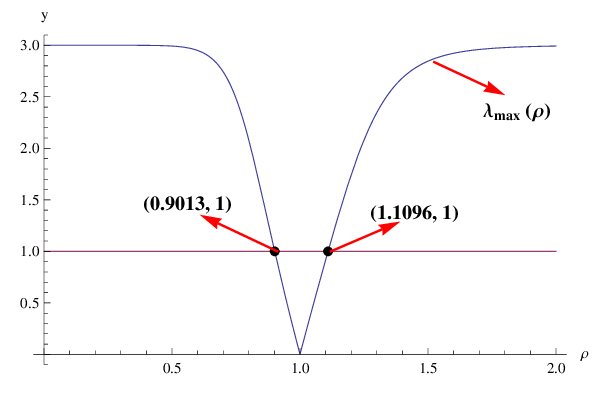}
\caption{(Color online) Intersection points between $\lambda_{\max}(\varphi)$ and the line $y = 1$. In the region where $\lambda_{\max}(\varphi) > 1$, the fixed point $Z = 1$ is repelling, corresponding to the phase transition region of the model.}
\label{Stability-1}
\end{figure}
\begin{figure}[!htbp]
\centering
\includegraphics[width=65mm]{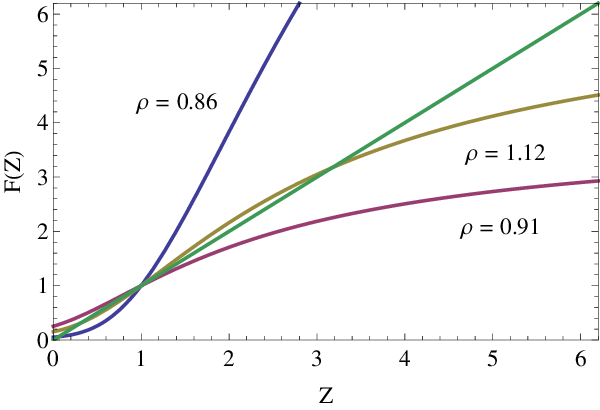}
\caption{(Color online) Phase diagram showing the stable fixed points of the one-dimensional rational map defined in \eqref{eq:1D-Stability1a} as a function of $\varphi$. The repelling regime of the $Z = 1$ fixed point (shaded/darker region) aligns with the phase transition boundaries identified in the previous figures.}
\label{fig:Stability1}
\end{figure}
The graph in Figure~\ref{fig:Stability1} illustrates the fixed points of the rational function defined in~\eqref{eq:1D-Stability1a}. 
As shown, three distinct fixed points exist for $\varphi = 0.86$ and $\varphi = 1.12$, whereas only a single fixed point exists for $\varphi = 0.91$. 
Hence, Theorem~\ref{thm:repelling-l0} confirms the presence of a phase transition in this specified region. 
Consistent with the results reported in our previous studies \cite{Akin2024,Akin-Muk-24JSTAT,Akin2023Chaos}, 
the phase transition occurs in both the ferromagnetic ($J>0$) and antiferromagnetic ($J<0$) regimes, 
The phase-transition region widens accordingly.

\subsection{Stability analysis of other fixed points}

For certain values of $\varphi$, the system \eqref{dynamicals3aa1}–\eqref{dynamicals3aa3} 
admits nontrivial fixed points that are distinct from $\ell_0^{(5)}$. In particular, for 
$\varphi \approx 1.12434$ we find an approximate fixed point 
$\mathbf{X} = (X_{-5},X_{-4},\dots,X_{-1},X_{1},\dots,X_{5},Z)^{\mathsf T} \in \mathbb{R}^{11}$ 
whose components satisfy
\begin{align}
X_i &= \left(\frac{\varphi^{-i} + 4\varphi^{i}}{5}\right)^{3}, 
\qquad &&\text{for } i=\pm1,\pm2,\pm3,\pm4,\pm5, \label{Zi_eq1} \\[4pt]
Z\; &= \left(\frac{\displaystyle\sum_{i=-5}^{5} \varphi^{5+i} X_i}
{\displaystyle\sum_{i=-5}^{5} \varphi^{5-i} X_i}\right)^{3},
\qquad &&\text{with } X_0=1. \label{Zi_eq2}
\end{align}
At $\varphi \approx 1.12434$, the expression in \eqref{Zi_eq2} evaluates numerically to 
$Z \approx 4$, so this fixed point can be represented, to a good approximation, as
\[
\mathbf{X} \approx \bigl( X_{-5},X_{-4},\dots,X_{-1},X_{1},\dots,X_{5},4 \bigr)^{\mathsf T},
\qquad
X_t = \left(\frac{1 + 4\varphi^{2t}}{5\varphi^{t}}\right)^{3},
\quad
t \in \{-5,-4,\dots,-1,1,2,\dots,5\}.
\]

Equivalently, we may write this compactly in vector form as
\[
\mathbf{X} \approx
\begin{pmatrix}
\displaystyle \left(\frac{1 + 4\varphi^{2t}}{5\varphi^{t}}\right)^{3}
\end{pmatrix}_{\substack{t=-5,\dots,-1\\[1pt]t=1,\dots,5}}
\oplus
\begin{pmatrix} 4 \end{pmatrix}.
\]

Figure~\ref{fig:Stability-1aa} shows the graph of the rational function in \eqref{Zi_eq2} together with the diagonal line $Z_{n+1} = Z_n$ for the parameter value $\varphi = 1.12434$. The intersection points of these two curves represent the fixed points of a one-dimensional dynamical system. Among them, one fixed point is located near $Z = 4$. At this point, the derivative of the rational function has a modulus strictly less than one, so orbits starting sufficiently close to $Z \approx 4$ converge to it under iteration. Hence, this fixed point is attractive in this study.

The fixed point at $Z = 1$, on the other hand, exhibits the opposite behavior for the same parameter value. There, the modulus of the derivative exceeds one, meaning that small deviations from $Z = 1$ increase with each iteration. Consequently, $Z = 1$ is a repelling (unstable) fixed point: trajectories initialized near $Z = 1$ are pushed away and eventually tend toward the stable fixed point at $Z \approx 4$.

\begin{figure}[!htbp]
\centering
\includegraphics[width=65mm]{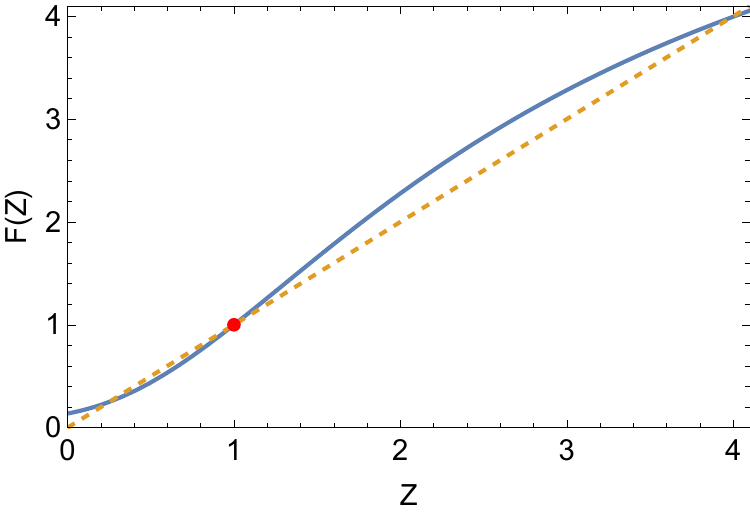}
\caption{(Color online) Fixed points of the rational function defined in \eqref{Zi_eq2} for $\varphi = 1.12434$. One of the fixed points is approximately $Z = 4$, and an analysis of its stability shows that this fixed point is attractive.}
\label{fig:Stability-1aa}
\end{figure}

In principle, one could study the local stability of this nontrivial fixed point by evaluating the Jacobian of the full system at $\mathbf{X}$ and analyzing the eigenvalues. However, in practice, the resulting formulas become prohibitively complicated and offer no additional qualitative insight beyond the stability results obtained above. Therefore, we refrain from performing this calculation here.

\section{Extremality and Non-Extremality of the Disordered Phase in the MSIM via Tree-Indexed Markov Chains}
\label{sec:Extremality}

For single-spin models on a CT, the extremality of the disordered
phase can be analyzed using a tree-indexed Markov chain determined by a
single stochastic matrix~\cite{Mukhamedov-2022}. In mixed-spin settings,
however, the corresponding description typically involves more than one
stochastic matrix. Following the strategy developed in our earlier work
\cite{Akin2024,akin2024-CJP,Akin-Muk-24JSTAT,Akin2023Chaos,Akin2025Chaos,Has-Far-2021},
we construct for the present model two explicit matrices that encode the
nearest-neighbor interaction structures. This allows us to employ the
Kesten--Stigum criterion~\cite{KestenStigum1966} to identify parameter regimes
in which the disordered phase is guaranteed to be nonextremal (reconstruction). Since the
Kesten--Stigum condition provides only a sufficient test for non-extremality,
we complement it with a Dobrushin-type analysis based on the associated
Dobrushin coefficients, which yield a sufficient condition for extremality (non-reconstruction).

Therefore, this section consists of two subsections. We now show how to construct two different stochastic matrices in two stages.

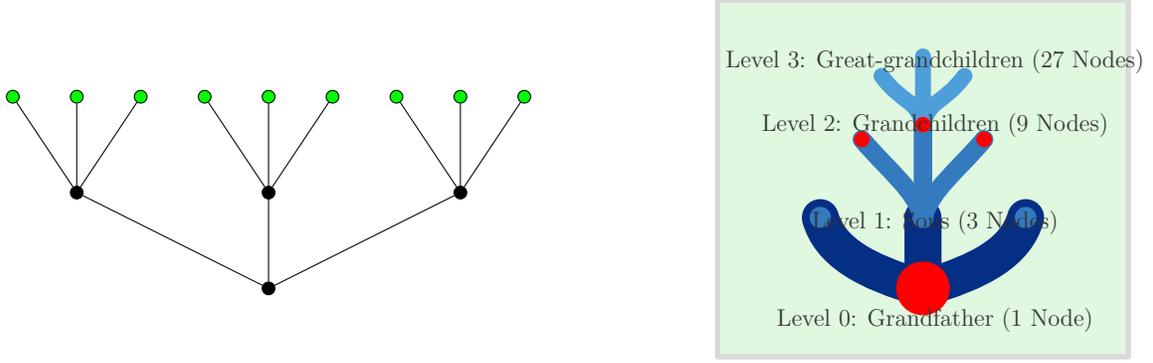
\begin{figure}[t]
\centering

\begin{tikzpicture}

\begin{scope}[
  scale=0.85, transform shape,
  grow=up,
  edge from parent/.style={draw},
  level 1/.style={sibling distance=3cm, level distance=1.5cm},
  level 2/.style={sibling distance=1cm, level distance=1.5cm},
  every node/.style={circle, draw, fill=black, inner sep=2pt}
]
\node (Root) {}
    child { node {}
        child { node[fill=green] {} }
        child { node[fill=green] {} }
        child { node[fill=green] {} }
    }
    child { node {}
        child { node[fill=green] {} }
        child { node[fill=green] {} }
        child { node[fill=green] {} }
    }
    child { node {}
        child { node[fill=green] {} }
        child { node[fill=green] {} }
        child { node[fill=green] {} }
    };
\end{scope}

\begin{scope}[xshift=8.7cm, scale=0.65, transform shape,
  x=1cm,y=1cm,
  line cap=round,
  line join=round,
  font=\fontsize{14pt}{16pt}\selectfont
]

\definecolor{darkBlue}{rgb}{0.02,0.18,0.52}
\definecolor{lightBlue}{rgb}{0.20,0.48,0.75}
\definecolor{veryLightBlue}{rgb}{0.30,0.62,0.85}
\definecolor{bg}{rgb}{0.88,0.97,0.88}
\definecolor{labelColor}{gray}{0.20}
\definecolor{frameGray}{gray}{0.85}

\fill[bg] (-4.2,-1.4) rectangle (4.2,5.9);
\draw[frameGray,line width=2pt] (-4.2,-1.4) rectangle (4.2,5.9);

\coordinate (root) at (0,0);
\coordinate (pL) at (-2.1,1.45);
\coordinate (pM) at (0,1.45);
\coordinate (pR) at (2.1,1.45);
\coordinate (qL) at (-1.25,3.05);
\coordinate (qM) at (0,3.35);
\coordinate (qR) at (1.25,3.05);
\coordinate (rL) at (-0.85,4.35);
\coordinate (rM) at (0,4.75);
\coordinate (rR) at (0.85,4.35);

\draw[darkBlue,line width=14pt]
  (root) .. controls (-1.2,0.35) and (-1.9,0.75) .. (pL);
\draw[darkBlue,line width=14pt] (root) -- (pM);
\draw[darkBlue,line width=14pt]
  (root) .. controls ( 1.2,0.35) and ( 1.9,0.75) .. (pR);

\draw[lightBlue,line width=7pt]
  (pM) .. controls (-0.2,2.05) and (-0.65,2.35) .. (qL);
\draw[lightBlue,line width=7pt] (pM) -- (qM);
\draw[lightBlue,line width=7pt]
  (pM) .. controls ( 0.2,2.05) and ( 0.65,2.35) .. (qR);

\draw[veryLightBlue,line width=6pt]
  (qM) .. controls (-0.15,3.90) and (-0.40,3.75) .. (rL);
\draw[veryLightBlue,line width=6pt] (qM) -- (rM);
\draw[veryLightBlue,line width=6pt]
  (qM) .. controls ( 0.15,3.90) and ( 0.40,3.75) .. (rR);

\fill[red] (root) circle[radius=0.55];
\fill[lightBlue] (pL) circle[radius=0.22];
\fill[lightBlue] (pM) circle[radius=0.22];
\fill[lightBlue] (pR) circle[radius=0.22];
\fill[red] (qL) circle[radius=0.16];
\fill[red] (qM) circle[radius=0.16];
\fill[red] (qR) circle[radius=0.16];
\fill[veryLightBlue] (rL) circle[radius=0.14];
\fill[veryLightBlue] (rM) circle[radius=0.14];
\fill[veryLightBlue] (rR) circle[radius=0.14];

\node[text=labelColor] at (0.25, 1.35) {Level 1: Sons (3 Nodes)};
\node[text=labelColor] at (0.25, 3.35) {Level 2: Grandchildren (9 Nodes)};
\node[text=labelColor] at (0.25, 4.65) {Level 3: Great-grandchildren (27 Nodes)};
\node[text=labelColor] at (0.25,-0.65) {Level 0: Grandfather (1 Node)};

\end{scope}

\end{tikzpicture}

\caption{(Left) A three-branch Cayley tree growing upward from the root. (Right) A three-branch Cayley tree of depth three: the root has three children (level 1), each of which has three children (level 2, 9 nodes in total), and each of those has three children (level 3, 27 nodes in total).}
\label{fig:cayley-two}
\end{figure}

\subsection{Construction of the stochastic matrices corresponding to the fixed point $\ell_0^{(5)}$ for $k=3$ and $s=5$}

To ensure compliance with the configuration definition given in Equation \eqref{eq:spin-configuration-def}, we define the corresponding stochastic matrices as follows:

\subsubsection*{Case I}\label{extremality-Case1}
Consider the configuration in which the $(n-2)$th shell of the CT
carries spins taking values in the set $\Phi$, whereas the $(n-1)$th shell is
assigned spins from set $\Psi$.  
In this setting, the corresponding transition matrix
$\mathbb{P} = (P_{ij})$ is given by
\begin{equation}\label{StocMat1}
P_{ij}
=
\frac{\exp\!\bigl(ij\,\beta J + \widetilde{h}_{j}\bigr)}
     {\displaystyle\sum_{u \in \Psi}
      \exp\!\bigl(iu\,\beta J + \widetilde{h}_{u}\bigr)},
\qquad i \in \Phi,\ j \in \Psi.
\end{equation}

Thus, our aim is to determine the parameter region in which the disordered phase corresponding to the fixed point $\ell_0^{(5)}$ is extremal. Substituting $Z = 1$ into the expression in~\eqref{StocMat1} yields the
following the form of the transition matrix.

\begin{align}\label{eq:StocM-s=5-P}
\mathbb{P}^{(5)}(\varphi)
=
\begin{pmatrix}
\dfrac{\varphi^{10}}{\varphi^{10}+1} & \dfrac{1}{\varphi^{10}+1}\\[4pt]
\dfrac{\varphi^{8}}{\varphi^{8}+1} & \dfrac{1}{\varphi^{8}+1}\\[2pt]
\vdots & \vdots\\[2pt]
\dfrac{\varphi^{2}}{\varphi^{2}+1} & \dfrac{1}{\varphi^{2}+1}\\[4pt]
\dfrac{1}{2} & \dfrac{1}{2}\\[4pt]
\dfrac{1}{1+\varphi^{2}} & \dfrac{\varphi^{2}}{1+\varphi^{2}}\\[2pt]
\vdots & \vdots\\[2pt]
\dfrac{1}{1+\varphi^{8}} & \dfrac{\varphi^{8}}{1+\varphi^{8}}\\[4pt]
\dfrac{1}{1+\varphi^{10}} & \dfrac{\varphi^{10}}{1+\varphi^{10}}
\end{pmatrix}.  
\end{align}

\subsubsection*{Case II}\label{extremality-Case2}
Suppose that the \((n-1)\)th shell of the CT is decorated 
with spins from the set \(\Psi\), 
while the \(n\)th shell is decorated with spins from the set 
\(\Phi\), where $n$ is an even positive integer.  

In this case, the corresponding probability transition matrix 
\(\mathbb{Q}=(Q_{ij})\) is given by
\begin{equation}\label{StocMat2} 
Q_{ij} \;=\; 
\frac{\exp\!\bigl(ij \beta J+\pmb{h}_{j}\bigr)}
     {\sum\limits_{u\in \Phi} 
      \exp\!\bigl(i u \beta J+\pmb{h}_{u}\bigr)},
\end{equation}
where \(i \in \Psi \) and 
\(j \in \Phi\).

From this definition \eqref{StocMat2}, we get the following probability transition matrix
\[
\mathbb{Q} \;=\; \bigl(Q_{ij}\bigr)_{i\in\{-\tfrac12,\tfrac12\},\,j=-5}^{5}
\;=\;
\begin{pmatrix}
\displaystyle\frac{\exp\!\bigl(\tfrac{5\beta J}{2}+\pmb{h}_{-5}\bigr)}{Z_{-}} &
\displaystyle\frac{\exp\!\bigl(\tfrac{4\beta J}{2}+\pmb{h}_{-4}\bigr)}{Z_{-}} &
\cdots &
\displaystyle\frac{\exp\!\bigl(-\tfrac{5\beta J}{2}+\pmb{h}_{5}\bigr)}{Z_{-}} \\[10pt]
\displaystyle\frac{\exp\!\bigl(\tfrac{-5\beta J}{2}+\pmb{h}_{-5}\bigr)}{Z_{+}} &
\displaystyle\frac{\exp\!\bigl(\tfrac{-4\beta J}{2}+\pmb{h}_{-4}\bigr)}{Z_{+}} &
\cdots &
\displaystyle\frac{\exp\!\bigl(\tfrac{5\beta J}{2}+\pmb{h}_{5}\bigr)}{Z_{+}}
\end{pmatrix},
\]
where \[
Z_{-} \;=\; \sum_{i=-5}^{5} \exp\!\Bigl(-\tfrac{i\beta J}{2} + \pmb{h}_i\Bigr), 
\qquad
Z_{+} \;=\; \sum_{i=-5}^{5} \exp\!\Bigl(\tfrac{i\beta J}{2} + \pmb{h}_i\Bigr).
\]
The transition probability matrix $\mathbb{Q}$ is a $2\times 11$ row-stochastic matrix.
Its rows are indexed by the spin states $i\in\Psi$ in the $(n\!-\!1)$st shell, while its
columns are indexed by the spin values $j\in\Phi$ in the $n$th shell. Introducing
\[
\varphi = e^{\beta J/2}, \qquad
X_j = e^{\mathbf{h}_j-\mathbf{h}_0}\quad (X_0=1),
\]
and evaluating at the fixed point \eqref{IP-FP-1a} (so that $X_j$ take their fixed-point
values), the entries of $\mathbb{Q}$ are
\begin{equation}\label{eq:Q_entries}
\mathbb{Q}_{-,j}=\frac{X_j\varphi^{-j}}{S_-}, \qquad
\mathbb{Q}_{+,j}=\frac{X_j\varphi^{\,j}}{S_+},
\qquad j=-5,-4,\ldots,5,
\end{equation}
where
\[
S_-=\sum_{i=-5}^{5} X_i\,\varphi^{-i}, \qquad
S_+=\sum_{i=-5}^{5} X_i\,\varphi^{\,i}.
\]
For $s=5$, this can be written schematically as
\[
\mathbb{Q}=
\left(
\begin{array}{ccccccc}
\frac{X_{-5}\varphi^{5}}{S_-} & \cdots & \frac{X_{-1}\varphi}{S_-} & \frac{X_{0}}{S_-}
& \frac{X_{1}\varphi^{-1}}{S_-} & \cdots & \frac{X_{5}\varphi^{-5}}{S_-} \\
\frac{X_{-5}\varphi^{-5}}{S_+} & \cdots & \frac{X_{-1}\varphi^{-1}}{S_+} & \frac{X_{0}}{S_+}
& \frac{X_{1}\varphi}{S_+} & \cdots & \frac{X_{5}\varphi^{5}}{S_+}
\end{array}
\right).
\]

Thus, we can write the $2\times 11$ stochastic matrix $\mathbb{Q}^{(5)}(\varphi)$ in the form
\begin{align}\label{eq:Stochas-Mat-Q(5)}
\mathbb{Q}^{(5)}(\varphi)
=
\frac{1}{S(\varphi)}
\left(
\begin{array}{ccccccccccc}
h_5(\varphi) & h_4(\varphi) & h_3(\varphi) & h_2(\varphi) & h_1(\varphi) & 1
             & g_1(\varphi) & g_2(\varphi) & g_3(\varphi) & g_4(\varphi) & g_5(\varphi)\\
g_5(\varphi) & g_4(\varphi) & g_3(\varphi) & g_2(\varphi) & g_1(\varphi) & 1
             & h_1(\varphi) & h_2(\varphi) & h_3(\varphi) & h_4(\varphi) & h_5(\varphi)
\end{array}
\right),
\end{align}
where
\[
h_n(\varphi)=\frac{(\varphi^{2n}+1)^3}{8\varphi^{2n}},\qquad 
g_n(\varphi)=\frac{(\varphi^{2n}+1)^3}{8\varphi^{4n}},\qquad S^{(5)}=1+\frac18\sum_{n=1}^{5}\frac{(\varphi^{2n}+1)^4}{\varphi^{4n}}\qquad n=1,\dots,5
\]

Note that
\[
h_n(\varphi)+g_n(\varphi)=\frac{(\varphi^{2n}+1)^4}{8\varphi^{4n}},
\]
so each row sum equals
\[
1+\frac18\sum_{n=1}^{5}\frac{(\varphi^{2n}+1)^4}{\varphi^{4n}},
\]
and the prefactor makes $\mathbb{Q}^{(5)}(\varphi)$ row-stochastic.

Referring to the right panel of Fig.~\ref{fig:cayley-two}, the entries of the matrix $\mathbb{P}^{(5)}(\varphi)$ are transition probabilities from a red vertex to a blue vertex, whereas the entries of $\mathbb{Q}^{(5)}(\varphi)$ are transition probabilities from a blue vertex to a red vertex. Accordingly, the entries of the product $\mathbb{P}^{(5)}(\varphi)\mathbb{Q}^{(5)}(\varphi)$ give the conditional probabilities between two consecutive red vertices in Fig. ~\ref{fig:cayley-two} (right panel), while the entries of $\mathbb{Q}^{(5)}(\varphi)\mathbb{P}^{(5)}(\varphi)$ yield the conditional probabilities from a first-level blue vertex to any blue third-level

\subsection{The extremality criterion (non-reconstruction) for a TISGM}\label{subsec:extremality}

In this subsection, we compute the Dobrushin coefficients for the stochastic matrices \(\mathbb{P}\) and \(\mathbb{Q}\), whose elements are given by Equations \eqref{StocMat1} and \eqref{StocMat2}, respectively.

The extremality of the disordered phase of the mixed--spin $(s,\tfrac12)$
Ising model on a CT of order $k$ can be analysed by means of
tree--indexed Markov chains. To recall the notion of a tree--indexed
Markov chain, let $V$ be the vertex set of a tree, let $\nu$ be a
probability measure on a finite single--site state space
$\mathcal{B}=\{1,2,\dots,q\}$, and let
$P=(P_{ij})_{i,j\in\mathcal{B}}$ be a transition matrix on
$\mathcal{B}$. A tree--indexed Markov chain is then a random field
$X:V\to\mathcal{B}$ obtained by first sampling the root $X(x_0)$
according to $\nu$, and then, for each vertex $v\neq x_0$, sampling
$X(v)$ from the transition probabilities determined by the state of its
parent, independently of the rest of the tree (see, for example,
\cite{Mossel2001,MosselPeres2003,KuelskeRozikov2017,Martin2003,MartinelliSinclairWeitz2007,
Moraal1978}).  

In order to investigate the extremality criterion for a TISGM, we recall
the quantities $\kappa$ and $\gamma$ introduced in~\cite{KuelskeRozikov2017}.
Let $\mu^{s}_{T_x}$ denote the finite-volume Gibbs measure on the subtree
$T_x$ in which the parent of $x$ has its spin fixed to $s$, and the
configuration on the bottom boundary of $T_x$ (i.e., on
$\partial T_x \setminus \{\text{parent of }x\}$) is specified by a boundary
condition $\eta$ (see, for example, \cite{Mossel2001,MosselPeres2003,KuelskeRozikov2017}).

For two probability measures $\mu_1$ and $\mu_2$ on the single-site space
$\mathcal{B}$, we denote by $\|\mu_1-\mu_2\|_x$ the variation distance
between their projections onto the spin at a vertex $x$, namely
\[
\|\mu_1 - \mu_2\|_x
=
\frac{1}{2}\sum_{i\in\mathcal{B}}
\bigl|\mu_1(\sigma(x)=i) - \mu_2(\sigma(x)=i)\bigr|.
\]
For a configuration $\eta$ and a site $x$, we write $\eta_{x,s}$ for the
configuration obtained from $\eta$ by setting the spin at $x$ to $s$.

\begin{defin}[{\cite[p.~146]{KuelskeRozikov2017}}]\label{def:kappa-gamma}
For a family of Gibbs measures $\{\mu^{s}_{T_x}\}$, define
$\kappa \equiv \kappa(\{\mu^{s}_{T_x}\})$ and
$\gamma \equiv \gamma(\{\mu^{s}_{T_x}\})$ by
\begin{itemize}
\item[(1)]
\[
\kappa
=
\sup_{z\in V}\;
\max_{w\prec z}\;
\max_{s_1,s_1',\,s_2,s_2'}
\sqrt{\,
\bigl\|\mu^{s_1}_{T_z} - \mu^{s_1'}_{T_z}\bigr\|_z\,
\bigl\|\mu^{s_2}_{T_w} - \mu^{s_2'}_{T_w}\bigr\|_w
\,},
\]
where $w\prec z$ indicates that $w$ is a neighbor (child) of $z$.

\item[(2)]
\[
\gamma
=
\sup_{A\subset V}\;
\max_{y\in A}\;
\max_{z\in A:\,z\sim y}\;
\max_{s_1,s_1',\,s_2,s_2'}
\bigl\|\mu^{\eta_{y,s_1}}_{A} - \mu^{\eta_{y,s_1'}}_{A}\bigr\|_z\,
\bigl\|\mu^{\eta_{z,s_2}}_{A\setminus\{z\}} -
       \mu^{\eta_{z,s_2'}}_{A\setminus\{z\}}\bigr\|_z,
\]
where the supremum is taken over all finite sets $A\subset V$, all boundary
conditions $\eta$, all neighbours $z\in A$ of $y\in A$, all neighbours
$w\in A$ of $z$, and all spins
$s_1,s_1'\in\Phi$, $s_2,s_2'\in\Psi$.
\end{itemize}
\end{defin}

Let \(\mathbb{H}\) be a stochastic matrix of size \(m \times n\). The \textbf{Dobrushin coefficient} \(\tau_{\mathbb{H}}\), which quantifies the contractive property of \(\mathbb{H}\), is defined by
\[
\tau_{\mathbb{H}} = \frac{1}{2} \max_{i, j} \sum_{\ell=1}^{n} \bigl| H_{i\ell} - H_{j\ell} \bigr|.
\]
Here, \(\ell\) indexes the elements of the single-site state space.

It is known (see~\cite[Theorem~3.3']{KuelskeRozikov2017}) that a sufficient
condition for the extremality of a translation-invariant Gibbs measure on a
CT of order $k$ is
\begin{equation}\label{eq:kappa-gamma-condition}
k\,\kappa\,\gamma < 1.
\end{equation}
Moreover, in our setting the constant $\kappa$ admits the particularly simple
representation
\[
\kappa = \sqrt{\tau_{\mathbb{P}}\tau_{\mathbb{Q}}},
\]
We define the Dobrushin coefficients associated with the stochastic matrices $\mathbb{P}^{(s)}(\varphi)$ and $\mathbb{Q}^{(s)}(\varphi)$ given in \eqref{eq:Matrix-P-arbitrary-s} and \eqref{eq:Matrix-Q-arbitrary-s}, respectively,  by
\begin{align}\label{eqs:Dobrushin-def1a}
\tau_{\mathbb{P}^{(s)}}(\varphi)
=
\frac{1}{2}\max_{i,j}
\left\{
\sum_{\ell=1}^{2}
\bigl|P_{i\ell} - P_{j\ell}\bigr|
\right\},
\qquad
\tau_{\mathbb{Q}^{(s)}}(\varphi)
=
\frac{1}{2}\max_{i,j}
\left\{
\sum_{\ell=1}^{2s+1}
\bigl|Q_{i\ell} - Q_{j\ell}\bigr|
\right\}.  
\end{align}

As illustrated in the right panel of Fig. In Figure  ~\ref{fig:cayley-two}, we consider a three-branch Cayley tree in which transitions between generations follow a bipartite structure. Specifically, matrix $\mathbb{P}$ represents the transitions from red vertices to blue vertices at the subsequent level, while $\mathbb{Q}$ denotes the transitions from these blue vertices to the red vertices in the following level. To guarantee the extremality of the disordered phase, the following contractive condition must be satisfied:
\begin{equation}\label{inequ:Dobrushin-3}
    k \cdot \tau_{\mathbb{P}^{(s)}}(\varphi)\cdot \tau_{\mathbb{Q}^{(s)}}(\varphi) < 1,
\end{equation}
where $k=3$ is the branching factor and $\tau_{\mathbb{P}^{(s)}}(\varphi), \tau_{\mathbb{Q}^{(s)}}(\varphi)$ are the corresponding Dobrushin coefficients. This inequality ensures that the information decay across consecutive levels (red-to-blue-to-red) is sufficiently strong to suppress the geometric growth of the tree, preventing the boundary conditions from influencing the root.

\begin{lem}\label{lem:Dobrushin-general}
Let $s\in\{1,2,3,4,5\}$ and let $\tau_{\mathbb{P}^{(s)}}(\varphi)$ and 
$\tau_{\mathbb{Q}^{(s)}}(\varphi)$ denote the Dobrushin coefficients of the
stochastic matrices $\mathbb{P}^{(s)}(\varphi)$ and $\mathbb{Q}^{(s)}(\varphi)$
corresponding to the mixed spin $(s,\tfrac12)$ model. Then, for all
$\varphi>0$, from \eqref{eqs:Dobrushin-def1a}, one has
\begin{equation}\label{eq:tauP-s-phi-def}
\tau_{\mathbb{P}^{(s)}}(\varphi)
=
\frac{\bigl|\varphi^{2s} - 1\bigr|}{\varphi^{2s} + 1},
\end{equation}
and
\begin{equation}\label{eq:tauQ-s-phi-def}
\tau_{\mathbb{Q}^{(s)}}(\varphi)
=
\frac{1}{8\,S_s(\varphi)}
\sum_{n=1}^{s}
\frac{\bigl(\varphi^{2n}+1\bigr)^3\,\bigl|\varphi^{2n}-1\bigr|}{\varphi^{4n}},
\qquad
S_s(\varphi)
=
1+\frac18\sum_{n=1}^{s}\frac{\bigl(\varphi^{2n}+1\bigr)^4}{\varphi^{4n}}.
\end{equation}
In particular, for $s=5$ the explicit expressions reduce to
\begin{equation}\label{eq:Dob-Coff-P(5)-a}
\tau_{\mathbb{P}^{(5)}}(\varphi)
= \frac{|\varphi^{10} - 1|}{\varphi^{10} + 1},
\end{equation}
and
\begin{equation}\label{eq:Dob-Coff-Q(5)-a}
\tau_{\mathbb{Q}^{(5)}}(\varphi)
=
\frac{1}{8\,S(\varphi)}
\sum_{n=1}^{5}
\frac{(\varphi^{2n}+1)^3\,\bigl|\varphi^{2n}-1\bigr|}{\varphi^{4n}},
\qquad
S_5(\varphi)
=
1+\frac18\sum_{n=1}^{5}\frac{(\varphi^{2n}+1)^4}{\varphi^{4n}}.
\end{equation} 
Moreover, the following properties hold.
\begin{enumerate}
  \item[(i)] \textbf{Symmetry.} For all $\varphi>0$,
  \begin{equation}\label{eq:tau-symmetry}
  \tau_{\mathbb{P}^{(s)}}(\varphi)
  =
  \tau_{\mathbb{P}^{(s)}}\!\left(\frac{1}{\varphi}\right),
  \qquad
  \tau_{\mathbb{Q}^{(s)}}(\varphi)
  =
  \tau_{\mathbb{Q}^{(s)}}\!\left(\frac{1}{\varphi}\right).
  \end{equation}
  
  \item[(ii)] \textbf{Bounds and values at $\varphi=1$.} For all $\varphi>0$,
  \begin{equation}\label{eq:tau-bounds}
  0 \,\le\, \tau_{\mathbb{P}^{(s)}}(\varphi) < 1,
  \qquad
  0 \,\le\, \tau_{\mathbb{Q}^{(s)}}(\varphi) < 1,
  \end{equation}
  and
  \begin{equation}\label{eq:tau-at-1}
  \tau_{\mathbb{P}^{(s)}}(1) = 0,
  \qquad
  \tau_{\mathbb{Q}^{(s)}}(1) = 0.
  \end{equation}
  
  \item[(iii)] \textbf{Continuity and asymptotics.} The functions
  $\varphi\mapsto\tau_{\mathbb{P}^{(s)}}(\varphi)$ and 
  $\varphi\mapsto\tau_{\mathbb{Q}^{(s)}}(\varphi)$ are continuous on $(0,\infty)$.
  Furthermore,
  \begin{equation}\label{eq:tau-asymptotics}
  \lim_{\varphi\to 0} \tau_{\mathbb{P}^{(s)}}(\varphi)
  =
  \lim_{\varphi\to\infty} \tau_{\mathbb{P}^{(s)}}(\varphi)
  = 1,
  \end{equation}
  and
  \begin{equation}\label{eq:tauQ-asymptotics}
  \lim_{\varphi\to 0} \tau_{\mathbb{Q}^{(s)}}(\varphi)
  =
  \lim_{\varphi\to\infty} \tau_{\mathbb{Q}^{(s)}}(\varphi)
  = 1.
  \end{equation}
\end{enumerate}
In particular, for each $s\in\{1,2,3,4,5\}$,
$\tau_{\mathbb{P}^{(s)}}(\varphi)$ and $\tau_{\mathbb{Q}^{(s)}}(\varphi)$ are
symmetric with respect to the point $\varphi=1$, vanish at $\varphi=1$, and
approach $1$ as $\varphi\to 0$ or $\varphi\to\infty$.
\end{lem}

\begin{proof}
We first treat the case $s=5$, for which the matrices
$\mathbb{P}^{(5)}(\varphi)$ and $\mathbb{Q}^{(5)}(\varphi)$ are given explicitly
in \eqref{eq:StocM-s=5-P} and \eqref{eq:Stochas-Mat-Q(5)}.

For $\mathbb{P}^{(5)}(\varphi)$ each row has the form $(a_i,1-a_i)$, so that
\[
\tau_{\mathbb{P}^{(5)}}(\varphi)
= \frac{1}{2}\max_{i,j} \left\{\sum_{\ell=1}^{2} |P_{i\ell} - P_{j\ell}|\right\}
= \max_{i,j} |a_i - a_j|,
\]
where $a_i = P_{i1}=\frac{1}{1+\varphi^{2i}}$ for $i\in\{-5,-4,\dots,5\}$.

Since $a_i$ is strictly decreasing in $i$ (for $\varphi>0$), the extreme values
occur at $i=-5$ and $i=5$, namely
\[
a_{\max} = a_{-5}=\frac{\varphi^{10}}{\varphi^{10} + 1},
\qquad
a_{\min} = a_{5}=\frac{1}{\varphi^{10} + 1}.
\]
Hence
\[
\max_{i,j} |a_i - a_j|
=
\left|\frac{\varphi^{10}}{\varphi^{10} + 1}
      -\frac{1}{\varphi^{10} + 1}\right|
=
\frac{|\varphi^{10} - 1|}{\varphi^{10} + 1},
\]
which yields \eqref{eq:Dob-Coff-P(5)-a}.

For the matrix $\mathbb{Q}^{(5)}(\varphi)$ in \eqref{eq:Stochas-Mat-Q(5)}, the
Dobrushin coefficient is
\[
\tau_{\mathbb{Q}^{(5)}}(\varphi)
=
\frac{1}{2}\max_{i,j}\left\{\sum_{\ell=1}^{11}
\bigl|\mathbb{Q}^{(5)}_{i\ell}(\varphi) - \mathbb{Q}^{(5)}_{j\ell}(\varphi)\bigr|\right\}.
\]
Since $\mathbb{Q}^{(5)}(\varphi)$ has only two rows, the maximum over $(i,j)$
reduces to this pair, and we obtain
\begin{equation}\label{eq:tauQ5-general}
\tau_{\mathbb{Q}^{(5)}}(\varphi)
=
\frac{1}{2}\sum_{\ell=1}^{11}
\bigl|\mathbb{Q}^{(5)}_{1\ell}(\varphi) - \mathbb{Q}^{(5)}_{2\ell}(\varphi)\bigr|.
\end{equation}

From \eqref{eq:Stochas-Mat-Q(5)}, we see that the differences between the two rows come in pairs
$h_n-g_n$ and $g_n-h_n$, whereas the middle column contributes zero. Hence
\eqref{eq:tauQ5-general} simplifies to
\[
\tau_{\mathbb{Q}^{(5)}}(\varphi)
=
\frac{1}{S(\varphi)}
\sum_{n=1}^{5} \bigl|h_n(\varphi) - g_n(\varphi)\bigr|.
\]
A direct calculation gives
\[
h_n(\varphi)-g_n(\varphi)
=
\frac{(\varphi^{2n}+1)^3}{8\varphi^{4n}}\bigl(\varphi^{2n}-1\bigr),
\]
so that
\[
\bigl|h_n(\varphi)-g_n(\varphi)\bigr|
=
\frac{(\varphi^{2n}+1)^3}{8\varphi^{4n}}\bigl|\varphi^{2n}-1\bigr|.
\]
Substituting this into the previous expression yields
\[
\tau_{\mathbb{Q}^{(5)}}(\varphi)
=
\frac{1}{8\,S(\varphi)}
\sum_{n=1}^{5}
\frac{(\varphi^{2n}+1)^3\,\bigl|\varphi^{2n}-1\bigr|}{\varphi^{4n}},
\]
with
\[
S(\varphi)
=
1+\frac18\sum_{n=1}^{5}\frac{(\varphi^{2n}+1)^4}{\varphi^{4n}},
\]
which is precisely \eqref{eq:Dob-Coff-Q(5)-a}.

For general $s\in\{1,2,3,4,5\}$ the matrices $\mathbb{P}^{(s)}(\varphi)$ and
$\mathbb{Q}^{(s)}(\varphi)$ have the same structure, with $10$ replaced by $2s$
and the sums run from $n=1$ to $s$. Repeating the above arguments with $5$
replaced by $s$ yields the general formulas
\eqref{eq:tauP-s-phi-def}--\eqref{eq:tauQ-s-phi-def}. The symmetry
\eqref{eq:tau-symmetry}, the bounds \eqref{eq:tau-bounds}, the values at
$\varphi=1$ in \eqref{eq:tau-at-1}, and the asymptotics
\eqref{eq:tau-asymptotics}--\eqref{eq:tauQ-asymptotics} follow immediately from
these explicit representations. This completes the proof.
\end{proof}

\begin{thm}\label{thm:extremality-Dobrushin-s}
Let $s\in\{1,2,3,4,5\}$ and consider the mixed spin $(s,\tfrac12)$
Ising model on CT of order $k=3$. Let
$\mathbb{P}^{(s)}(\varphi)$ and $\mathbb{Q}^{(s)}(\varphi)$ be the
stochastic matrices associated with the disordered
translation-invariant Gibbs measure $\mu_0^{(s)}$, and let
$\tau_{\mathbb{P}^{(s)}}(\varphi)$ and $\tau_{\mathbb{Q}^{(s)}}(\varphi)$
be their Dobrushin coefficients, given by
\eqref{eq:tauP-s-phi-def}--\eqref{eq:tauQ-s-phi-def} in
Lemma~\ref{lem:Dobrushin-general}.

Assume that the parameters $\kappa$ and $\gamma$ satisfy
\[
\kappa^2
=
\tau_{\mathbb{P}^{(s)}}(\varphi)\,\tau_{\mathbb{Q}^{(s)}}(\varphi),
\qquad
\gamma = \kappa.
\]
Then the Dobrushin-type condition \eqref{eq:kappa-gamma-condition}
for the extremality of $\mu_0^{(s)}$ reduces to
\begin{equation}\label{eq:Dobrushin-reduced-s}
3\,\tau_{\mathbb{P}^{(s)}}(\varphi)\,\tau_{\mathbb{Q}^{(s)}}(\varphi) < 1.
\end{equation}

Moreover, for each $s\in\{1,2,3,4,5\}$ there exist numbers
\[
0 < \varphi_1^{(s)} < 1 < \varphi_2^{(s)},
\qquad
\varphi_1^{(s)}\,\varphi_2^{(s)} = 1,
\]
such that
\begin{equation}\label{eq:Dobrushin-E-region}
3\,\tau_{\mathbb{P}^{(s)}}(\varphi)\,\tau_{\mathbb{Q}^{(s)}}(\varphi) < 1
\quad\Longleftrightarrow\quad
\varphi \in \bigl(\varphi_1^{(s)},\,\varphi_2^{(s)}\bigr).
\end{equation}
For all $\varphi$ in this interval the corresponding disordered
translation-invariant Gibbs measure $\mu_0^{(s)}$ is an

The endpoints $\varphi_1^{(s)}$ and $\varphi_2^{(s)}$ are the (numerically
determined) solutions of
\[
3\,\tau_{\mathbb{P}^{(s)}}(\varphi)\,\tau_{\mathbb{Q}^{(s)}}(\varphi) = 1,
\]
and the symmetry
\[
\tau_{\mathbb{P}^{(s)}}(\varphi)\,\tau_{\mathbb{Q}^{(s)}}(\varphi)
=
\tau_{\mathbb{P}^{(s)}}\!\left(\tfrac{1}{\varphi}\right)
\tau_{\mathbb{Q}^{(s)}}\!\left(\tfrac{1}{\varphi}\right)
\]
implies $\varphi_1^{(s)} = 1/\varphi_2^{(s)}$.
\end{thm}

\begin{lem}\cite[Theorem~5.3]{Akin-Muk-24JSTAT}
Let $\varphi>0$. Then
\[
\frac{3(\varphi^4 - 1)^2}{\varphi^8 + 4\varphi^6 + 14\varphi^4 + 4\varphi^2 + 1} - 1 < 0
\]
holds if and only if
\begin{equation}\label{ineq:extrmal-reg-s=1}
\sqrt{\frac{1+\sqrt{13} - \sqrt{10 + 2\sqrt{13}}}{2}}
<
\varphi
<
\sqrt{\frac{1+\sqrt{13} + \sqrt{10 + 2\sqrt{13}}}{2}}.
\end{equation}
\end{lem}

\begin{proof}
Since the denominator
\[
\varphi^8 + 4\varphi^6 + 14\varphi^4 + 4\varphi^2 + 1 > 0
\qquad(\varphi>0)
\]
is strictly positive, the inequality
\[
\frac{3(\varphi^4 - 1)^2}{\varphi^8 + 4\varphi^6 + 14\varphi^4 + 4\varphi^2 + 1} - 1 < 0
\]
is equivalent to
\begin{equation}\label{eq:ineq-poly-phi}
3(\varphi^4 - 1)^2
<
\varphi^8 + 4\varphi^6 + 14\varphi^4 + 4\varphi^2 + 1.
\end{equation}
Introduce the variable $t=\varphi^2>0$. Then \eqref{eq:ineq-poly-phi} becomes
\[
3(t^2 - 1)^2
<
t^4 + 4t^3 + 14t^2 + 4t + 1,
\]
or, after rearranging,
\begin{equation}\label{eq:F-neg}
F(t) := t^4 - 2t^3 - 10t^2 - 2t + 1 < 0.
\end{equation}
Dividing by $t^2>0$ we obtain
\[
\frac{F(t)}{t^2}
= t^2 + \frac{1}{t^2} - 2\Bigl(t + \frac{1}{t}\Bigr) - 10.
\]
Set
\[
u := t + \frac{1}{t} \ge 2.
\]
Then $t^2 + t^{-2} = u^2 - 2$, and therefore
\[
\frac{F(t)}{t^2}
= (u^2 - 2) - 2u - 10
= u^2 - 2u - 12
= (u-1)^2 - 13.
\]
Thus \eqref{eq:F-neg} is equivalent to
\[
(u-1)^2 - 13 < 0
\quad\Longleftrightarrow\quad
(u-1)^2 < 13
\quad\Longleftrightarrow\quad
1-\sqrt{13} < u < 1+\sqrt{13}.
\]
Since $u\ge 2$, the left inequality is automatically satisfied, and we are left with
\[
u = t + \frac{1}{t} < 1+\sqrt{13}.
\]

The equation $t + t^{-1} = 1+\sqrt{13}$ has exactly two positive roots
\[
t_{1,2}
=
\frac{1+\sqrt{13} \pm \sqrt{(1+\sqrt{13})^2 - 4}}{2}
=
\frac{1+\sqrt{13} \pm \sqrt{10 + 2\sqrt{13}}}{2},
\]
and the inequality $t + t^{-1} < 1+\sqrt{13}$ is equivalent to
\[
t_1 < t < t_2.
\]
Recalling that $t=\varphi^2$, we obtain
\[
t_1 < \varphi^2 < t_2
\quad\Longleftrightarrow\quad
\sqrt{t_1} < \varphi < \sqrt{t_2},
\]
that is,
\[
\sqrt{\frac{1+\sqrt{13} - \sqrt{10 + 2\sqrt{13}}}{2}}
<
\varphi
<
\sqrt{\frac{1+\sqrt{13} + \sqrt{10 + 2\sqrt{13}}}{2}}.
\]
This proves the claimed equivalence.
\end{proof}

Here one can easily show that \[
\sqrt{\frac{1+\sqrt{13} - \sqrt{10 + 2\sqrt{13}}}{2}}\,
\sqrt{\frac{1+\sqrt{13} + \sqrt{10 + 2\sqrt{13}}}{2}} = 1.
\]

\begin{align*}\label{ineq:extrmal-reg-s=1}\nonumber
\sqrt{\frac{1+\sqrt{13} - \sqrt{10 + 2\sqrt{13}}}{2}}
&< \varphi <
\sqrt{\frac{1+\sqrt{13} + \sqrt{10 + 2\sqrt{13}}}{2}} \\[4pt]
&\approx 0.478 \;< \varphi <\; 2.092.
\end{align*}

\begin{rem}
Lemma~\ref{lem:Dobrushin-general} implies that inequality
\eqref{eq:Dobrushin-reduced-s} holds on the interval
\eqref{ineq:extrmal-reg-s=1}. Hence, in the case $s=1$ the measure
$\mu_0^{(1)}$ describes precisely the parameter region
$0.478 \;< \varphi <\; 2.092$ in which the disordered phase is extremal.
This conclusion is in full agreement with \cite[Theorem~5.3]{Akin-Muk-24JSTAT}.
\end{rem}

\subsection{ Extremality of the disordered phase associated with $(2,1/2)$-MSIM on the CT of order three}
We start from the explicit formulas for $s=2$. From  the general formulas
\eqref{eq:tauP-s-phi-def}--\eqref{eq:tauQ-s-phi-def}, we have
\[
\tau_{\mathbb{P}^{(2)}}(\varphi)
=
\frac{\bigl|\varphi^{4}-1\bigr|}{\varphi^{4}+1},
\qquad
\tau_{\mathbb{Q}^{(2)}}(\varphi)
=
\bigl|\varphi^{2}-1\bigr|\,
\frac{(\varphi^{4}+1)^{3}(\varphi^{2}+1) + \varphi^{4}(\varphi^{2}+1)^{3}}
     {8\varphi^{8} + (\varphi^{2}+1)^{4}\varphi^{4} + (\varphi^{4}+1)^{4}}.
\]

From the formula $3\,\tau_{\mathbb{P}^{(2)}}(\varphi)\,\tau_{\mathbb{Q}^{(2)}}(\varphi)$. Thus, the final explicit form is
\[
3\,\tau_{\mathbb{P}^{(2)}}(\varphi)\,\tau_{\mathbb{Q}^{(2)}}(\varphi)
=\frac{3\,\bigl|\varphi^{2}-1\bigr|^{2}(\varphi^{2}+1)^{2}\left((\varphi^{4}+1)^{3} + \varphi^{4}(\varphi^{2}+1)^{2}\right)}
     {(\varphi^{4}+1)\,\bigl[8\varphi^{8} + (\varphi^{2}+1)^{4}\varphi^{4} + (\varphi^{4}+1)^{4}\bigr]}.
\]

For the case \( s = 2 \), the disordered phase is extremal if the Dobrushin condition
\[
3\,\tau_{\mathbb{P}^{(2)}}(\varphi)\,\tau_{\mathbb{Q}^{(2)}}(\varphi) - 1 < 0
\]
holds, where \(\tau_{\mathbb{P}^{(2)}}(\varphi)\) and \(\tau_{\mathbb{Q}^{(2)}}(\varphi)\) are defined in \eqref{eq:tauP-s-phi-def}--\eqref{eq:tauQ-s-phi-def}. The inequality describes the parameter range (in \(\varphi\)) in which the disordered phase is extremal.

The boundary of this extremality region is defined by the equation
\begin{equation}
\label{eq:extremality}
\mathcal{D}^{(s)}(\varphi):=3\,\tau_{\mathbb{P}^{(2)}}(\varphi)\,\tau_{\mathbb{Q}^{(2)}}(\varphi)-1 =0.
\end{equation}
Because an exact analytic solution of \eqref{eq:extremality} is extremely difficult to obtain, the region was computed numerically using \textsc{Mathematica}  \cite{Mathematica}.
We found exactly two positive solutions, $\varphi_1^{(2)}$ and $\varphi_2^{(2)}$, which satisfy the following approximate values:
\begin{align*}
\varphi_1^{(2)} & \approx 0.6763353, \\
\varphi_2^{(2)} & \approx 1.4785566.
\end{align*}
Owing to the symmetry transformation $\varphi \mapsto 1/\varphi$, these solutions are related by their product being equal to unity: $\varphi_1^{(2)}\,\varphi_2^{(2)} = 1$. This implies that $\varphi_1^{(2)} = 1/\varphi_2^{(2)}$.

Consequently, for $s=2$ the
disordered phase $\mu_0^{(2)}$ corresponding to the fixed point $\ell_0^{(2)}$ is extremal
precisely in the interval
\[
\varphi \in \bigl(\varphi_1^{(2)},\,\varphi_2^{(2)}\bigr)
\approx (0.6763353,\,1.4785566).
\]

Using \textit{Mathematica} \cite{Mathematica} and considering Algorithm~\ref{alg:dobrushin_extremality}, we generated the following table for the spin values $s=1,2,3,4,5$.

\begin{table}[h]
\centering
\begin{tabular}{c|c|c}
\hline
$s$ & Interval where $3\tau_{\mathbb{P}^{(s)}}(\varphi)\tau_{\mathbb{Q}^{(s)}}(\varphi)-1<0$ \\ \hline
1 & $0.478 \lesssim \varphi \lesssim 2.092$ \\
2 & $0.676 \lesssim \varphi \lesssim 1.479$ \\
3 & $0.765 \lesssim \varphi \lesssim 1.308$ \\
4 & $0.815 \lesssim \varphi \lesssim 1.227$ \\
5 & $0.847 \lesssim \varphi \lesssim 1.180$ \\
\hline
\end{tabular}
\caption{Numerical intervals in $\varphi>0$ where $3\tau_{\mathbb{P}^{(s)}}(\varphi)\tau_{\mathbb{Q}^{(s)}}(\varphi)-1<0$ for $s=1,2,3,4,5$.}\label{tab:extremality-1-5}
\end{table}

\begin{thm}\label{thm:extremal-intervals-s1-5}
Let $\mu_0^{(s)}$ denote the disordered translation-invariant Gibbs measure of the mixed-spin $(s,\tfrac12)$ Ising model on the CT of order three for $s \in \{1,2,3,4,5\}$.
\par
Then, for each value of $s$, the measure $\mu_0^{(s)}$ is extremal precisely on the set of parameters $\varphi>0$ for which the following inequality holds
\[
3\,\tau_{\mathbb{P}^{(s)}}(\varphi)\,\tau_{\mathbb{Q}^{(s)}}(\varphi) - 1 < 0.
\]
The corresponding extremality regions, determined numerically by solving $3\,\tau_{\mathbb{P}^{(s)}}(\varphi)\,\tau_{\mathbb{Q}^{(s)}}(\varphi)=1$, are given by the following $\varphi$ intervals (the bounds are provided up to the displayed precision):

\begin{enumerate}
    \item[(i)] For $s=1$: $\mu_0^{(1)}$ is extremal for $0.478 \;\lesssim\; \varphi \;\lesssim\; 2.092$ \cite[Theorem~5.3]{Akin-Muk-24JSTAT}.
    \item[(ii)] For $s=2$: $\mu_0^{(2)}$ is extremal for $0.676 \;\lesssim\; \varphi \;\lesssim\; 1.479.$
    \item[(iii)] For $s=3$: $\mu_0^{(3)}$ is extremal for $0.765 \;\lesssim\; \varphi \;\lesssim\; 1.308.$
    \item[(iv)] For $s=4$: $\mu_0^{(4)}$ is extremal for $0.815 \;\lesssim\; \varphi \;\lesssim\; 1.227.$
    \item[(v)] For $s=5$: $\mu_0^{(5)}$ is extremal for $0.847 \;\lesssim\; \varphi \;\lesssim\; 1.180.$
\end{enumerate}
\end{thm}

\begin{algorithm}[!htbp]
\caption{Extremality test of the disordered phase via the Dobrushin criterion ($k=3$)}
\label{alg:dobrushin_extremality}
\begin{algorithmic}[1]
\Require Spin quantum number $s\in\mathbb{Z}^+$, thermal parameter $\varphi=e^{\beta J/2} \in\mathbb{R}^+$.
\Ensure Decision for the disordered phase: \textsc{Extremal (certified)} or \textsc{Inconclusive}; value of the criterion $\mathcal{D}^{(s)}(\varphi)$.

\State \textbf{(1) Auxiliary functions and normalization}
\For{$n=1$ \textbf{to} $s$}
  \State $h_n(\varphi)\gets \dfrac{(\varphi^{2n}+1)^3}{8\,\varphi^{2n}}$
  \State $g_n(\varphi)\gets \dfrac{(\varphi^{2n}+1)^3}{8\,\varphi^{4n}}$
\EndFor
\State $S_s(\varphi)\gets 1+\dfrac{1}{8}\displaystyle\sum_{n=1}^{s}\dfrac{(\varphi^{2n}+1)^4}{\varphi^{4n}}$

\vspace{0.25em}
\State \textbf{(2) Compute Dobrushin coefficient $\tau_{\mathbb{Q}}^{(s)}(\varphi)$}
\Comment{$\mathbb{Q}^{(s)}(\varphi)\in\mathbb{R}^{2\times(2s+1)}$; the maximum over rows reduces to the difference of the two rows.}
\State $\tau_{\mathbb{Q}}^{(s)}(\varphi)\gets \dfrac{1}{S_s(\varphi)}\displaystyle\sum_{n=1}^{s}\bigl|h_n(\varphi)-g_n(\varphi)\bigr|$

\vspace{0.25em}
\State \textbf{(3) Compute Dobrushin coefficient $\tau_{\mathbb{P}}^{(s)}(\varphi)$}
\Comment{$\mathbb{P}^{(s)}(\varphi)\in\mathbb{R}^{(2s+1)\times 2}$ with $m\in\{-s,\ldots,s\}$.}
\State $P_{m,-1/2}(\varphi)\gets \dfrac{1}{1+\varphi^{2m}}$ \hfill (and $P_{m,+1/2}=1-P_{m,-1/2}$)
\State $\tau_{\mathbb{P}}^{(s)}(\varphi)\gets
\max\limits_{m,m'\in\{-s,\ldots,s\}}
\bigl|P_{m,-1/2}(\varphi)-P_{m',-1/2}(\varphi)\bigr|$
\Comment{The map $m\mapsto P_{m,-1/2}(\varphi)$ is monotone for fixed $\varphi>0$, hence the maximum is attained at $m=s$ and $m'=-s$.}
\State $\tau_{\mathbb{P}}^{(s)}(\varphi)\gets
\left|\dfrac{1}{1+\varphi^{2s}}-\dfrac{1}{1+\varphi^{-2s}}\right|$

\vspace{0.25em}
\State \textbf{(4) Apply the Dobrushin criterion ($k=3$)}
\Comment{Dobrushin provides a \emph{sufficient} condition for extremality: $3\,\tau_{\mathbb{P}}^{(s)}(\varphi)\tau_{\mathbb{Q}}^{(s)}(\varphi)<1$.}
\State $\mathcal{D}^{(s)}(\varphi)\gets 3\,\tau_{\mathbb{P}}^{(s)}(\varphi)\,\tau_{\mathbb{Q}}^{(s)}(\varphi)-1$
\If{$\mathcal{D}^{(s)}(\varphi)<0$}
  \State \Return \textsc{Extremal (certified)}, $\mathcal{D}^{(s)}(\varphi)$
\Else
  \State \Return \textsc{Inconclusive}, $\mathcal{D}^{(s)}(\varphi)$
\EndIf
\end{algorithmic}
\end{algorithm}

In Figure~\ref{fig:extremality-s}, we display the graphs of the Dobrushin function
$\mathcal{D}^{(s)}(\varphi)=3\,\tau_{\mathbb{P}^{(s)}}(\varphi)\,\tau_{\mathbb{Q}^{(s)}}(\varphi)-1$
for $s=1,2,3,4,5$. Each value of $s$ is represented by a different color.
These curves indicate, for every $s$, the parameter intervals in which the
disordered phase associated with the measure $\mu_0^{(s)}$ is an extremal.
One observes that, as $s$ increases, the interval where the corresponding
curve lies below zero, becomes narrower. The negative region of the function
Therefore, marks the extremality region of the disordered phase $\mu_0^{(s)}$.
In comparison with our previous work, this shows that the order $k$ of the
CT significantly influences the size of the extremality region.
\begin{figure}[!htbp]
\centering
\includegraphics[width=100mm]{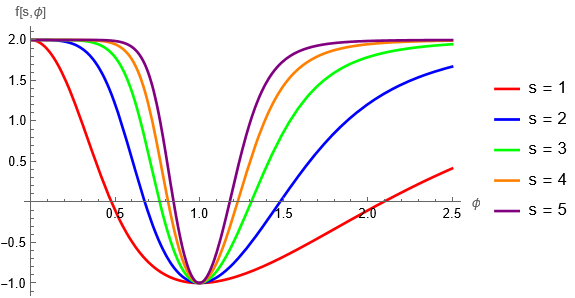}
\caption{(Color online) Plots of the function 
$\mathcal{D}^{(s)}(\varphi)=3\,\tau_{\mathbb{P}^{(s)}}(\varphi)\,\tau_{\mathbb{Q}^{(s)}}(\varphi)-1$
for $s = 1,2,3,4,5$. These curves illustrate, for each $s$, the parameter
regions in which the disordered phase associated with the measure
$\mu_0^{(s)}$ is extremal.}
\label{fig:extremality-s}
\end{figure}
\begin{rem}
We observe that as the spin parameter $s$ increases, the interval of $\varphi$ values for which the disordered phase is extremal narrows down and symmetrically concentrates towards $\varphi = 1$ (see Figure~\ref{fig:extremality-s}).
\end{rem}

\subsection{The non-extremality (reconstruction) of the disordered phases}\label{subsec:non-extremality}
The tree-indexed Markov chain approach is used to analyze whether disordered
phases of lattice models on Cayley-type structures are nonextremal \cite{RR-2021,Rozikov2018}. In this
setting, the CT is assumed to be semi-infinite with branching number
$k$, meaning that each vertex has $k$ forward descendants. One constructs a
stochastic matrix based on nearest-neighbor interactions and studies its
spectral properties. 

\begin{thm}\label{thm:reconstruction-}\cite[Theorem 4.1.]{Mossel2004}  Let $M$ be a channel corresponding to an ergodic Markov chain. Let $T_b$ be a $b$-ary tree. The reconstruction problem is census-solvable if $b|\lambda_2(M)|^2 > 1$, and is not census-solvable if $b|\lambda_2(M)|^2 < 1$. For general trees, the reconstruction problem is solvable when 
\begin{align}\label{eq:KS-condition-0}
br(T)|\lambda_2(M)|^2 > 1,
\end{align}
where $br(T)$ denotes the branching number of the tree.
\end{thm}

The expression in \eqref{eq:KS-condition-0} represents the Kesten-Stigum condition~\cite{KestenStigum1966}. This section explores the specific parameters and regions where inequality holds true. According to the KS condition~\cite{KestenStigum1966},
a sufficient criterion for a Gibbs measure $\mu$ associated with a stochastic
matrix $H$ to be non-extreme is
\begin{align}\label{ex:KS-Condition-k=3}
 k\,|\lambda_2(H)|^{2} -1>0,
\end{align}
where $\lambda_2(H)$ denotes the second-largest eigenvalue (in absolute value) of $H$ in absolute
value.

\subsubsection{Non-extremality of the Disordered Phase Using a Markov Chain Construction for $s=2$}

To analyze the non-extremality of the disordered phase for spin $s=2$, we consider a mixed-spin configuration on the CT: one shell carries spins in $\{-2,-1,0,1,2\}$, while the adjacent shell carries spins in $\{-\tfrac12, \tfrac12\}$. Along each edge connecting these shells, the spin configuration can be described using a two-step Markov chain. The transition from a spin-2 shell to a spin-$\tfrac12$ shell is governed by the matrix $\mathbb{P}^{(2)}(\varphi)$, and the reverse transition by $\mathbb{Q}^{(2)}(\varphi)$.  

We obtain the symmetric fixed point $\ell_0^{(2)}$ as 
\[
Z= e^{\widetilde{h}_{1/2}-\widetilde{h}_{-1/2}} = 1,
\qquad
X_{\pm 1}= \Bigl(\tfrac{\varphi^{2}+1}{2\varphi}\Bigr)^{3},
\qquad
X_{\pm 2} = \Bigl(\tfrac{\varphi^{4}+1}{2\varphi^{2}}\Bigr)^{3},
\qquad
X_{0}= 1.
\]
Substituting these relations into the general expressions for
$\mathbb{P}^{(2)}(\varphi)$ and $\mathbb{Q}^{(2)}(\varphi)$ yields the following explicit
$\varphi$-dependent 

For $i \in \{-2,-1,0,1,2\}$ and $j \in \{-\tfrac12,\tfrac12\}$,  from  \eqref{StocMat1} and \eqref{eq:Matrix-P-arbitrary-s}, the stochastic matrix $\mathbb{P}^{(2)}(\varphi)$ at $\ell_0^{(2)}$ is given by
\begin{align}\label{P-rho-s2-explicit}
\mathbb{P}^{(2)}(\varphi)=
\begin{pmatrix}
\dfrac{\varphi^{4}}{1+\varphi^{4}} & \dfrac{1}{1+\varphi^{4}} \\[0.4em]
\dfrac{\varphi^{2}}{1+\varphi^{2}} & \dfrac{1}{1+\varphi^{2}} \\[0.4em]
\dfrac{1}{2} & \dfrac{1}{2} \\[0.4em]
\dfrac{1}{1+\varphi^{2}} & \dfrac{\varphi^{2}}{1+\varphi^{2}} \\[0.4em]
\dfrac{1}{1+\varphi^{4}} & \dfrac{\varphi^{4}}{1+\varphi^{4}}
\end{pmatrix}.
\end{align}
The rows are ordered by $i=-2,-1,0,1,2$, and the columns by $j=-\tfrac12,\tfrac12$.

We consider the mixed $(2,\tfrac12)$-spin Ising model and focus on the
transition matrix from the spin-$\tfrac12$ layer $\Psi = \{-\tfrac12,\tfrac12\}$
to the spin-$2$ layer $\Phi = \{-2,-1,0,1,2\}$. Similarly, from \eqref{StocMat2} and \eqref{eq:Matrix-Q-arbitrary-s}, the transition matrix $\mathbb{Q}^{(2)}(\varphi)$ can be written explicitly as
\begin{equation}\label{Q-rho-s2-explicit}
\mathbb{Q}^{(2)}(\varphi)
=
\frac{1}{S_2(\varphi)}
\begin{pmatrix}
\displaystyle
\frac{(\varphi^{4}+1)^{3}}{8\varphi^{4}} &
\displaystyle
\frac{(\varphi^{2}+1)^{3}}{8\varphi^{2}} &
1 &
\displaystyle
\frac{(\varphi^{2}+1)^{3}}{8\varphi^{4}} &
\displaystyle
\frac{(\varphi^{4}+1)^{3}}{8\varphi^{8}}
\\[1.0em]
\displaystyle
\frac{(\varphi^{4}+1)^{3}}{8\varphi^{8}} &
\displaystyle
\frac{(\varphi^{2}+1)^{3}}{8\varphi^{4}} &
1 &
\displaystyle
\frac{(\varphi^{2}+1)^{3}}{8\varphi^{2}} &
\displaystyle
\frac{(\varphi^{4}+1)^{3}}{8\varphi^{4}}
\end{pmatrix},
\end{equation}
where \begin{equation*}\label{S-rho-s2}
S_2(\varphi)
=
1 + \frac{(\varphi^{2}+1)^{4}}{8\varphi^{4}}
  + \frac{(\varphi^{4}+1)^{4}}{8\varphi^{8}}.
\end{equation*}

\subsubsection{Effective two-step transition matrix and eigenvalues for $s=2$}

For $s=2$, at the symmetric fixed point $\ell_0^{(2)}$ the one-step transition
matrices $\mathbb{P}^{(2)}(\varphi)$ and $\mathbb{Q}^{(2)}(\varphi)$ (from the spin--$2$
shell to the spin--$\tfrac12$ shell, and back) depend only on the single
parameter $\varphi = e^{\beta J/2}$.  
The effective two-step transition along a ray of the tree is then governed by
the $2\times 2$ matrix
\[
\mathbb{R}^{(2)}(\varphi) := \mathbb{Q}^{(2)}(\varphi)\,\mathbb{P}^{(2)}(\varphi).
\]

Using the explicit forms of $\mathbb{P}^{(2)}(\varphi)$ and $\mathbb{Q}^{(2)}(\varphi)$
given in \eqref{P-rho-s2-explicit} and \eqref{Q-rho-s2-explicit}, a straightforward computation yields
\begin{equation}\label{eq:M-rho-matrix-s2}
\mathbb{R}^{(2)}(\varphi)
=
\frac{1}{D(\varphi)}
\begin{pmatrix}
A(\varphi) & B(\varphi) \\
B(\varphi) & A(\varphi)
\end{pmatrix},
\end{equation}
where
\begin{align}
A(\varphi)
&=
\varphi^{16} + 3\varphi^{12} + 2\varphi^{10} + 8\varphi^{8}
+ 2\varphi^{6} + 3\varphi^{4} + 1,
\label{eq:A-rho-s2}\\[4pt]
B(\varphi)
&=
2\varphi^{4}\bigl(\varphi^{8} + \varphi^{6} + 6\varphi^{4} + \varphi^{2} + 1\bigr),
\label{eq:B-rho-s2}\\[4pt]
D(\varphi)
&=
\varphi^{16} + 5\varphi^{12} + 4\varphi^{10} + 20\varphi^{8}
+ 4\varphi^{6} + 5\varphi^{4} + 1.
\label{eq:D-rho-s2}
\end{align}
By construction one has
\begin{equation}\label{eq:D-equals-A-plus-B-s2}
D(\varphi) = A(\varphi) + B(\varphi),
\end{equation}
so that $\mathbb{R}^{(2)}(\varphi)$ is a symmetric stochastic matrix.

Since $\mathbb{R}^{(2)}(\varphi)$ is of the form
\(
\bigl[\begin{smallmatrix}A & B\\ B & A\end{smallmatrix}\bigr]/D
\),
its eigenvalues are
\begin{equation}\label{eq:M-eigs-raw-s2}
\lambda_1(\varphi) = \frac{A(\varphi)+B(\varphi)}{D(\varphi)},
\qquad
\lambda_2(\varphi) = \frac{A(\varphi)-B(\varphi)}{D(\varphi)}.
\end{equation}
Using \eqref{eq:D-equals-A-plus-B-s2}, the leading eigenvalue reduces to
\begin{equation*}\label{eq:lambda1-equals-1-s2}
\lambda_1(\varphi) = 1,
\end{equation*}
as expected for the stochastic matrix.

The second largest eigenvalue is given explicitly by
\begin{equation}\label{eq:lambda2-rational-s2}
\lambda_2(\varphi)
=
\frac{A(\varphi)-B(\varphi)}{D(\varphi)}
=
\frac{\varphi^{16} + \varphi^{12} - 4\varphi^{8} + \varphi^{4} + 1}
     {\varphi^{16} + 5\varphi^{12} + 4\varphi^{10} + 20\varphi^{8} + 4\varphi^{6} + 5\varphi^{4} + 1},
\qquad \varphi = e^{\beta J/2} > 0.
\end{equation}
Moreover, the numerator factorises as
\begin{equation*}\label{eq:lambda2-factorised-num-s2}
\varphi^{16} + \varphi^{12} - 4\varphi^{8} + \varphi^{4} + 1
=
(\varphi-1)^{2}(\varphi+1)^{2}(\varphi^{2}+1)^{2}(\varphi^{8}+3\varphi^{4}+1),
\end{equation*}
so that
\begin{equation}\label{eq:lambda2-factorised-s2}
\lambda_2(\varphi)
=
\frac{(\varphi-1)^{2}(\varphi+1)^{2}(\varphi^{2}+1)^{2}(\varphi^{8}+3\varphi^{4}+1)}
     {\varphi^{16} + 5\varphi^{12} + 4\varphi^{10} + 20\varphi^{8} + 4\varphi^{6} + 5\varphi^{4} + 1}.
\end{equation}
Since the denominator in \eqref{eq:lambda2-rational-s2} is strictly
positive for all $\varphi>0$, the sign of $\lambda_2(\varphi)$ is determined
by the numerator, which is non–negative and vanishes only at $\varphi=1$.
Thus
\[
\lambda_2(\varphi) \ge 0,\qquad
\lambda_2(1) = 0,\qquad
\lambda_2(\varphi) > 0 \ \text{for}\ \varphi \neq 1.
\]
Furthermore,
\[
\lim_{\varphi\to 0^{+}}\lambda_2(\varphi) = 1,
\qquad
\lim_{\varphi\to +\infty}\lambda_2(\varphi) = 1,
\]
and one checks that $\lambda_2(\varphi)$ is strictly decreasing on $(0,1)$
and strictly increasing on $(1,\infty)$, with a unique global minimum at
$\varphi=1$. The symmetry
\[
\lambda_2(\varphi) = \lambda_2\bigl(1/\varphi\bigr)
\]
follows directly from the factorisation
\eqref{eq:lambda2-factorised-s2}.

We now analyse the inequality
\begin{equation}\label{eq:ineq-3lambda2-s2}
3\lambda_2(\varphi)^2 - 1 > 0,
\end{equation}
which is the non–extremality criterion in the present case. Since
$\lambda_2(\varphi)\ge 0$ for all $\varphi>0$, the inequality
\eqref{eq:ineq-3lambda2-s2} is equivalent to
\begin{equation}\label{eq:ineq-abs-s2}
3\lambda_2(\varphi)^2 - 1 > 0
\quad\Longleftrightarrow\quad
\lambda_2(\varphi) > \frac{1}{\sqrt{3}}.
\end{equation}

By continuity and monotonicity of $\lambda_2(\varphi)$ on $(0,1)$ and
$(1,\infty)$, there exist exactly two positive solutions
$\varphi_- \in (0,1)$ and $\varphi_+ \in (1,\infty)$ of the equation
\[
\lambda_2(\varphi) = \frac{1}{\sqrt{3}},
\]
and they satisfy $\varphi_- = 1/\varphi_+$. Numerically,
\begin{equation}\label{eq:rho-pm-s2}
\varphi_- \approx 0.561522,
\qquad
\varphi_+ \approx 1.780873.
\end{equation}
Combining \eqref{eq:ineq-abs-s2} and \eqref{eq:rho-pm-s2}, we obtain
\begin{equation}\label{eq:ineq-regions-s2}
3\lambda_2(\varphi)^2 - 1 > 0
\quad\Longleftrightarrow\quad
0 < \varphi < \varphi_- \;\;\text{ or }\;\; \varphi > \varphi_+.
\end{equation}

In the ferromagnetic regime $J>0$, where $\varphi=e^{\beta J/2}>1$, only
the upper branch is physically relevant, so that
\eqref{eq:ineq-3lambda2-s2} reduces to
\begin{equation*}\label{eq:ineq-phys-s2}
3\lambda_2(\varphi)^2 - 1 > 0
\quad\Longleftrightarrow\quad
\varphi > \varphi_+ \approx 1.780873.
\end{equation*}
Equivalently, in terms of $\beta J$ this corresponds to
\[
\beta J > 2\ln \varphi_+ \approx 1.15421,
\]
which characterises the parameter region in which the disordered fixed
point $Z=1$ loses its stability along the ray.

 Therefore, we can write the following theorem.
\begin{thm}\label{thm:non-extremal-s2-k3}
Let $\mu_0^{(2)}$ be the free boundary (disordered) Gibbs measure associated with the disordered fixed point $\ell_0^{(2)}$ of the mixed $(2,\tfrac12)$-spin Ising model on the CT of order $k=3$. Then the disordered phase (DP) $\mu_0^{(2)}$ is \emph{not extremal} in the regime
\begin{equation*}\label{eq:rho-region-nonext-s2-short}
\varphi \in (0,\varphi_-)\,\cup\,(\varphi_+,+\infty),
\end{equation*}
where the thresholds are $\varphi_- \approx 0.561522$ and $\varphi_+ \approx 1.780873$ (with $\varphi_- = \varphi_+^{-1}$).
\end{thm}

\begin{table}[h]
\centering
\caption{KS condition analysis for $s = 1,2,3,4,5$. 
For each $s$, the table lists the positive roots of 
$3\bigl|\lambda_{\max}^{(s)}(\varphi)\bigr|^2 - 1 = 0$ 
and the corresponding non-extremal parameter regions 
$3\bigl|\lambda_{\max}^{(s)}(\varphi)\bigr|^2 > 1$.}\label{tab:KS_results}
\begin{tabular}{ccc}
\toprule
$s$ & $\varphi$ roots of $3|\lambda_{\max}^{(s)}|^2 = 1$ & Regions where $3|\lambda_{\max}^{(s)}|^2 > 1$ \\
\midrule
1 & $0.3453$, $2.8957$ & $(0, 0.3453) \cup (2.8957, \infty)$ \\
2 & $0.5615$, $1.7809$ & $(0, 0.5615) \cup (1.7809, \infty)$ \\
3 & $0.6692$, $1.4943$ & $(0, 0.6692) \cup (1.4943, \infty)$ \\
4 & $0.7341$, $1.3622$ & $(0, 0.7341) \cup (1.3622, \infty)$ \\
5 & $0.7776$, $1.2861$ & $(0, 0.7776) \cup (1.2861, \infty)$ \\
\bottomrule
\end{tabular}
\end{table}
To generate Table \ref{tab:KS_results}, we compute the second-largest eigenvalue of the product matrix $\mathbb{R}^{(s)}(\varphi) = \mathbb{Q}^{(s)}(\varphi) \mathbb{P}^{(s)}(\varphi)$, where the stochastic matrices $\mathbb{P}^{(s)}$ and $\mathbb{Q}^{(s)}$ are defined in \eqref{eq:Matrix-P-arbitrary-s} and \eqref{eq:Matrix-Q-arbitrary-s} for arbitrary $s$. These calculations were performed using a custom code in Mathematica.

Let $\mu_0^{(s)}$ denote the disordered, translation-invariant Gibbs measure of the mixed-spin $(s,\tfrac12)$ Ising model on the CT of order three for $s \in \{1,2,3,4,5\}$.
The disordered Gibbs measure $\mu_0^{(s)}$ is \emph{non-extremal} (i.e., reconstruction occurs) if the KS condition $3|\lambda_{\max}^{(s)}(\varphi)|^2 > 1$ is satisfied \eqref{ex:KS-Condition-k=3}. This condition defines a symmetric non-extremality region:
\begin{equation}
\varphi \in (0,\varphi_{-}^{(s)})\,\cup\,(\varphi_{+}^{(s)},+\infty),
\end{equation}
where $\varphi_{-}^{(s)}$ and $\varphi_{+}^{(s)}$ are the roots of $3|\lambda_{\max}^{(s)}(\varphi)|^2 = 1$, such that $\varphi_{-}^{(s)} = (\varphi_{+}^{(s)})^{-1}$.
\par
Numerically, the bounds for the non-extremality regions are given by the following approximate values:
\begin{enumerate}
    \item[(i)] For $s=1$: $\varphi_{-}^{(1)} \approx 0.3453$, $\varphi_{+}^{(1)} \approx 2.8957$.
    \item[(ii)] For $s=2$: $\varphi_{-}^{(2)} \approx 0.5615$, $\varphi_{+}^{(2)} \approx 1.7809$.
    \item[(iii)] For $s=3$: $\varphi_{-}^{(3)} \approx 0.6692$, $\varphi_{+}^{(3)} \approx 1.4943$.
    \item[(iv)] For $s=4$: $\varphi_{-}^{(4)} \approx 0.7341$, $\varphi_{+}^{(4)} \approx 1.3622$.
    \item[(v)] For $s=5$: $\varphi_{-}^{(5)} \approx 0.7776$, $\varphi_{+}^{(5)} \approx 1.2861$.
\end{enumerate}
In light of these observations, we derive the main theorem of this subsection as follows:

\begin{thm}\label{thm:nonextreme-KS}
Let $k=3$ and let $\lambda_{\max}^{(s)}(\varphi)$ denote the second-largest
eigenvalue in absolute value of the $2\times2$ matrix
$\mathbb{R}^{(s)}(\varphi)=\mathbb{Q}^{(s)}(\varphi)\mathbb{P}^{(s)}(\varphi)$
for the mixed--spin $(s,\tfrac12)$ Ising model, with $s=1,2,3,4,5$.
Then the Kesten--Stigum condition
\(
3\,\bigl|\lambda_{\max}^{(s)}(\varphi)\bigr|^{2} > 1
\)
holds precisely on the following intervals of $\varphi>0$
\begin{itemize}
  \item $s=1$:\quad $\varphi \in (0,\,0.3453)\,\cup\,(2.8957,\,\infty)$,
  \item $s=2$:\quad $\varphi \in (0,\,0.5615)\,\cup\,(1.7809,\,\infty)$,
  \item $s=3$:\quad $\varphi \in (0,\,0.6692)\,\cup\,(1.4943,\,\infty)$,
  \item $s=4$:\quad $\varphi \in (0,\,0.7341)\,\cup\,(1.3622,\,\infty)$,
  \item $s=5$:\quad $\varphi \in (0,\,0.7776)\,\cup\,(1.2861,\,\infty)$,
\end{itemize}
where the numerical values are approximations as listed in
Table~\ref{tab:KS_results}. In particular, for these values of $\varphi$,
the corresponding disordered Gibbs measure $\mu_0^{(s)}$ is nonextremal.
\end{thm}

\begin{figure}[h!]
    \centering
    \includegraphics[width=0.6\textwidth]{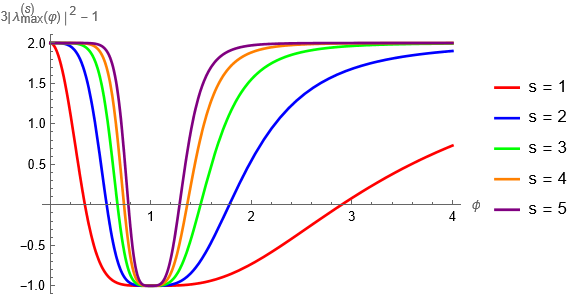}
    \caption{(Color online) Plot of the Kesten--Stigum function 
    $g(s,\varphi) = 3|\lambda_{\max}^{(s)}(\varphi)|^2 - 1$ 
    for the mixed-spin $(s,\tfrac{1}{2})$ Ising model on the CT of order three, for $s = 1,\dots,5$. 
    For each $s$, the interval of $\varphi$ where $g(s,\varphi) > 0$ identifies the non-extremality (reconstruction) regime of the disordered Gibbs measure $\mu_0^{(s)}$.}
    \label{fig:ks_non_extremality_plot}
\end{figure}

Figure~\ref{fig:ks_non_extremality_plot} displays the Kesten--Stigum function 
$g(s,\varphi)=3|\lambda_{\max}^{(s)}(\varphi)|^{2}-1$ for the mixed $(s,1/2)$ 
Ising model on the ternary CT, with spin parameters $s=1,\dots,5$. 
For each $s$, the condition $g(s,\varphi)>0$ marks the set of coupling values 
$\varphi$ for which the disordered Gibbs measure $\mu_{0}^{(s)}$ is 
non‑extremal, that is, the regime in which reconstruction is possible. 
As $s$ increases, the non‑extremal interval widens, in contrast to the 
behavior of the extremal region. This observation raises the following question: 
what is the physical interpretation of the disordered phase $\mu_{0}^{(s)}$ 
in parameter regions that are neither purely extremal nor purely non‑extremal?

\begin{figure}[!htbp]
\centering
\includegraphics[width=0.4\textwidth]{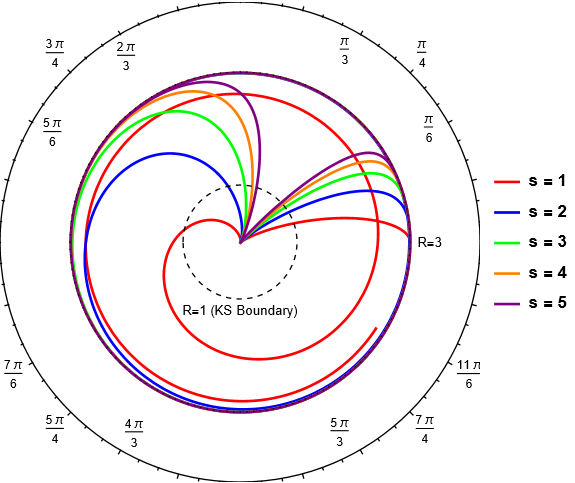}
\caption{(Color online) Polar-coordinate representation of the functions $g(s,\varphi)+1$ shown in Fig.~\ref{fig:ks_non_extremality_plot} for spin quantum numbers $s=1,2,\ldots,5$. This view highlights how the non-extremal region varies with temperature (via $\varphi$) and with the spin value.}
\label{fig:ks_non_extremality_Polar-plot}
\end{figure}

Figure~\ref{fig:ks_non_extremality_Polar-plot} presents a polar-coordinate version of the functions plotted in Fig.~\ref{fig:ks_non_extremality_plot}. In this representation, the dependence of $g(s,\varphi)+1$ on the temperature parameter $\varphi$ becomes transparent, making it easy to track how the curves evolve as $\varphi$ increases toward satisfying the KS condition. As shown, $g(s,\varphi)+1$ approaches the limiting value $3$ asymptotically. The disordered phases are not nonextremal inside the unit circle ($R<1$), whereas they become nonextremal in the annular region bounded by $R=1$ and $R=3$ (i.e., $1<R<3$).

\subsection{Comparative analysis of extremality criteria and determination of inconclusive regions}

In this subsection we identify the parameter ranges in which the extremality status of the disordered phase cannot be settled, by quantifying the gap between the Dobrushin and Kesten--Stigum thresholds, with emphasis on the spin-\(2\) case. For convenience, the Dobrushin-based extremality check is outlined in Algorithm~\ref{alg:dobrushin_extremality}. We label the set \(\mathcal{D}^{(2)}(\varphi)\ge 0\) as \emph{inconclusive}, since meeting this condition (or, likewise, the Kesten--Stigum bound) is necessary for non-extremality but does not, by itself, certify it.

As the subsequent \(s=2\) analysis shows, overlaying the proven extremal region with the region predicted to be nonextremal leaves two intermediate intervals in which the disordered phase cannot be definitively classified as extremal or nonextremal (see Fig. ~\ref{fig:inconclusive-region1}). In the corresponding phase diagram, the blue region marks the verified nonextremal regime determined by the Kesten--Stigum criterion, whereas the orange region denotes the extremal regime ensured by the Dobrushin condition. The remaining white bands indicate parameter values where neither criterion yields a decisive conclusion, thereby visualizing the separation between the two theoretical bounds in the spin-\(2\) setting.

More generally, we adopt the same convention for an arbitrary spin \(s\): the domain where \(\mathcal{D}^{(s)}(\varphi) \ge 0\) is treated as inconclusive, as these threshold tests provide only one-sided guarantees of extremality. Consequently, even when the extremal and predicted non-extremal regions are identified, there may be intermediate parameter windows (cf. Fig.~\ref{fig:inconclusive-region1}), where the qualitative nature of the disordered phase remains unresolved. As demonstrated in this comparative analysis, Theorem~\ref{thm:reconstruction-} fails to provide a definitive classification for the mixed-spin Ising model within these specific regimes.

\begin{figure}[!htbp]
\centering
\includegraphics[width=110mm]{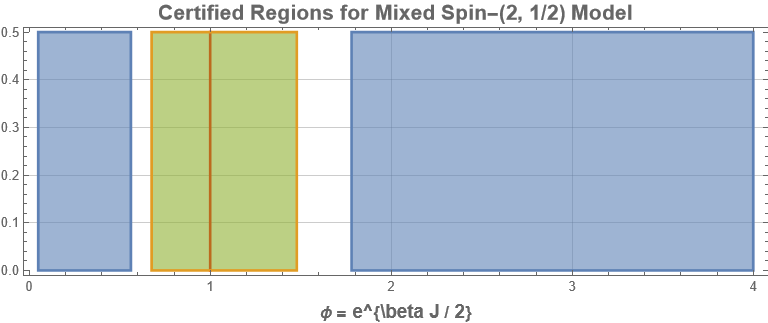}
\caption{(Color online) The blue area represents the confirmed non-extremal phase, while the orange area shows the region where the disordered phase is extremal. The white areas correspond to regions that do not satisfy either of these conditions.}
\label{fig:inconclusive-region1}
\end{figure}

\subsection{Three readings of $\varphi=e^{J/(2T)}$ (local stability, extremality, reconstruction): the sharp $s=5$ comparison}
\label{subsec:s5-three-way-comparison}

We encode coupling vs.\ temperature by the single dimensionless parameter
\[
\varphi=\exp\!\Big(\frac{J}{2T}\Big),\qquad (k_B=1),
\]
so that
\[
J>0 \iff \varphi>1 \ \text{(ferromagnetic)},\qquad
J<0 \iff 0<\varphi<1 \ \text{(antiferromagnetic)}.
\]
Because the model is bipartite, it is natural to measure information survival via the
\emph{two-step} effective channel on $\Psi$,
\[
\mathbb R^{(s)}_\Psi(\varphi)=\mathbb Q^{(s)}(\varphi)\,\mathbb P^{(s)}(\varphi),
\]
which compares the states on the same parity sublattice (grandparent $\to$ grandchild). In
particular, the relevant long-range mode is governed by the magnitude of the second-largest eigenvalue $|\lambda_{\max}^{(s)}(\varphi)|$, and the resulting thresholds are expected to
appear in reciprocal pairs under $\varphi\mapsto 1/\varphi$ (equivalently, $J\mapsto -J$).

\medskip
\noindent\textbf{For $s=5$ there are three numerically determined ``milestones'', one per language}

\paragraph{(i) Dynamical stability / phase-transition signal (repelling $\ell_0^{(5)}$).}
Theorem~\ref{thm:repelling-l0} states that the symmetric fixed point $\ell_0^{(5)}$ is
\emph{repelling} for
\[
\varphi\in(0,\,0.901258\ldots)\ \cup\ (1.109560\ldots,\,\infty),
\]
hence it is locally \emph{attracting} only in the narrow high-temperature window
$\varphi\in(0.901258\ldots,\,1.109560\ldots)$ around $\varphi=1$. In terms of the
dimensionless ratio
\[
\frac{|J|}{T}=2|\log\varphi|,
\]
the loss of local stability begins at approximately $|J|/T\simeq 0.208$ (on both ferro and
antiferro sides).

\paragraph{(ii) Extremality $\leftrightarrow$ non-reconstruction (Dobrushin certification)}
For $s=5$ (see Tab. \ref{tab:extremality-1-5}), the Dobrushin-type sufficient condition
\[
3\,\tau_{\mathbb P^{(5)}}(\varphi)\,\tau_{\mathbb Q^{(5)}}(\varphi)<1
\]
holds on the numerical interval
\[
\varphi\in(0.847,\,1.180).
\]
Therefore, throughout this region the free-boundary disordered state $\mu_0^{(5)}$ is
\emph{certified extremal}, equivalently boundary information is \emph{certified} to die out
(non-reconstruction).

\paragraph{(iii) Two-step channel amplification / reconstruction (Kesten--Stigum certification).}
For $s=5$, the KS sufficient condition
\[
3\,\bigl|\lambda_2^{(5)}(\varphi)\bigr|^2>1
\]
is certified on
\[
\varphi\in(0,\,0.7776)\ \cup\ (1.2861,\,\infty),
\]
equivalently at strong coupling $|J|/T\gtrsim 2\ln(1.2861)\approx 0.503$.
On the antiferromagnetic side ($\varphi\ll 1$), the surviving long-range mode is naturally
\emph{staggered} across levels; using the two-step kernel $\mathbb R_\Psi$ aligns with this
because it probes same-parity correlations, so anti-alignment on single edges does not
artificially cancel the signal.

\begin{figure}[htbp]
    \centering
    \includegraphics[width=0.45\textwidth]{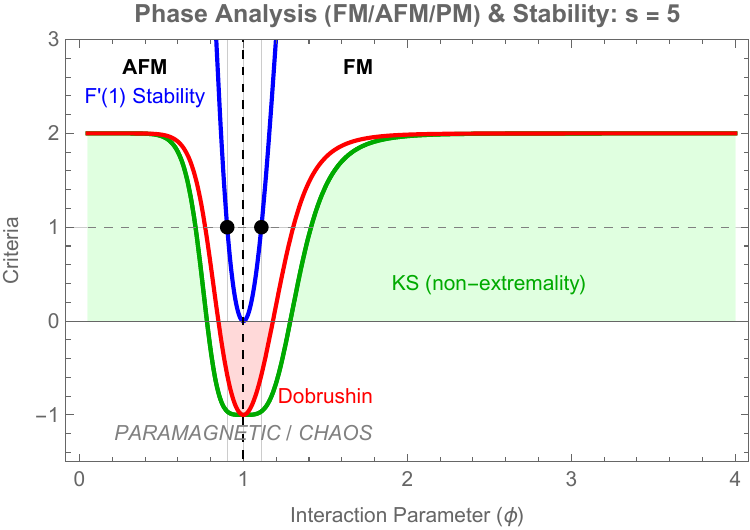}
    \caption{\textbf{Phase diagram and stability analysis for the mixed $(s, 1/2)$ model ($s=5$).} 
    The plot illustrates the correspondence between the stability criterion $|F'(1)|$ (blue), Kesten-Stigum bound $g_{KS}$ (green), and Dobrushin bound $D_{dob}$ (red). 
    The black points indicate the intersection with the stability threshold $y=1$ at $\varphi_1$ and $\varphi_2$. 
    The shaded regions represent the extremality (green) and nonextremality (red) phases. 
    The dashed vertical line at $\varphi=1$ marks the purely paramagnetic state.}
    \label{fig:stability_reconstruction}
\end{figure}

\medskip
\noindent\textbf{Putting the three languages side-by-side: phase transition $\neq$ reconstruction.}
Because (ii) and (iii) are only \emph{sufficient} criteria, they need not meet at a single
threshold. For $s=5$ this produces a clean ``separation of phenomena.”
\paragraph{Eight parameter regimes for $\varphi$ (ferromagnetic and antiferromagnetic sides, $s=5$).}
Combining both sides, we obtain the following eight regimes (see Fig.~\ref{fig:stability_reconstruction}):
\begin{enumerate}
\item[\textbf{(F1)}] \textbf{$1<\varphi<1.109560\ldots$}:
$\ell_0^{(5)}$ is attracting and Dobrushin applies, hence extremality/non-reconstruction holds.

\item[\textbf{(F2)}] \textbf{$1.109560\ldots<\varphi<1.180$}:
$\ell_0^{(5)}$ is repelling (multiplicity signal), yet Dobrushin still certifies extremality/non-reconstruction.

\item[\textbf{(F3)}] \textbf{$1.180<\varphi<1.2861$}:
$\ell_0^{(5)}$ is repelling; \emph{gray zone}: neither Dobrushin nor KS decides (sharper tools needed).

\item[\textbf{(F4)}] \textbf{$\varphi>1.2861$}:
KS certifies reconstruction (hence non-extremality/long-range memory), and in the same regime, $\ell_0^{(5)}$ is repelling, indicating an ongoing local instability consistent with a phase-transition mechanism.

\item[\textbf{(AF1)}] \textbf{$0.901258\ldots<\varphi<1$}:
$\ell_0^{(5)}$ is attracting and Dobrushin applies, hence extremality/non-reconstruction holds.

\item[\textbf{(AF2)}] \textbf{$0.847<\varphi<0.901258\ldots$}:
$\ell_0^{(5)}$ is repelling, yet Dobrushin still certifies extremality (anti-mirror of the ferromagnetic band).

\item[\textbf{(AF3)}] \textbf{$0.7776<\varphi<0.847$}:
$\ell_0^{(5)}$ is repelling, but this remains a \emph{gray zone} (undecided by these criteria).
: undecided based on sufficient criteria.

\item[\textbf{(AF4)}] \textbf{$0<\varphi<0.7776$}:
The KS certifies reconstruction/non-extremality (staggered long-range memory), and in the same range $\ell_0^{(5)}$ is also repelling.
\end{enumerate}

\medskip
\noindent\textbf{Critical ferromagnetic band.}
For $s=5$, the interaction parameter exhibits a distinguished ferromagnetic window
\[
1.109560\ldots<\varphi<1.180,
\]
in which a local instability has already set in, while the free-boundary disordered measure
remains extremal. In particular, this regime carries a phase-transition signal without
(reconstructible) long-range memory; see Fig.~\ref{fig:stability_reconstruction}.

\medskip
\noindent\textbf{Reciprocity check.}
The numerically observed thresholds on the antiferromagnetic side are, to high accuracy,
reciprocals of those on the ferromagnetic side (cf.\ Fig.~\ref{fig:stability_reconstruction}):
\[
0.901258\ldots \approx \frac{1}{1.109560\ldots},\qquad
0.7776\approx \frac{1}{1.2861},\qquad
0.847\approx \frac{1}{1.180}.
\]
This agrees with the $\varphi\mapsto 1/\varphi$ symmetry of the two-step (same-parity)
description and provides a strong internal consistency check that the ferro and antiferro
regimes were incorporated coherently into all three criteria.

\medskip
\noindent\textbf{Phase transition $\neq$ reconstruction for $s=5$ (and the antiferromagnetic mirror).}
For $s=5$, the loss of local stability of the symmetric fixed point (a \emph{local} phase-transition
signal) occurs strictly before the KS-certified reconstruction threshold. On the ferromagnetic side
($\varphi>1$), the three criteria produce the nested structure
\[
\underbrace{1<\varphi<\varphi_{\mathrm{stab}}^{+}}_{\substack{\ell_{0}^{(5)}\ \text{attracting}\\ \text{(locally stable disordered)}}}
\ \subset\
\underbrace{1<\varphi\le \varphi_{\mathrm{Dob}}^{+}}_{\substack{\text{Dobrushin certifies}\\ \text{extremal / non-reconstruction}}}
\ <\
\underbrace{\varphi_{\mathrm{Dob}}^{+}<\varphi<\varphi_{\mathrm{KS}}^{+}}_{\substack{\text{gap: neither criterion decides}\\ \text{(sharper tools needed)}}}
\ <\
\underbrace{\varphi>\varphi_{\mathrm{KS}}^{+}}_{\substack{\text{KS certifies}\\ \text{reconstruction / non-extremal}}},
\]
with numerical values
\[
\varphi_{\mathrm{stab}}^{+}=1.109560\ldots,\qquad
\varphi_{\mathrm{Dob}}^{+}\approx 1.180,\qquad
\varphi_{\mathrm{KS}}^{+}\approx 1.2861.
\]
Accordingly, the \emph{critical} ferromagnetic band
\[
\varphi_{\mathrm{stab}}^{+}<\varphi\le \varphi_{\mathrm{Dob}}^{+}
\quad (1.109560\ldots<\varphi\lesssim 1.180)
\]
is precisely the regime in which a phase-transition mechanism is already indicated
($\ell_0^{(5)}$ is repelling), yet extremality of the free-boundary disordered state is still
\emph{certified} (no reconstructible long-range memory).

On the antiferromagnetic side ($0<\varphi<1$), the picture mirrors under $\varphi\mapsto 1/\varphi$:
\[
\underbrace{\varphi<\varphi_{\mathrm{KS}}^{-}}_{\substack{\text{KS certifies}\\ \text{reconstruction / non-extremal}}}
\ <\
\underbrace{\varphi_{\mathrm{KS}}^{-}<\varphi<\varphi_{\mathrm{Dob}}^{-}}_{\substack{\text{gap: neither criterion decides}}}
\ <\
\underbrace{\varphi_{\mathrm{Dob}}^{-}\le \varphi<1}_{\substack{\text{Dobrushin certifies}\\ \text{extremal / non-reconstruction}}}
\ \supset\
\underbrace{\varphi_{\mathrm{stab}}^{-}<\varphi<1}_{\substack{\ell_{0}^{(5)}\ \text{attracting}}},
\]
where
\[
\varphi_{\mathrm{stab}}^{-}=0.901258\ldots,\qquad
\varphi_{\mathrm{Dob}}^{-}\approx 0.847,\qquad
\varphi_{\mathrm{KS}}^{-}\approx 0.7776.
\]
Thus, the antiferromagnetic counterpart of the band ``phase transition $\neq$ reconstruction'' is
\[
\varphi_{\mathrm{Dob}}^{-}\le \varphi<\varphi_{\mathrm{stab}}^{-}
\quad (0.847\lesssim \varphi<0.901258\ldots).
\]

\section{Markov Entropy Rate and Thermodynamic Disorder: A Dual Perspective}\label{sec:entropy}

In this section, we present the derivation of both the thermodynamic entropy and the specific (per-site) entropy for the mixed $(s,\tfrac12)$-spin Ising model defined on the CT. Our approach explicitly leverages the associated Markov chain representation of the system to express these entropy measures \cite{Reza1994}. Notably, we successfully obtained an explicit formula for the entropy rate at the symmetric fixed point for the particular case of the $(s,\tfrac12)$ model on a CT of order three, expressing the result as a function dependent solely on the thermal parameter, $\varphi$.

For the $(s,\tfrac12)$ mixed-spin Ising model on an $n$-shell CT, the
finite-volume Gibbs measure $\mu_n^{\mathbf{h}}(\xi_n)$ under a boundary
field $\mathbf{h}$ is given by the formula in~\eqref{Gibbs1}. The
finite-volume thermodynamic entropy in $V_n$ is defined by
\begin{equation}\label{eq:Sn-def}
S_n(\beta)
=
-\sum_{\xi_n}
\mu_n^{\mathbf{h}}(\xi_n)\,
\log \mu_n^{\mathbf{h}}(\xi_n).
\end{equation}
Substituting \eqref{Gibbs1} into \eqref{eq:Sn-def} yields
\begin{equation}\label{eq:Sn-explicit}
\begin{aligned}
S_n(\beta)
&=
-\sum_{\xi_n}\mu_n^{\mathbf{h}}(\xi_n)
\biggl[
-\log Z_n(\beta,\mathbf{h})
-\beta H_n(\xi_n)
+
\sum_{x\in W_n} h_{\xi(x)}(x)
\biggr]
\\[2mm]
&=
\log Z_n(\beta,\mathbf{h})
+
\beta\,\mathbb{E}_{\mu_n^{\mathbf{h}}}\bigl[H_n(\xi_n)\bigr]
-
\sum_{x\in W_n}
\mathbb{E}_{\mu_n^{\mathbf{h}}}\bigl[h_{\xi(x)}(x)\bigr].
\end{aligned}
\end{equation}
The specific (per-site) entropy is defined by the thermodynamic limit
\begin{equation*}\label{eq:entropy-density}
s(\beta)
=
\lim_{n\to\infty}\frac{S_n(\beta)}{|V_n|},
\end{equation*}
whenever the limit exists.

Introducing the free energy per site
\begin{equation*}\label{eq:free-energy-density}
f(\beta)
=
-\frac{1}{\beta}\,
\lim_{n\to\infty}\frac{1}{|V_n|}
\log Z_n(\beta,\mathbf{h}),
\end{equation*}
and the internal energy per site
\begin{equation*}\label{eq:energy-density}
e(\beta)
=
\lim_{n\to\infty}
\frac{1}{|V_n|}
\mathbb{E}_{\mu_n^{\mathbf{h}}}\bigl[H_n(\xi_n)\bigr],
\end{equation*}
one obtains the standard thermodynamic relation
\begin{equation*}\label{eq:s-ef}
s(\beta)
=
\beta\bigl(e(\beta)-f(\beta)\bigr),
\end{equation*}
with the Boltzmann constant $k_B=1$.

In our parametrisation $\varphi = e^{\frac{\beta J}{2}}$, differentiation
with respect to $\beta$ satisfies
\begin{align*}\label{eq:beta-rho-derivative}
\frac{\partial}{\partial \beta}
=
\frac{J}{2}\,\varphi\,
\frac{\partial}{\partial \varphi}.
\end{align*}
Hence, if $f(\varphi)$ is known explicitly as a function of $\varphi$, then
\begin{equation*}\label{eq:s-rho-form}
s(\beta)
=
\beta^2\frac{\partial f}{\partial\beta}
=
\beta^2\frac{J}{2}\,\varphi\,\frac{\partial f}{\partial\varphi}.
\end{equation*}

\subsubsection*{Entropy Perspectives: Thermodynamic vs. Information-Theoretic}

A fundamental distinction exists between thermodynamic entropy and entropy rate in terms of their scope and characterization. Thermodynamic entropy ($S$) serves as a global, static measure of total disorder within the system, conventionally derived from the free energy $F$ relative to temperature $T$ via the relation $S = -\frac{\partial F}{\partial T}$ (see \cite{Akin2025IJMPB,Akin23CJP}). In this context, the free energy is obtained from the partition function $Z$ of the model, representing the macroscopic state of the entire lattice. 

In contrast, the entropy rate ($h$) provides a local, process-oriented perspective within an information-theoretic framework and quantifies the average information produced by a stochastic process per step. For the mixed-spin $(s, 1/2)$ model on a CT, while the thermodynamic entropy characterizes the equilibrium state of the collective system, the entropy rate represents the growth rate of uncertainty as the ancestral signal propagates from the root to the leaves of the CT.

By employing the Markov entropy rate, we can bridge the gap between local spin interactions and the global extremality of the Gibbs measure. This approach provides a more rigorous quantification of information preservation across phase boundaries because it directly links the spectral properties of the transition matrices to the dynamic stability of the disordered phase.

\subsection{Entropy rate from the Markov-chain representation}\label{sec:entropy_rate_derivation}

For TISGMs, the model admits a Markov chain representation along the CT shells.
The shell-to-shell dynamics between consecutive layers, corresponding to the $(s,\tfrac12)$-spin
sublattices, are governed by the stochastic matrices $\mathbb{P}^{(s)}(\varphi)$ and
$\mathbb{Q}^{(s)}(\varphi)$ given in \eqref{eq:Matrix-P-arbitrary-s} and
\eqref{eq:Matrix-Q-arbitrary-s}. The induced $2\times 2$ transition matrix on the spin-$\tfrac12$ state space
$\Psi=\{-\tfrac12,\tfrac12\}$ is defined by
\begin{equation}\label{eq:R-def}
\mathbb{R}_\Psi^{(s)}(\varphi)
:=
\mathbb{Q}^{(s)}(\varphi)\,\mathbb{P}^{(s)}(\varphi)
=
\bigl(R^{(s)}_{ij}(\varphi)\bigr)_{i,j\in\Psi}.
\end{equation}

This construction yields a two-state time-homogeneous Markov chain along the CT shells. 
Its stationary distribution $\pi^{(s)}(\varphi)=(\pi^{(s)}_i(\varphi))_{i\in\Psi}$ is characterized by
\begin{equation}\label{eq:stationary-R}
\pi^{(s)}(\varphi)\,\mathbb{R}_\Psi^{(s)}(\varphi)=\pi^{(s)}(\varphi),
\qquad
\sum_{i\in\Psi}\pi^{(s)}_i(\varphi)=1.
\end{equation}
The (information-theoretic) entropy rate of this chain, denoted by $h_{\mathrm{MC}}^{(s)}(\varphi)$, is
\begin{equation}\label{eq:markov-entropy-rate}
h_{\mathrm{MC}}^{(s)}(\varphi)
=
-\sum_{i\in\Psi}\pi^{(s)}_i(\varphi)\sum_{j\in\Psi}
R^{(s)}_{ij}(\varphi)\,\log R^{(s)}_{ij}(\varphi).
\end{equation}

\paragraph{Two-step induced chain on the $\Psi$-layers.}
Since the shells alternate between $\Psi$- and $\Phi$-layers, the product
$\mathbb{R}^{(s)}_{\Psi}(\varphi)=\mathbb{Q}^{(s)}(\varphi)\,\mathbb{P}^{(s)}(\varphi)$
describes the induced two-step transition from one $\Psi$-shell to the next one.
For TISGMs the parameter $\varphi$ does not depend on the shell index, so
$\mathbb{R}^{(s)}_{\Psi}(\varphi)$ is time-homogeneous and admits the stationary distribution
$\pi^{(s)}(\varphi)$ in \eqref{eq:stationary-R}; the corresponding entropy rate is given by
\eqref{eq:markov-entropy-rate}.

If, more generally, $\varphi=\varphi_t$ varies with the shell index, the induced chain becomes
time-inhomogeneous. In that case one considers the stepwise conditional entropies
\[
h_t = H(X_{t+1}\mid X_t), 
\qquad 
p_{t+1}=p_t\,\mathbb{R}^{(s)}_{\Psi}(\varphi_t),
\]
and averages them over time, for instance via
\[
\overline{h}_n=\frac{1}{n}\sum_{t=0}^{n-1} h_t.
\]

\subsubsection{Parametrization and Explicit Formula}
Assume $a^{(s)}(\varphi),\, b^{(s)}(\varphi)\in[0,1]$ and $a^{(s)}(\varphi)+b^{(s)}(\varphi)>0$.
 We parametrize the $2\times 2$ transition matrix from $\Psi$ to $\Psi$ as
\[
\mathbb{R}_\Psi^{(s)}(\varphi)=
\begin{pmatrix}
1-a^{(s)}(\varphi) & a^{(s)}(\varphi)\\
b^{(s)}(\varphi) & 1-b^{(s)}(\varphi)
\end{pmatrix}.
\]
With this parametrization, the stationary distribution is
\[
\pi^{(s)}=
\left(
\frac{b^{(s)}(\varphi)}{a^{(s)}(\varphi)+b^{(s)}(\varphi)},
\;
\frac{a^{(s)}(\varphi)}{a^{(s)}(\varphi)+b^{(s)}(\varphi)}
\right).
\]
Substituting $\pi^{(s)}$ into \eqref{eq:markov-entropy-rate} yields the entropy rate in the explicit form
\begin{equation}\label{eq:markov-entropy-rate-ab}
\begin{aligned}
h_{\mathrm{MC}}^{(s)}(\varphi)
&=
-\frac{b^{(s)}(\varphi)}{a^{(s)}(\varphi)+b^{(s)}(\varphi)}
\Bigl[
\bigl(1-a^{(s)}(\varphi)\bigr)\log\bigl(1-a^{(s)}(\varphi)\bigr)
+
a^{(s)}(\varphi)\log a^{(s)}(\varphi)
\Bigr]
\\[1mm]
&\quad
-\frac{a^{(s)}(\varphi)}{a^{(s)}(\varphi)+b^{(s)}(\varphi)}
\Bigl[
b^{(s)}(\varphi)\log b^{(s)}(\varphi)
+
\bigl(1-b^{(s)}(\varphi)\bigr)\log\bigl(1-b^{(s)}(\varphi)\bigr)
\Bigr].
\end{aligned}
\end{equation}
Equivalently, using the binary entropy function $H_2$, we may write
\begin{equation}\label{eq:markov-entropy-rate-binary}
h_{\mathrm{MC}}^{(s)}(\varphi)
=
\frac{b^{(s)}(\varphi)}{a^{(s)}(\varphi)+b^{(s)}(\varphi)}\,H_2\!\bigl(a^{(s)}(\varphi)\bigr)
+
\frac{a^{(s)}(\varphi)}{a^{(s)}(\varphi)+b^{(s)}(\varphi)}\,H_2\!\bigl(b^{(s)}(\varphi)\bigr),
\end{equation}
where
\[
H_2(x) = -\bigl[(1-x)\log(1-x)+x\log x\bigr], \qquad 0\le x\le 1.
\]

\subsubsection{Connection to Specific Entropy}
On a $k$-regular CT, the asymptotic density of boundary nodes satisfies $|W_n|/|V_n|\to (k-1)/k$ as $n\to\infty$. Consequently, the contribution of the spin-$\tfrac12$ boundary process to the \textbf{specific entropy $s(\beta)$} is directly proportional to the entropy rate $h_{\mathrm{MC}}^{(s)}(\varphi)$ given in \eqref{eq:markov-entropy-rate}. Crucially, this entropy rate remains strictly positive along the non-extremality (reconstruction) region, which is defined by the KS condition. This positive value reflects the \textbf{persistence of information} regarding the boundary condition deep inside the tree structure.

\subsection{ Markov chain representation and entropy rate for $s=1$, $k=3$}

We now specialise to the $(1,\tfrac12)$ mixed-spin Ising model on the CT of order $k=3$. The spin sets are
\[
\Phi=\{-1,0,1\},
\qquad
\Psi=\Bigl\{-\tfrac12,\tfrac12\Bigr\},
\]
and we focus on the symmetric (disordered) fixed point of the corresponding
TISGMs. Introducing
\[
X_i=e^{U_i}=e^{h_i-h_0},\quad i=\pm1,
\qquad
Z=e^{R}=e^{\widetilde{h}_{\tfrac12}-\widetilde{h}_{-\tfrac12}},
\qquad
X_0=1,
\]
the dynamical system for $k=3$ takes the form
\begin{align}
X_{-1}
&=
\left(
\frac{\varphi^{2}+Z}{\varphi(1+Z)}
\right)^{3},
\label{ds-1-k3-a}
\\[3pt]
X_1
&=
\left(
\frac{1+\varphi^{2}Z}{\varphi(1+Z)}
\right)^{3},
\label{ds-1-k3-b}
\\[3pt]
Z
&=
\left(
\frac{X_{-1} + \varphi X_0 + \varphi^{2} X_{1}}
     {\varphi^{2}X_{-1} + \varphi X_0 + X_{1}}
\right)^{3},
\label{ds-1-k3-c}
\end{align}
with $X_0=1$. Eliminating $X_{\pm1}$ from \eqref{ds-1-k3-a}–\eqref{ds-1-k3-b}
and substituting into \eqref{ds-1-k3-c}, one obtains the one-dimensional
rational map
\begin{equation}\label{FZ-s1k3}
F(Z)
=
\left(
\frac{
\varphi^{4}(1+Z)^{3}
+
\varphi^{2}(1+\varphi^{2}Z)^{3}
+
(Z+\varphi^{2})^{3}
}{
\varphi^{4}(1+Z)^{3}
+
\varphi^{2}(Z+\varphi^{2})^{3}
+
(1+\varphi^{2}Z)^{3}
}
\right)^{3},
\end{equation}
so that the fixed points satisfy $F(Z)=Z$.  
The disordered phase corresponds to a symmetric fixed point, $Z=1$.

\subsubsection{One-step transition matrices.}
From \eqref{eq:Matrix-P-arbitrary-s}, the transition matrix from the spin-$1$ shell $\Phi$ to the spin-$\tfrac12$ shell
$\Psi$ is defined by
\begin{align}\label{eq:Stochastic-M-P-s=1}
\mathbb{P}^{(1)}(\varphi) =
\begin{pmatrix}
\dfrac{\varphi^{2}}{\varphi^{2}+1} & \dfrac{1}{\varphi^{2}+1} \\[4pt]
\dfrac{1}{2}                 & \dfrac{1}{2}         \\[4pt]
\dfrac{1}{\varphi^{2}+1}        & \dfrac{\varphi^{2}}{\varphi^{2}+1}
\end{pmatrix},  
\end{align}
where the rows are indexed by $i\in\{-1,0,1\}$ and the columns by
$j\in\{-\tfrac12,\tfrac12\}$.

Similarly, from \eqref{eq:Matrix-Q-arbitrary-s}, the transition matrix from the spin-$\tfrac12$ shell $\Psi$ to the spin-$1$ shell
$\Phi$ is defined by
\begin{equation}\label{eq:Stochastic-M-Q-s=1}
\mathbb{Q}^{(1)}(\varphi)
=
\frac{1}{\varphi^{8} + 4\varphi^{6} + 14\varphi^{4} + 4\varphi^{2} + 1}
\begin{pmatrix}
\varphi^{2}(\varphi^{2}+1)^{3} & 8\varphi^{4} & (\varphi^{2}+1)^{3} \\[4pt]
(\varphi^{2}+1)^{3} & 8\varphi^{4} & \varphi^{2}(\varphi^{2}+1)^{3}
\end{pmatrix},
\end{equation}
where rows correspond to $i=-\tfrac12,\tfrac12$ and columns to
$j=-1,0,1$.

\subsubsection{Markov chain representation and entropy rate for $1/2$-spin layer}

 Taking the matrices $\mathbb{P}^{(1)}(\varphi)$ and $\mathbb{Q}^{(1)}(\varphi)$ given in \eqref{eq:Stochastic-M-P-s=1} and \eqref{eq:Stochastic-M-Q-s=1} 
and forming the product $\mathbb{R}_\Psi^{(1)}(\varphi)=\mathbb{Q}(\varphi)\mathbb{P}(\varphi)$, we find
that the effective $2\times2$ transition matrix on $\Psi$ takes the form
\begin{equation}\label{eq:R-rho-final}
\mathbb{R}_\Psi^{(1)}(\varphi)
=
\begin{pmatrix}
p(\varphi) & 1-p(\varphi) \\
1-p(\varphi) & p(\varphi)
\end{pmatrix},
\end{equation}
where
\begin{equation}\label{eq:p-rho-final}
p(\varphi)
=
\frac{A(\varphi)}{D(\varphi)},
\qquad
A(\varphi) = \varphi^{6} + \varphi^{4} + 2\varphi^{3} + \varphi^{2} + 1,
\qquad
D(\varphi) = 4\varphi^{3} + (\varphi^{2}+1)^{3}.
\end{equation}
Since $0<p(\varphi)<1$ for all $\varphi>0$ and $\mathbb{R}_\Psi^{(1)}(\varphi)$ is symmetric and
stochastic, the stationary distribution is
\[
\pi = \Bigl(\tfrac12,\tfrac12\Bigr).
\]

Using \eqref{eq:markov-entropy-rate-ab}, the entropy rate of the effective Markov chain at
the symmetric fixed point becomes
\begin{equation}\label{eq:hMC-rho-final}
h_{\mathrm{MC}}^{(1)}(\varphi)
=
-\Bigl[
p(\varphi)\,\log p(\varphi)
+
\bigl(1-p(\varphi)\bigr)\,\log\bigl(1-p(\varphi)\bigr)
\Bigr],
\end{equation}
with $p(\varphi)$ given by \eqref{eq:p-rho-final}. This provides an explicit
$\varphi$–dependent expression for the entropy rate at the disordered phase of the
$(1,\tfrac12)$ mixed-spin Ising model on the CT of order three.

\begin{figure}[!htbp]
\centering
\includegraphics[width=80mm]{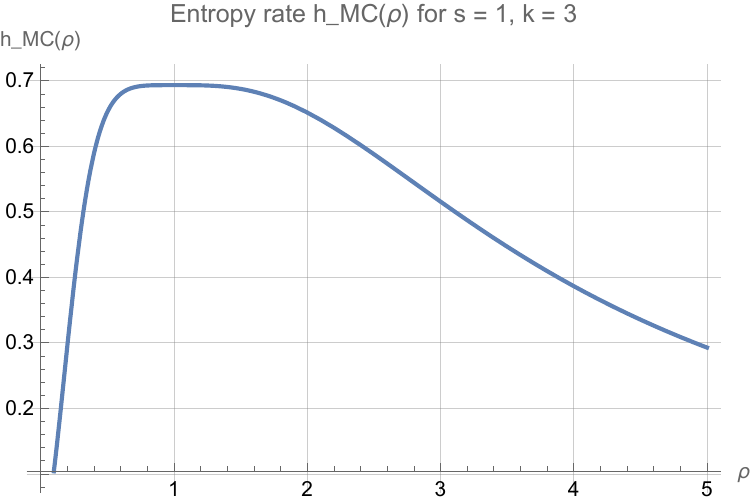}
\caption{(Color online) Entropy rate $h_{\mathrm{MC}}^{(1)}(\varphi)$ for the mixed
spin-$(1,\tfrac12)$ Ising model on a CT of order three at the symmetric
fixed point.}
\label{fig:Entropy1a}
\end{figure}

The dependence of the Markov entropy rate $h_{\mathrm{MC}}^{(1)}(\varphi)$ on $\varphi$
for the mixed spin-$(1,\tfrac12)$ Ising model on a Cayley tree of order $k=3$ is shown
in Fig.~\ref{fig:Entropy1a}. The computation is carried out at the symmetric fixed point
($Z=1$), corresponding to the disordered phase.

Plotted in natural units against the thermal parameter $\varphi=e^{\beta J/2}$, the curve
quantifies the average uncertainty produced per shell-to-shell step in the induced
two-state Markov chain. In particular, it provides an information-theoretic proxy for the
degree of thermodynamic disorder and highlights how changes in the thermal environment
modulate information production. This perspective is useful for assessing the stability of
the disordered phase and for isolating the contribution of the $s=1$ component to the
resulting entropy rate.

\subsubsection{Effective Markov chain on the spin-$1$ layer and entropy rate for $s=1$, $k=3$}\label{subsec:MC-s1k3}

Starting from the one-step stochastic matrices $\mathbb{P}(\varphi)$ and $\mathbb{Q}(\varphi)$ in
\eqref{eq:Stochastic-M-P-s=1} and \eqref{eq:Stochastic-M-Q-s=1}, the induced two-step transition
matrix on the spin-$1$ layer $\Phi=\{-1,0,1\}$ is
\[
\widetilde{\mathbb{R}}_\Phi^{(1)}(\varphi)
:=
\mathbb{P}(\varphi)\,\mathbb{Q}(\varphi)
\in\mathbb{R}^{3\times 3}.
\]
A direct computation yields
\begin{equation}\label{eq:Rphi-rho-correct}
\widetilde{\mathbb{R}}_\Phi^{(1)}(\varphi)
=
\frac{1}{\Delta(\varphi)}
\begin{pmatrix}
a(\varphi) & b(\varphi) & c(\varphi) \\[2pt]
d(\varphi) & b(\varphi) & d(\varphi) \\[2pt]
c(\varphi) & b(\varphi) & a(\varphi)
\end{pmatrix},
\end{equation}
where
\begin{align}
\Delta(\varphi)
&=
\varphi^{8} + 4\varphi^{6} + 14\varphi^{4} + 4\varphi^{2} + 1,
\label{eq:Delta-rho-correct}\\[4pt]
a(\varphi)
&=
(\varphi^{2}+1)^{2}(\varphi^{4}+1),
\label{eq:a-rho-correct}\\[4pt]
b(\varphi)
&=
8\varphi^{4},
\label{eq:b-rho-correct}\\[4pt]
c(\varphi)
&=
2\varphi^{2}(\varphi^{2}+1)^{2},
\label{eq:c-rho-correct}\\[4pt]
d(\varphi)
&=
\frac{(\varphi^{2}+1)^{4}}{2}.
\label{eq:d-rho-correct}
\end{align}
Moreover,
\[
a(\varphi)+b(\varphi)+c(\varphi)=\Delta(\varphi),
\qquad
2d(\varphi)+b(\varphi)=\Delta(\varphi),
\]
so each row of $\widetilde{\mathbb{R}}_\Phi^{(1)}(\varphi)$ sums to $1$, i.e.\ it is row-stochastic.

\medskip
\noindent\textbf{Stationary distribution.}
By symmetry we have $\pi_{-1}(\varphi)=\pi_{1}(\varphi)$ for the stationary distribution
$\pi(\varphi)=(\pi_{-1}(\varphi),\pi_{0}(\varphi),\pi_{1}(\varphi))$ on $\Phi$.
Solving $\pi(\varphi)\widetilde{\mathbb{R}}_\Phi^{(1)}(\varphi)=\pi(\varphi)$ together with
$\sum_{i\in\Phi}\pi_i(\varphi)=1$ gives
\begin{equation}\label{eq:pi-rho-correct}
\pi_0(\varphi)
=
\frac{8\varphi^{4}}{\Delta(\varphi)},
\qquad
\pi_{\pm1}(\varphi)
=
\frac{\varphi^{8} + 4\varphi^{6} + 6\varphi^{4} + 4\varphi^{2} + 1}
     {2\,\Delta(\varphi)}.
\end{equation}

\medskip
\noindent\textbf{Entropy rate.}
Let $H_i(\varphi)$ denote the Shannon entropy of the $i$-th row of
$\widetilde{\mathbb{R}}_\Phi^{(1)}(\varphi)$ (for $i\in\Phi$):
\[
H_i(\varphi)
=
-\sum_{j\in\Phi} \widetilde{R}^{(1)}_{\Phi,ij}(\varphi)\,
\log \widetilde{R}^{(1)}_{\Phi,ij}(\varphi).
\]
From \eqref{eq:Rphi-rho-correct} we obtain
\begin{align}
H_{-1}(\varphi)
&=
-\biggl[
\frac{a(\varphi)}{\Delta(\varphi)}\log\frac{a(\varphi)}{\Delta(\varphi)}
+
\frac{b(\varphi)}{\Delta(\varphi)}\log\frac{b(\varphi)}{\Delta(\varphi)}
+
\frac{c(\varphi)}{\Delta(\varphi)}\log\frac{c(\varphi)}{\Delta(\varphi)}
\biggr],
\label{eq:Hminus1-correct}\\[4pt]
H_0(\varphi)
&=
-\biggl[
\frac{b(\varphi)}{\Delta(\varphi)}\log\frac{b(\varphi)}{\Delta(\varphi)}
+
2\,\frac{d(\varphi)}{\Delta(\varphi)}\log\frac{d(\varphi)}{\Delta(\varphi)}
\biggr],
\label{eq:H0-correct}
\end{align}
and by symmetry $H_1(\varphi)=H_{-1}(\varphi)$. Therefore, the entropy rate of the induced Markov
chain on the spin-$1$ layer is
\begin{equation}\label{eq:hMC-Rphi-correct}
\widetilde{h}_{\mathrm{MC}}^{(1)}(\varphi)
=
2\,\pi_{-1}(\varphi)\,H_{-1}(\varphi)
+
\pi_0(\varphi)\,H_0(\varphi),
\end{equation}
with $\pi$ from \eqref{eq:pi-rho-correct} and $a,b,c,d,\Delta$ given in
\eqref{eq:Delta-rho-correct}--\eqref{eq:d-rho-correct}. The resulting function
$h_{\mathrm{MC}}^{(1)}(\varphi)$ is plotted in Fig.~\ref{fig:Entropy1b}.

\begin{figure}[!htbp]
\centering
\includegraphics[width=80mm]{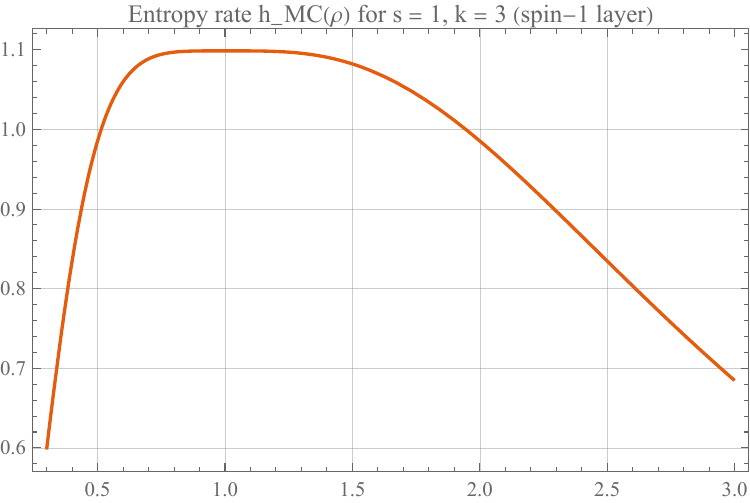}
\caption{(Color online) Entropy rate $\widetilde{h}_{\mathrm{MC}}^{(1)}(\varphi)$ from
\eqref{eq:hMC-Rphi-correct} for the mixed spin-$(1,\tfrac12)$ Ising model on a CT of order three
($k=3$) at the symmetric fixed point.}
\label{fig:Entropy1b} 
\end{figure}

\subsection{Markov chain representation and entropy rate for $s=2$, $k=3$}

For $s=2$ and $k=3$, the effective two-step transition along a ray of the
CT is governed by the $2\times 2$ stochastic matrix
\begin{equation}\label{eq:M-rho-matrix-s2-again}
\mathbb{R}_\Psi^{(2)}(\varphi)
=
\frac{1}{D(\varphi)}
\begin{pmatrix}
A(\varphi) & B(\varphi) \\
B(\varphi) & A(\varphi)
\end{pmatrix},
\end{equation}
where $A(\varphi)$, $B(\varphi)$ and $D(\varphi)$ are given by
\begin{align}
A(\varphi)
&=
\varphi^{16} + 3\varphi^{12} + 2\varphi^{10} + 8\varphi^{8}
+ 2\varphi^{6} + 3\varphi^{4} + 1,
\label{eq:A-rho-s2-again}\\[4pt]
B(\varphi)
&=
2\varphi^{4}\bigl(\varphi^{8} + \varphi^{6} + 6\varphi^{4} + \varphi^{2} + 1\bigr),
\label{eq:B-rho-s2-again}\\[4pt]
D(\varphi)
&=
\varphi^{16} + 5\varphi^{12} + 4\varphi^{10} + 20\varphi^{8}
+ 4\varphi^{6} + 5\varphi^{4} + 1,
\label{eq:D-rho-s2-again}
\end{align}
and satisfy
\begin{equation}\label{eq:D-equals-A-plus-B-s2-again}
D(\varphi) = A(\varphi) + B(\varphi).
\end{equation}
The state space of this effective Markov chain is \( \Psi = \Bigl\{-\tfrac12,\tfrac12\Bigr\},\)
and the rows of $\mathbb{R}_\Psi^{(2)}(\varphi)$ give the transition probabilities
between successive spin-$\tfrac12$ levels along a ray.

\paragraph{Stationary distribution.}
Since $\mathbb{R}_\Psi^{(2)}(\varphi)$ is symmetric and of the form
\[
\mathbb{R}_\Psi^{(2)}(\varphi)
=
\begin{pmatrix}
p(\varphi) & 1-p(\varphi) \\
1-p(\varphi) & p(\varphi)
\end{pmatrix},
\qquad
p(\varphi) := \frac{A(\varphi)}{D(\varphi)},
\]
the two rows are identical. Hence the stationary distribution is
uniform:
\begin{equation}\label{eq:pi-s2-k3}
\pi = \Bigl(\tfrac12,\tfrac12\Bigr),
\qquad
\pi\,\mathbb{R}_\Psi^{(2)}(\varphi) = \pi.
\end{equation}

\paragraph{Entropy rate}
The entropy rate of this two-state Markov chain is defined by
\begin{equation}\label{eq:hMC-def-s2-k3}
h_{\mathrm{MC}}^{(2)}(\varphi)
=
-\sum_{i\in\Psi}\pi_i
\sum_{j\in\Psi} R_{ij}(\varphi)\,\log R_{ij}(\varphi).
\end{equation}
Using \eqref{eq:pi-s2-k3} and the fact that both rows of
$\mathbb{R}_\Psi^{(2)}(\varphi)$ coincide with $(p(\varphi),1-p(\varphi))$, we obtain
\begin{equation}\label{eq:hMC-binary-s2-k3}
h_{\mathrm{MC}}^{(2)}(\varphi)
=
-\Bigl[
p(\varphi)\,\log p(\varphi)
+
\bigl(1-p(\varphi)\bigr)\,\log\bigl(1-p(\varphi)\bigr)
\Bigr],
\end{equation}
where
\begin{equation}\label{eq:p-of-rho-s2-k3}
p(\varphi)
=
\frac{A(\varphi)}{D(\varphi)}
=
\frac{
\varphi^{16} + 3\varphi^{12} + 2\varphi^{10} + 8\varphi^{8}
+ 2\varphi^{6} + 3\varphi^{4} + 1
}{
\varphi^{16} + 5\varphi^{12} + 4\varphi^{10} + 20\varphi^{8}
+ 4\varphi^{6} + 5\varphi^{4} + 1
}.
\end{equation}
Thus $h_{\mathrm{MC}}^{(2)}(\varphi)$ is just the binary entropy of $p(\varphi)$.
It satisfies $0 \le h_{\mathrm{MC}}^{(2)}(\varphi) \le \log 2$, attaining its
minimum $0$ in the ordered limits $\varphi\to 0^{+}$ and
$\varphi\to +\infty$, and its maximum $\log 2$ at the point where
$p(\varphi)=1/2$.
\begin{figure}[!htbp]
\centering
\includegraphics[width=80mm]{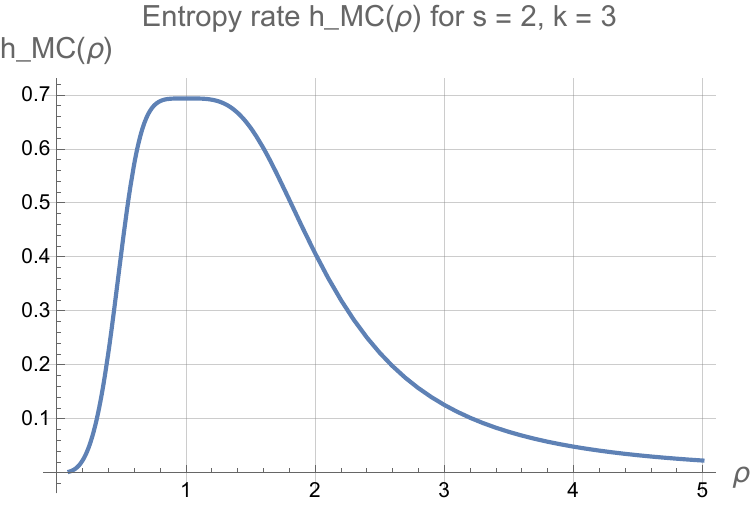}
\caption{(Color online) Entropy rate $h_{\mathrm{MC}}(\varphi)$ for the mixed
spin-$(2,\tfrac12)$ Ising model on a CT of order three at the symmetric
fixed point.}
\label{fig:Entropy-s=2}
\end{figure}

\subsection{Entropy rate for $s=5$ and $k=3$}

Considering the stochastic matrix given in \eqref{app-eq:Stochastis-M-R-s=5}, for the mixed $(5,\tfrac12)$–spin Ising model on a CT of order
$k=3$, at the symmetric fixed point $\ell_0^{(5)}$ the effective two–step
transition matrix on the spin–$\tfrac12$ states has the form
\[
\mathbb{R}^{(5)}(\varphi)
=
\frac{1}{D(\varphi)}
\begin{pmatrix}
A(\varphi) & B(\varphi)\\[2pt]
B(\varphi) & A(\varphi)
\end{pmatrix},
\qquad \varphi = e^{\beta J/2}>0,
\]
where
\begin{align*}
A(\varphi)
&=
\varphi^{40}
+\varphi^{36}
+\varphi^{32}
+2\varphi^{30}
+3\varphi^{28}
+2\varphi^{26}
+3\varphi^{24}
+2\varphi^{22}
+14\varphi^{20}
\\[-2pt]
&\hspace{2em}
+2\varphi^{18}
+3\varphi^{16}
+2\varphi^{14}
+3\varphi^{12}
+2\varphi^{10}
+\varphi^{8}
+\varphi^{4}
+1,
\\[4pt]
B(\varphi)
&=
2\varphi^{30}
+2\varphi^{28}
+2\varphi^{26}
+2\varphi^{24}
+2\varphi^{22}
+24\varphi^{20}
+2\varphi^{18}
+2\varphi^{16}
+2\varphi^{14}
+2\varphi^{12}
+2\varphi^{10},
\\[4pt]
D(\varphi)
&=
A(\varphi)+B(\varphi)
\\
&=
\varphi^{40}
+\varphi^{36}
+\varphi^{32}
+4\varphi^{30}
+5\varphi^{28}
+4\varphi^{26}
+5\varphi^{24}
+4\varphi^{22}
+38\varphi^{20}
\\[-2pt]
&\hspace{2em}
+4\varphi^{18}
+5\varphi^{16}
+4\varphi^{14}
+5\varphi^{12}
+4\varphi^{10}
+\varphi^{8}
+\varphi^{4}
+1.
\end{align*}
Since $D(\varphi)=A(\varphi)+B(\varphi)$, the matrix $\mathbb{R}^{(5)}(\varphi)$ is
symmetric and stochastic. Writing
\[
p(\varphi) := \frac{A(\varphi)}{D(\varphi)},
\]
we can rewrite $\mathbb{R}^{(5)}(\varphi)$ as
\begin{equation}\label{eq:M-rho-s5-2state}
\mathbb{R}^{(5)}(\varphi)
=
\begin{pmatrix}
p(\varphi) & 1-p(\varphi)\\[2pt]
1-p(\varphi) & p(\varphi)
\end{pmatrix}.
\end{equation}
In particular, $0<p(\varphi)<1$ for all $\varphi>0$, and the stationary
distribution is symmetric,
\[
\pi = \Bigl(\tfrac12,\tfrac12\Bigr).
\]

The entropy rate of this two–state Markov chain is therefore
\begin{equation}\label{eq:hMC-s5-def}
h_{\mathrm{MC}}^{(5)}(\varphi)
=
-\sum_{i\in\{-\tfrac12,\tfrac12\}}\pi_i
\sum_{j\in\{-\tfrac12,\tfrac12\}} R_{ij}(\varphi)\,\log R_{ij}(\varphi).
\end{equation}
Substituting \eqref{eq:M-rho-s5-2state} and $\pi=(\tfrac12,\tfrac12)$
into \eqref{eq:hMC-s5-def} gives the explicit $\varphi$–dependent formula
\begin{equation}\label{eq:hMC-s5-final}
h_{\mathrm{MC}}^{(5)}(\varphi)
=
-\Bigl[
p(\varphi)\,\log p(\varphi)
+
\bigl(1-p(\varphi)\bigr)\,
\log\bigl(1-p(\varphi)\bigr)
\Bigr],
\end{equation}
with
\[
p(\varphi) = \frac{A(\varphi)}{D(\varphi)}
\]
as defined above. This quantity is proportional to the specific entropy
of the translation–invariant splitting Gibbs measure corresponding to
the disordered phase of the $(5,\tfrac12)$–mixed spin Ising model on the
CT of order three.

\begin{figure}[!htbp]
\centering
\includegraphics[width=0.5\linewidth]{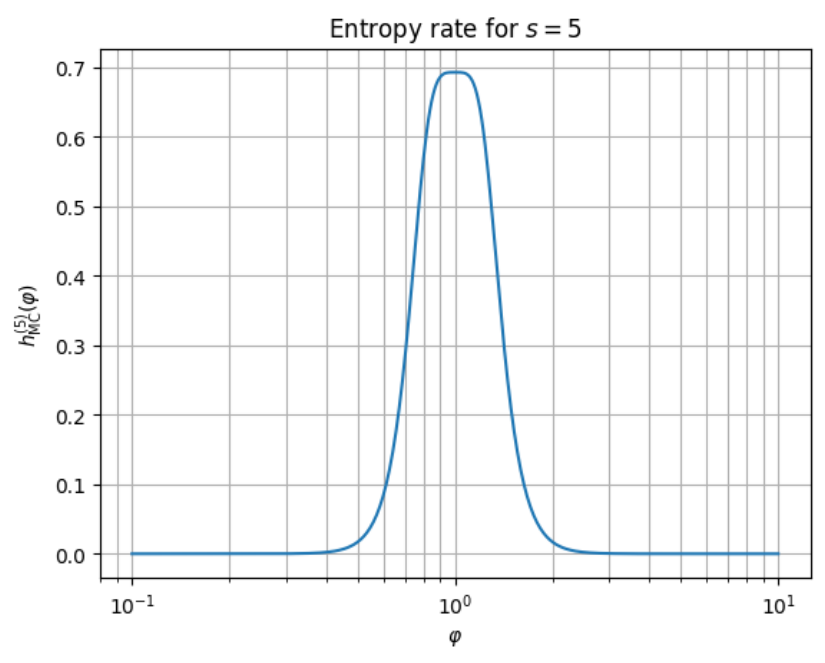}
\caption{(Color online) Entropy rate $h_{\mathrm{MC}}^{(5)}(\varphi)$ given in \eqref{eq:hMC-s5-final} for the mixed
spin-$(5,\tfrac12)$ Ising model on a CT of order three at the symmetric
fixed point.}
\label{fig:Entropy-s=5}
\end{figure}

\subsection{Entropy Rate Test}\label{subset:Algorithma2-entropy}
For clarity, Algorithm~\ref{alg:entropy_rate} provides a summary of the entropy rate test procedure.
\begin{algorithm}[!htbp]
\caption{Computation of the Markov entropy rate $h_{\mathrm{MC}}^{(s)}(\varphi)$}
\label{alg:entropy_rate}
\begin{algorithmic}[1]
\Require Spin quantum number $s\in\mathbb{Z}^+$, thermal parameter $\varphi\in\mathbb{R}^+$.
\Ensure Markov entropy rate $h_{\mathrm{MC}}^{(s)}(\varphi)$.

\State \textbf{Auxiliary quantities:}
\State $h_n(\varphi)\gets \dfrac{(\varphi^{2n}+1)^3}{8\,\varphi^{2n}}$,\quad
      $g_n(\varphi)\gets \dfrac{(\varphi^{2n}+1)^3}{8\,\varphi^{4n}}$
\State $S_s(\varphi)\gets 1+\dfrac{1}{8}\displaystyle\sum_{n=1}^{s}\dfrac{(\varphi^{2n}+1)^4}{\varphi^{4n}}$

\vspace{0.2em}
\State \textbf{Construct the matrices $P^{(s)}(\varphi)$ and $Q^{(s)}(\varphi)$:}
\State Initialize $P^{(s)}\in\mathbb{R}^{(2s+1)\times 2}$
\For{$i=-s$ \textbf{to} $s$}
  \State $P^{(s)}_{i,\, -1/2}\gets \dfrac{1}{1+\varphi^{2i}}$
  \State $P^{(s)}_{i,\, +1/2}\gets \dfrac{\varphi^{2i}}{1+\varphi^{2i}}$
\EndFor

\State Form the unnormalized row vectors (length $2s+1$):
\State $\mathbf{q}_1\gets \big(h_s,\ldots,h_1,1,g_1,\ldots,g_s\big)$
\State $\mathbf{q}_2\gets \big(g_s,\ldots,g_1,1,h_1,\ldots,h_s\big)$
\State Define $Q^{(s)}\in\mathbb{R}^{2\times(2s+1)}$ by
\State $Q^{(s)}(\varphi)\gets \dfrac{1}{S_s(\varphi)}
\begin{pmatrix}
\mathbf{q}_1\\[0.2em]
\mathbf{q}_2
\end{pmatrix}$

\vspace{0.2em}
\State \textbf{Effective (two-state) transition matrix:}
\State $M^{(s)}(\varphi)\gets Q^{(s)}(\varphi)\,P^{(s)}(\varphi)\in\mathbb{R}^{2\times 2}$
\vspace{0.2em}
\State \textbf{Stationary distribution:}
\State Find $\boldsymbol{\pi}^{(s)}$ such that $\boldsymbol{\pi}^{(s)} M^{(s)}=\boldsymbol{\pi}^{(s)}$ and $\sum_{i=1}^{2}\pi^{(s)}_i=1$
\vspace{0.2em}
\State \textbf{Markov entropy rate:}
\State $h_{\mathrm{MC}}^{(s)}(\varphi)\gets 0$
\For{$i=1$ \textbf{to} $2$}
  \For{$j=1$ \textbf{to} $2$}
    \If{$M^{(s)}_{ij}(\varphi)>0$}
      \State $h_{\mathrm{MC}}^{(s)}(\varphi)\gets
      h_{\mathrm{MC}}^{(s)}(\varphi)-\pi^{(s)}_i\,M^{(s)}_{ij}(\varphi)\,\log\!\big(M^{(s)}_{ij}(\varphi)\big)$
    \EndIf
  \EndFor
\EndFor
\State \Return $h_{\mathrm{MC}}^{(s)}(\varphi)$
\end{algorithmic}
\end{algorithm}

\subsection{Graphical analysis of entropy rate for different $s$-spin values}
By showing the explicit $\varphi$-dependence for these distinct values of $s$, the plot effectively demonstrates how the intrinsic spin magnitude affects the long-term thermodynamic disorder (or information content) of the system in the disordered phase. Analyzing this family of curves allows us to observe the specific influence of the high-spin component ($s$) on the overall uncertainty generated per step along the tree shells. Furthermore, the plot provides crucial insights into the phase stability limits and the boundary of the reconstruction region (defined by the KS condition) as the spin magnitude $s$ is varied, revealing how the system's ability to retain boundary information changes with increasing spin complexity. 
\begin{figure}[!htbp]
\centering
\includegraphics[width=90mm]{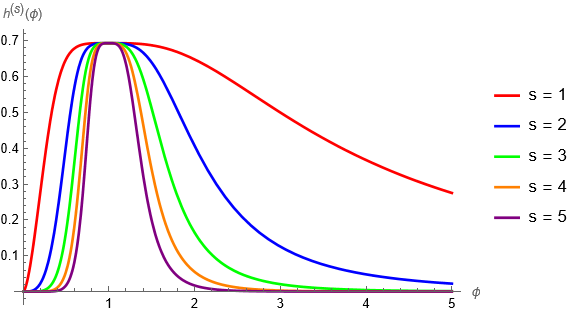}
\caption{(Color online) Entropy rate $h_{\mathrm{MC}}^{(s)}(\varphi)$ for the mixed
spin-$(s,\tfrac12)$ Ising model on a CT of order three at the symmetric
fixed point, for $s=1,2,3,4,5$.}
\label{fig:Entropy-s-1-5}
\end{figure}

The Markov entropy rate $h_{\mathrm{MC}}^{(s)}(\varphi)$ is computed using the procedure outlined in Algorithm~\ref{alg:entropy_rate}. This involves determining the transition matrices $P^{(s)}(\varphi)$ and $Q^{(s)}(\varphi)$ based on the spin state and thermal parameter, multiplying them to obtain the 2-state Markov transition matrix $M^{(s)}(\varphi)$, and subsequently solving for the stationary distribution $\pi^{(s)}(\varphi)$. Finally, the entropy rate is calculated using the standard formula involving $\pi^{(s)}(\varphi)$ and $M^{(s)}(\varphi)$. 

Figure \ref{fig:Entropy-s-1-5} presents the calculated Markov entropy rate, $h_{\mathrm{MC}}^{(s)}(\varphi)$, for the mixed spin-$(s,\tfrac12)$ Ising model (MSIM) on the CT with coordination number $k=3$. The figure specifically illustrates the behavior of $h_{\mathrm{MC}}^{(s)}(\varphi)$ at the symmetric (disordered) fixed point across an increasing sequence of spin quantum numbers, $s$, ranging from $s=1$ to $s=5$. 

In Figure \ref{fig:polar_entropy_s1_s6}, the entropy rate is visualized using a polar coordinate system to illustrate its dependence on both the thermal parameter and spin magnitude. Specifically: The angular coordinate is parametrized by $\varphi = e^{\beta J/2}$, where $\beta$ denotes the inverse temperature ($\beta = 1/(k_B T)$) and $J$ is the exchange coupling constant. The radial coordinate directly represents the Markov entropy rate $h_{\mathrm{MC}}^{(s)}(\varphi)$, expressed in natural units. The distinct curves correspond to increasing spin quantum numbers, $s$, ranging from $s=1$ to $s=6$. This polar representation effectively demonstrates the interplay between the thermal effects (controlled by $\varphi$) and the intrinsic spin magnitude $s$ on the overall thermodynamic disorder of the system.

\begin{figure}[!htbp]
\centering
\includegraphics[width=0.4\textwidth]{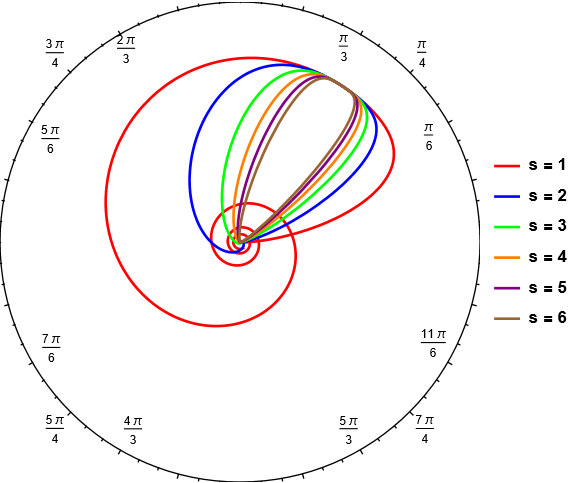}
\caption{Polar plot of the Markov entropy rate $h_{\mathrm{MC}}^{(s)}(\varphi)$ for the mixed-spin $(s, 1/2)$ Ising model on a CT with coordination number $k=3$ for spin quantum numbers $s = 1, 2, \dots, 6$.}
\label{fig:polar_entropy_s1_s6}
\end{figure}

\subsection*{Rationale for Coordinate Choice}

This visualization choice is deliberate, as representing the function in polar coordinates significantly \textbf{eases the comprehension of its underlying behavior} compared with conventional orthogonal (Cartesian) coordinates.

\subsection*{Tracking Convergence and Symmetry}
By expanding the domain of the visualization, we can effectively \textbf{track the number of cycles} required for the function $h_{\mathrm{MC}}^{(s)}(\varphi)$ to ultimately reach zero.

The key observations from the figure are as follows:
\begin{itemize}
    \item \textbf{Symmetry Development:} As the spin value, $s$, increases ($s=1, 2, 3, 4, 5, 6$), the function's behavior progressively becomes more \textbf{symmetrical with respect to the axis $\varphi=1$}.
    \item \textbf{Convergence for $s=1$:} Specifically, for the case where the spin value is $s=1$, the function's convergence to zero is observed to occur precisely \textbf{after 4 cycles}.
\end{itemize}

\subsection{Interpretation and parameter ranges}\label{sec:interp_param_ranges}

\subsubsection*{Temperature parametrization and entropy-scale bounds}

Both entropy rates are bounded by the logarithm of the effective state-space size of the corresponding induced Markov chain. 
On the spin-$\tfrac12$ layer ($\Psi$), the chain has two states, hence
\[
0\le h_{\mathrm{MC}}^{(s)}(\varphi)\le \log 2 \quad \text{for all } s,
\]
whereas on the spin-$s$ layer ($\Phi$) the induced chain has $2s+1$ states and therefore
\[
0\le \widetilde{h}_{\mathrm{MC}}^{(s)}(\varphi)\le \log(2s+1),
\]
which increases with $s$. At $\varphi=1$ (the unbiased/high-temperature point) the transitions become uniform and the entropy rates attain their maximal values, while for $\varphi\to 0$ or $\varphi\to\infty$
the dynamics become increasingly deterministic and both rates tend to $0$. For illustrations, see
Fig.~\ref{fig:Entropy1b} for $\widetilde{h}_{\mathrm{MC}}^{(s)}(\varphi)$, and
Figs.~\ref{fig:Entropy1a}, \ref{fig:Entropy-s=2}, \ref{fig:Entropy-s=5}, and \ref{fig:Entropy-s-1-5}
for $h_{\mathrm{MC}}^{(s)}(\varphi)$.

We use the temperature parametrization
\[
\varphi(T)=\exp\!\left(\frac{J}{2T}\right),\qquad T>0.
\]
For $J>0$, $T\to\infty$ implies $\varphi(T)\to 1^{+}$ and $T\to 0^{+}$ implies $\varphi(T)\to\infty$,
so the physically accessible regime is $\varphi\ge 1$. For $J<0$, $T\to\infty$ implies
$\varphi(T)\to 1^{-}$ and $T\to 0^{+}$ implies $\varphi(T)\to 0^{+}$, so the physically accessible
regime is $\varphi\le 1$. In particular, under $J>0$ the portion $\varphi<1$ in plots over intervals such
as $\varphi\in(0.01,5)$ is not reachable in the temperature parametrization.

\subsubsection*{Why $QP$ is uniformly bounded while $PQ$ grows with $s$}

The coarse-grained chain $\mathbb{R}_\Psi^{(s)}(\varphi)=\mathbb{Q}^{(s)}(\varphi)\mathbb{P}^{(s)}(\varphi)$
has $|\Psi|=2$ states, whereas the lifted chain
$\mathbb{R}_\Phi^{(s)}(\varphi)=\mathbb{P}^{(s)}(\varphi)\mathbb{Q}^{(s)}(\varphi)$ has $|\Phi|=2s+1$ states.
Since the entropy rate of an $N$-state Markov chain satisfies $h\le \log N$ (in nats), it follows that
\[
h_{QP}(s,\varphi)\le \log 2,
\qquad
h_{PQ}(s,\varphi)\le \log(2s+1).
\]
Thus $h_{QP}$ is uniformly bounded by $\log 2\simeq 0.693$, while the natural entropy scale for $h_{PQ}$
grows like $\log(2s+1)$ as $s$ increases, which simply reflects coarse-graining versus refinement.

\subsubsection*{High-temperature interpretation of the peak near $\varphi\simeq 1$}

Under $\varphi=\exp\!\bigl(J/(2T)\bigr)$, the regime $\varphi\simeq 1$ corresponds to $J/(2T)\simeq 0$,
that is, a large $T$ (high temperature). In this region transition probabilities are typically more spread out
(closer to uniform); therefore, the entropy rate approaches its upper bound. As $T$ decreases (moving $\varphi$
away from $1$), transitions concentrate and metastable structure may emerge; correspondingly the spectral
gap may shrink, consistent with a phase-separation/extremality narrative in large-$s$ limits.

\section{Coupling--temperature control and the information--physics bridge}
\label{subsec:bridge}

We encode the competition between interaction strength $J$ and thermal noise $T$ by the
dimensionless parameter
\[
\varphi \;=\; e^{\beta J/2} \;=\; \exp\!\Big(\frac{J}{2T}\Big),
\qquad (k_B=1).
\]
Thus $\varphi\to 1$ corresponds to the high--temperature (noise--dominated) regime, while
deviations from $\varphi=1$ represent low temperature/strong coupling (large $|J|/T$).
The symmetry $\varphi\mapsto 1/\varphi$ treats ferromagnetic ($J>0$, $\varphi>1$) and
antiferromagnetic ($J<0$, $\varphi<1$) regimes in a unified way.

On the rooted CT with branching number $k$ (in our setting $k=3$), the same
control parameter $\varphi$ governs both thermodynamic behavior (stability and phase
transitions) and \emph{information flow} across generations. In particular, at the
disordered (free-boundary) fixed point $Z=1$, the bipartite mixed-spin specification
induces two alternating edge-channels $\Phi\to\Psi$ and $\Psi\to\Phi$ with kernels
$\mathbb{P}^{(s)}(\varphi)$ and $\mathbb{Q}^{(s)}(\varphi)$. Along each ray, the induced
\emph{two-step} effective channel on $\Psi=\{-\tfrac12,\tfrac12\}$ is
\[
\mathbb{R}^{(s)}_{\Psi}(\varphi)
\;=\;
\mathbb{Q}^{(s)}(\varphi)\,\mathbb{P}^{(s)}(\varphi),
\]
so that root-to-boundary transmission can be studied with the standard
broadcasting/reconstruction formalism for trees.

\subsection{Entropy rate, extremality, reconstruction, and phase transition: one phenomenon in three languages}
\label{subsec:entropy-extremality-reconstruction}

A rooted tree provides a common mathematical skeleton for three parallel questions across distinct scientific domains:
(i) in \textbf{statistical mechanics} \cite{Bleher1990}, whether the disordered (free-boundary) Gibbs measure is \emph{extremal} (representing a pure physical phase) or \emph{ nonextremal } (a mixture of phases);
(ii) in \textbf{information theory} \cite{Reza1994,Evans2000}, whether the root symbol can be inferred from observations at the boundary, a problem known as \emph{reconstruction} \cite{Mossel2004}; and 
(iii) in \textbf{evolutionary biology and phylogenetics} \cite{Muffato2023}, whether an ancestral state can be recovered from extant taxa at the leaves of the tree.

In this study, we demonstrate that for the mixed-spin $(s, 1/2)$ Ising model, these three perspectives converge at the critical stability threshold of the disordered phase, where the onset of reconstruction coincides with the loss of extremality and the emergence of a nontrivial entropy rate.

\paragraph{Extremality $\leftrightarrow$ non-reconstruction (tail triviality)}
For tree-indexed Markov specifications, extremality of the free-boundary Gibbs measure
is equivalent to triviality of the tail $\sigma$-field, hence to the impossibility of
reconstructing the root from the boundary as $n\to\infty$ \cite{Mossel2004PT-Phylogeny-entropy}. If $X_0$ is the root state and
$Y_n$ is the configuration at level $n$, \emph{non-reconstruction} may be stated as
\[
\lim_{n\to\infty}
\bigl\|
\mathbb{P}(X_0\in\cdot \mid Y_n)-\mathbb{P}(X_0\in\cdot)
\bigr\|_{\mathrm{TV}}
=0,
\qquad\text{equivalently}\qquad
\lim_{n\to\infty} I(X_0;Y_n)=0.
\]
In our mixed-spin setting, the disordered TISGM $\mu_{0}^{(s)}$ (the symmetric fixed
point $Z=1$) is extremal precisely when boundary influence contracts inward; it becomes
non-extremal once a long-range memory mode survives, making the reconstruction possible.

\paragraph{Two-stage channel and the effective eigenmode}
Because information propagates in two alternating steps, the persistence of correlations
is captured by the second largest eigenvalue of the $2\times 2$ kernel
$\mathbb{R}^{(s)}_{\Psi}(\varphi)$. Writing
\[
\lambda_2(\varphi)\;=\;\lambda_2\!\bigl(\mathbb{R}^{(s)}_{\Psi}(\varphi)\bigr),
\]
the Kesten--Stigum bound provides a sufficient reconstruction criterion:
\[
k\,|\lambda_2(\varphi)|^2 \;>\; 1
\quad\Longrightarrow\quad
\text{reconstruction (hence non-extremality) of }\mu_0^{(s)}.
\]
Intuitively, weak single-edge correlations can be amplified by branching, producing
nontrivial root-to-boundary memory at sufficiently strong coupling (i.e.\ $\varphi$ far
from $1$).

\paragraph{Dobrushin contraction as a strong data-processing statement}
A complementary sufficient extremality criterion is obtained from Dobrushin-type
contraction bounds. With contraction coefficients
$\tau_{\mathbb{P}^{(s)}}(\varphi)$ and $\tau_{\mathbb{Q}^{(s)}}(\varphi)$,
\[
k\,\tau_{\mathbb{P}^{(s)}}(\varphi)\,\tau_{\mathbb{Q}^{(s)}}(\varphi)\;<\;1
\quad\Longrightarrow\quad
\text{extremality (hence non-reconstruction) of }\mu_0^{(s)}.
\]
Information-theoretically, this expresses a strong data-processing phenomenon: boundary
perturbations lose distinguishability exponentially fast as they propagate inward.
Thus Dobrushin typically certifies extremality near $\varphi\simeq 1$ (high temperature),
while Kesten--Stigum certifies non-extremality at sufficiently strong coupling. Since
both are sufficient (not necessary), an intermediate parameter region where neither
applies naturally, calling for sharper reconstruction/uniqueness tools.

\paragraph{Where entropy rate enters (local randomness vs.\ long-range memory)}
Thermodynamic disorder along rays can be quantified by the entropy rate of the stationary
two-state Markov chain driven by $\mathbb{R}^{(s)}_{\Psi}(\varphi)$:
\[
h_{\mathrm{MC}}^{(s)}(\varphi)
=
-\sum_{i\in\Psi}\pi_i^{(s)}(\varphi)\sum_{j\in\Psi}R^{(s)}_{ij}(\varphi)\log R^{(s)}_{ij}(\varphi).
\]
This measures the average uncertainty generated per generation. Crucially,
$h_{\mathrm{MC}}^{(s)}$ reflects \emph{local} randomness, whereas extremality versus nonextremality reflects \emph{global} memory. Hence high entropy rate does not imply
high recoverability: near $\varphi\approx 1$ one typically has large
$h_{\mathrm{MC}}^{(s)}$ but small $|\lambda_2|$ (memory loss), while in low-temperature /
strong-coupling regimes $h_{\mathrm{MC}}^{(s)}$ can decrease even as $|\lambda_2|$
increases, and reconstruction becomes possible.

\subsection{Phase transitions as loss of stability of the disordered fixed point.}
From the dynamical-systems viewpoint, a phase transition occurs when the symmetric fixed
point $Z=1$ loses local stability (becomes repelling), leading to additional fixed points
and, hence, multiple TISGMs. The same control parameter $\varphi=e^{\beta J/2}$ therefore
simultaneously governs (a) stability/phase transition of the disordered phase,
(b) extremality vs.\ non-extremality (reconstruction), and (c) entropy-rate disorder.

\paragraph{Phylogenetic interpretation.}
In phylogenetic language, $X_0$ is an ancestral character at the root and $Y_n$ are the
extant taxa at depth $n$ \cite{Mossel2004}. The matrices $\mathbb{P}^{(s)}$ and $\mathbb{Q}^{(s)}$ act as
substitution/noise channels across the edges. The two-step kernel
$\mathbb{R}^{(s)}_{\Psi}=\mathbb{Q}^{(s)}\mathbb{P}^{(s)}$ quantifies the net
information passed from grandparent to grandchild when the intermediate generation is
latent, mirroring the practical ancestral state reconstruction. The non-reconstruction regime
corresponds to an ``inference limit'': beyond it, boundary data become asymptotically
uninformative regarding the root. Matrices $\mathbb{P}^{(s)}$ and $\mathbb{Q}^{(s)}$ act as time-homogeneous noise channels across edges. Due to this homogeneity, the two-step kernel $\mathbb{R}^{(s)}_{\Psi}=\mathbb{Q}^{(s)}\mathbb{P}^{(s)}$ remains invariant across generations, quantifying the net information passed from grandparent to grandchild at any depth.
\begin{table}[!htbp]
\centering
\renewcommand{\arraystretch}{1.25}
\setlength{\tabcolsep}{6pt}
\begin{tabular}{|c|p{3.9cm}|p{4.1cm}|p{4.1cm}|}
\hline
\textbf{Object / Criterion} &
\textbf{Information Theory} &
\textbf{Statistical Physics} &
\textbf{Biology / Phylogenetics} \\
\hline
$X_0$ &
Source/root symbol &
Root spin/state &
Common-ancestor trait/state \\
\hline
$Y_n$ &
Observations at depth $n$ (boundary data) &
Boundary configuration at level $n$ &
Extant taxa (leaves) at depth $n$ \\
\hline
Non-reconstruction &
$I(X_0;Y_n)\to 0$ (information loss) &
\textbf{Extremal} free-boundary Gibbs state (tail-trivial; no long-range memory) &
Signal erosion; $\mathrm{ASR}$ impossible in the limit \\
\hline
Reconstruction &
$\limsup_{n\to\infty} I(X_0;Y_n) > 0$ &
\textbf{Non-extremality} (nontrivial mixture; long-range memory) &
Signal retention; $\mathrm{ASR}$ possible \\
\hline
Threshold / bounds &
Spectral/broadcasting limits (e.g.\ KS: $b|\lambda_2|^2>1$) &
KS as non-extremality indicator; Dobrushin-type conditions as extremality guarantees &
Detectability threshold for ancestral inference (phase transition in $\mathrm{ASR}$) \\
\hline
\end{tabular}
\caption{One phenomenon in three languages: information flow, Gibbs extremality, and ancestral inference on trees. Here $\mathrm{ASR}$ denotes \emph{Ancestral State Reconstruction}.}
\label{tab:analogy}
\end{table}

Table \ref{tab:analogy} highlights the mathematical isomorphism between information theory, statistical physics and phylogenetics. This analogy centers on the \textbf{reconstruction problem}: whether a root state $X_0$ can be inferred from the noisy leaf observations $Y_n$ at depth $n$.

\begin{itemize}
    \item \textbf{Non-reconstruction Phase:} As $n \to \infty$, the mutual information $I(X_0; Y_n) \to 0$. This corresponds to a \textit{pure phase} in physics and a total \textit{signal erosion} in biology.
    \item \textbf{Reconstruction Phase:} A non-vanishing correlation remains ($\limsup I(X_0; Y_n) > 0$), implying \textit{long-range memory} and feasible \textit{ancestral inference}.
    \item \textbf{Critical Threshold:} The transition is governed by the \textbf{Kesten-Stigum bound} ($k\lambda_2^2 > 1$), defining the universal limit for information persistence across all three domains.
\end{itemize}

\section{Conclusions}\label{sec:conclusion}
We studied mixed spin-$(s,\tfrac12)$ Ising models on a \emph{semi-infinite} Cayley tree with branching number
$k=3$, focusing on translation-invariant splitting Gibbs measures (TISGMs). Beyond providing a computable link
between dynamical stability, reconstruction/extremality, and entropy production, a key point is that passing to
$k=3$ produces \emph{stability regions that differ markedly from those obtained in our earlier work} \cite{Akin2024,Akin-Muk-24JSTAT}: the location
and width of the stability/instability bands shift, and new parameter windows arise in which the disordered fixed
point changes its stability behavior.

For the representative case $s=5$, the associated recursion is an $11$-dimensional dynamical system.
A stability analysis identifies phase-transition regions through the loss of local stability of the disordered fixed
point, detected when the Jacobian satisfies $|\lambda_{\max}|\ge 1$. In particular, compared to the lower-branching
setting, the $k=3$ geometry yields a quantitatively different stability diagram, confirming that the branching
structure decisively affects the phase-transition mechanism encoded by the recursion.

At the disordered fixed point we derived explicit one-step stochastic matrices $\mathbb{P}^{(s)}(\varphi)$ and
$\mathbb{Q}^{(s)}(\varphi)$ and the effective two-step kernel
$\mathbb{R}^{(s)}(\varphi)=\mathbb{Q}^{(s)}(\varphi)\mathbb{P}^{(s)}(\varphi)$ on the spin-$\tfrac12$ layer.
Its nontrivial eigenvalue $\lambda_2(\varphi)$ controls correlation persistence and drives the classical
reconstruction tests~\cite{Mossel2001,MosselPeres2003}. With $k$ denoting the branching number (each vertex in shell $n$
has $k$ forward neighbors in shell $n+1$), the Kesten--Stigum condition reads $k\,\lambda_2(\varphi)^2>1$, and hence
$3\,\lambda_2(\varphi)^2>1$ in our setting. In parallel, Dobrushin-type contraction yields a rigorous sufficient
condition for extremality, leaving intermediate ``gray'' regimes where sharper tools are needed.

We also introduced the Markov entropy rate $h_{\mathrm{MC}}^{(s)}(\varphi)$ of the effective chain as a quantitative
observable measurement of uncertainty production per generation. An additional structural message concerns the natural magnitude of the two entropy rates: each is bounded by the logarithm of the effective state-space size of the induced Markov chain. In particular, the spin-$\tfrac12$ (two-state) description on $\Psi$ satisfies
\[
0\le h_{\mathrm{MC}}^{(s)}(\varphi)\le \log 2 \qquad \text{for all } s,
\]
whereas the induced chain on the spin-$s$ layer $\Phi$ has $2s+1$ states and hence
\[
0\le \widetilde{h}_{\mathrm{MC}}^{(s)}(\varphi)\le \log(2s+1),
\]
so the latter entropy-rate scale grows with $s$ and reflects the refinement from the two-state coarse-graining to the full $(2s+1)$-state dynamics. We presented numerical or graphical illustrations of the entropy rate behavior for values of \(s\) from 1 to 5, which confirmed the qualitative trends predicted by the spectral criteria.

 While $h_{\mathrm{MC}}^{(s)}(\varphi)$ is not, by itself,
an extremality diagnostic, it provides a computable thermodynamic/information-theoretic companion to the spectral criteria and helps quantify the disorder of the effective dynamics across parameter regimes.
A key qualitative message is that a \emph{local} phase-transition signal (loss of stability of the disordered fixed
point) need not coincide with the reconstruction threshold: there are parameter ranges in which local instability has
already initiated, yet the free-boundary disordered Gibbs measure remains extremal according to the available rigorous
criteria. This separation clarifies that ``phase transition'' and ``reconstruction'' capture distinct mechanisms on trees.

Overall, the methodology is flexible and can be extended to larger spin values, alternative coupling structures, or
more general hierarchical graphs, providing a systematic route to explore phase stability and long-range memory in
mixed-spin systems on the trees.

\appendix
\section{Construction of Stochastic Matrices at the Fixed Point $\ell_0^{(s)}$}

In this appendix we construct the stochastic matrices associated with the
$(s,\tfrac12)$ mixed-spin Ising model on a semi-infinite CT of order
three. Unlike symmetric single-spin models, mixed-spin systems lead
naturally to two distinct transition kernels, and hence to two stochastic
matrices, see~\cite{Akin2024,Akin-Muk-24JSTAT,Akin2023Chaos}. Using the
definitions in Eqs.~\eqref{StocMat1} and~\eqref{StocMat2}, we introduce
$\mathbb{P}^{(s)}(\varphi)\in\mathbb{R}^{(2s+1)\times 2}$ and
$\mathbb{Q}^{(s)}(\varphi)\in\mathbb{R}^{2\times(2s+1)}$, where
$\mathbb{P}^{(s)}$ has $2s+1$ rows and $2$ columns, while $\mathbb{Q}^{(s)}$
has $2$ rows and $2s+1$ columns. 
We put
\[
\varphi = e^{\beta J/2},
\qquad
\Phi = \{-s,-s+1,\dots,s-1,s\},
\qquad
\Psi = \Bigl\{-\tfrac12,\tfrac12\Bigr\}.
\]

\subsection{The stochastic matrix $\mathbb{P}^{(s)}(\varphi)$ for $k=3$}
At the symmetric fixed point $Z=1$ the transition probabilities from
$\Phi$ to $\Psi$ are
\[
P_{i,-\frac12}(\varphi)
=
\frac{1}{1+\varphi^{2i}},
\qquad
P_{i,\frac12}(\varphi)
=
\frac{\varphi^{2i}}{1+\varphi^{2i}},
\qquad i\in\Phi.
\]
Ordering the rows by $i=s,s-1,\dots,1,0,-1,\dots,-s$, one obtains the
$(2s+1)\times 2$ stochastic matrix
\begin{align}\label{eq:Matrix-P-arbitrary-s}
\mathbb{P}^{(s)}(\varphi)
=
\begin{pmatrix}
\dfrac{\varphi^{2s}}{\varphi^{2s}+1} & \dfrac{1}{\varphi^{2s}+1}\\[4pt]
\dfrac{\varphi^{2(s-1)}}{\varphi^{2(s-1)}+1} & \dfrac{1}{\varphi^{2(s-1)}+1}\\[2pt]
\vdots & \vdots\\[2pt]
\dfrac{\varphi^{2}}{\varphi^{2}+1} & \dfrac{1}{\varphi^{2}+1}\\[4pt]
\dfrac{1}{2} & \dfrac{1}{2}\\[4pt]
\dfrac{1}{1+\varphi^{2}} & \dfrac{\varphi^{2}}{1+\varphi^{2}}\\[2pt]
\vdots & \vdots\\[2pt]
\dfrac{1}{1+\varphi^{2(s-1)}} & \dfrac{\varphi^{2(s-1)}}{1+\varphi^{2(s-1)}}\\[4pt]
\dfrac{1}{1+\varphi^{2s}} & \dfrac{\varphi^{2s}}{1+\varphi^{2s}}
\end{pmatrix}.  
\end{align}
\subsection{The stochastic matrix $\mathbb{Q}^{(s)}(\varphi)$ for $k=3$}
Let us consider the fixed point associated wit the dynamical system given in Equations~\eqref{TI-1}–\eqref{TI-3} as 
\begin{align}\label{eq:FP-initial-arbitrary-s}
\ell_0^{(s)} :=
\begin{cases}
X_i^{(0)} = \left( \dfrac{\varphi^{2|i|} + 1}{2\varphi^{|i|}} \right)^{3}, 
& \quad i \in \{-s,\dots,-1,1,\dots,s\}, \\[6pt]
Z^{(0)} = 1, & \quad \text{otherwise}.
\end{cases}
\end{align}
At the symmetric fixed point $\ell_0^{(s)}$ the transition probabilities from
$\Psi$ to $\Phi$ can be written in the compact form
\begin{align}\label{eq:Matrix-Q-arbitrary-s}
\mathbb{Q}^{(s)}(\varphi)
=
\frac{1}{S_s(\varphi)}
\left(
\begin{array}{ccccccc}
h_s(\varphi) & h_{s-1}(\varphi) & \cdots & h_1(\varphi) &
1 &
g_1(\varphi) & \cdots g_s(\varphi)
\\[2pt]
g_s(\varphi) & g_{s-1}(\varphi) & \cdots & g_1(\varphi) &
1 &
h_1(\varphi) & \cdots h_s(\varphi)
\end{array}
\right),    
\end{align}
where the columns are indexed by $j=-s,-s+1,\dots,-1,0,1,\dots,s$, and
\[
h_n(\varphi)
=
\frac{(\varphi^{2n}+1)^3}{8\varphi^{2n}},
\qquad
g_n(\varphi)
=
\frac{(\varphi^{2n}+1)^3}{8\varphi^{4n}},
\qquad n=1,\dots,s,
\]
while the normalising factor $S_s(\varphi)$ is given by
\[
S_s(\varphi)
=
1+\frac18\sum_{n=1}^{s}\frac{(\varphi^{2n}+1)^4}{\varphi^{4n}}.
\]
By construction,
\[
\sum_{j=-s}^{s} Q^{(s)}_{ij}(\varphi) = 1,
\qquad i\in\Psi,
\]
so $\mathbb{Q}^{(s)}(\varphi)$ is a $2\times(2s+1)$ row–stochastic matrix.

The stochastic matrices given in \eqref{eq:Matrix-P-arbitrary-s} and \eqref{eq:Matrix-Q-arbitrary-s} can be composed to form
square stochastic matrices as follows
\[
\mathbb{R}^{(s)}(\varphi)=\mathbb{Q}^{(s)}(\varphi)\,\mathbb{P}^{(s)}(\varphi)
\in\mathbb{R}^{2\times 2},
\qquad\text{or}\qquad
\widetilde{\mathbb{R}}^{(s)}(\varphi)=\mathbb{P}^{(s)}(\varphi)\,\mathbb{Q}^{(s)}(\varphi)
\in\mathbb{R}^{(2s+1)\times(2s+1)}.
\]
These matrices will be used later to study the extremality versus non-extremality in the disordered phase and to compute the entropy rate.

In our earlier work \cite{Akin2024}, we  showed that the second-largest eigenvalue is the same for the
two compositions $\mathbb{R}^{(s)}(\varphi)=\mathbb{Q}^{(s)}(\varphi)\mathbb{P}^{(s)}(\varphi)$
and $\widetilde{\mathbb{R}}^{(s)}(\varphi)=\mathbb{P}^{(s)}(\varphi)\mathbb{Q}^{(s)}(\varphi)$.
Although both orderings therefore yield the same Kesten--Stigum threshold, the
computational effort is very different: as $s$ grows, extracting eigenvalues
from $\widetilde{\mathbb{R}}^{(s)}(\varphi)$ becomes increasingly costly since it
is a $(2s+1)\times(2s+1)$ matrix. In contrast, $\mathbb{R}^{(s)}(\varphi)$ is
always $2\times2$, so its eigenvalues—and in particular the second-largest
one—can be obtained in closed form with minimal effort.

\subsection{Resulting $2 \times 2$ Stochastic Matrix $\mathbb{R}^{(s)}(\varphi)=\mathbb{Q}^{(s)}(\varphi) \, \mathbb{P}^{(s)}(\varphi)$}
The matrix $\mathbb{R}^{(s)}(\varphi)$ is defined by the product $\mathbb{R}^{(s)}(\varphi) = \mathbb{Q}^{(s)}(\varphi) \, \mathbb{P}^{(s)}(\varphi)$. Given that $\mathbb{Q}^{(s)}$ is $2 \times (2s+1)$ and $\mathbb{P}^{(s)}$ is $(2s+1) \times 2$, the product is a $2 \times 2$ row-stochastic matrix exhibiting a symmetric structure:
\begin{align}\label{eq:Staochastic-Mat-R-s}
\mathbb{R}^{(s)}(\varphi) = \frac{1}{S_s(\varphi)} \begin{pmatrix} A(\varphi) & B(\varphi) \\ B(\varphi) & A(\varphi) \end{pmatrix}.  
\end{align}
\section*{Component Definitions}

\subsection*{Intermediate Functions (from $\mathbb{Q}^{(s)}$)}
The functions $h_n(\varphi)$ and $g_n(\varphi)$ are:
$$
h_n(\varphi) = \frac{(\varphi^{2n}+1)^3}{8\varphi^{2n}},
\qquad
g_n(\varphi) = \frac{(\varphi^{2n}+1)^3}{8\varphi^{4n}},
\qquad n=1,\dots,s.
$$

\subsection*{Normalization Factor $S_s(\varphi)$}
The factor derived from the row sum of $\mathbb{Q}^{(s)}$ is
$$
S_s(\varphi) = 1+\frac18\sum_{n=1}^{s}\frac{(\varphi^{2n}+1)^4}{\varphi^{4n}}.
$$

\subsection*{Diagonal Sum $A(\varphi)$}
The diagonal elements $R_{i,i}$ are calculated using the component sum $A(\varphi)$:
$$
A(\varphi) = \frac{1}{2} + \sum_{n=1}^{s} \left[ h_{n} \frac{\varphi^{2n}}{\varphi^{2n}+1} + g_{n} \frac{1}{\varphi^{2n}+1} \right]
$$
After algebraic simplification ($\tau_n^A$):
$$
A(\varphi) = \frac{1}{2} + \sum_{n=1}^{s} \frac{(\varphi^{2n}+1)^2(\varphi^{4n}+1)}{8\varphi^{4n}}
$$

\subsection*{Off-Diagonal Sum $B(\varphi)$}
The off-diagonal elements $R_{i,k}$ ($i \neq k$) are calculated using the component sum $B(\varphi)$:
$$
B(\varphi) = \frac{1}{2} + \sum_{n=1}^{s} \left[ h_{n} \frac{1}{\varphi^{2n}+1} + g_{n} \frac{\varphi^{2n}}{\varphi^{2n}+1} \right]
$$
After algebraic simplification ($\tau_n^B$):
$$
B(\varphi) = \frac{1}{2} + \sum_{n=1}^{s} \frac{(\varphi^{2n}+1)^2}{4\varphi^{2n}}.
$$

For illustration, let us take $s=5$. In this case the effective $2\times2$
matrix $\mathbb{R}^{(5)}$ on $\Psi=\{-\tfrac12,\tfrac12\}$ is given by
\begin{align}\label{app-eq:Stochastis-M-R-s=5}
\mathbb{R}^{(5)}(\varphi)
&=
\mathbb{Q}^{(5)}(\varphi)\,\mathbb{P}^{(5)}(\varphi)
=
\begin{pmatrix}
\displaystyle \sum_{i=-5}^{5} \frac{X_i \varphi^{-i}}{S_{-}^{(5)}} \, P^{(5)}_{i,-\tfrac12} &
\displaystyle \sum_{i=-5}^{5} \frac{X_i \varphi^{-i}}{S_{-}^{(5)}} \, P^{(5)}_{i,\tfrac12} \\[2mm]
\displaystyle \sum_{i=-5}^{5} \frac{X_i \varphi^{i}}{S_{+}^{(5)}} \, P^{(5)}_{i,-\tfrac12} &
\displaystyle \sum_{i=-5}^{5} \frac{X_i \varphi^{i}}{S_{+}^{(5)}} \, P^{(5)}_{i,\tfrac12}
\end{pmatrix},   
\end{align}
where
\[
P^{(5)}_{i,-\frac12}
=
\frac{\exp\!\bigl(-\tfrac12 i \beta J + \widetilde{h}_{-\tfrac12}\bigr)}
     {\displaystyle\sum_{u\in\{-\tfrac12,\tfrac12\}}
      \exp\!\bigl(iu \beta J + \widetilde{h}_{u}\bigr)},
\qquad
P^{(5)}_{i,\frac12}
=
\frac{\exp\!\bigl(\tfrac12 i \beta J + \widetilde{h}_{\tfrac12}\bigr)}
     {\displaystyle\sum_{u\in\{-\tfrac12,\tfrac12\}}
      \exp\!\bigl(iu \beta J + \widetilde{h}_{u}\bigr)}.
\]

At the symmetric fixed point $Z=1$, using \eqref{dynamicals3aa1}–\eqref{dynamicals3aa2}
we have
\begin{align}\label{IP-FP-1a}
X_{\pm 1}^{(0)} = \frac{(\varphi^{2}+1)^3}{8\varphi^{3}},\
X_{\pm 2}^{(0)} =\frac{(\varphi^{4}+1)^3}{8\varphi^{6}},\
X_{\pm 3}^{(0)} = \frac{(\varphi^{6}+1)^3}{8\varphi^{9}},\
X_{\pm 4}^{(0)} =\frac{(\varphi^{8}+1)^3}{8\varphi^{12}},\
X_{\pm 5}^{(0)}  = \frac{(\varphi^{10}+1)^3}{8\varphi^{15}}
\end{align}
and $X_0^{(0)}=1$. 

Substituting these values into the expression for
$\mathbb{Q}^{(5)}(\varphi)$ and into the normalising factors
\[
S_{-}^{(5)}(\varphi) = \sum_{j=-5}^{5} X_j^{(0)} \varphi^{-j},
\qquad
S_{+}^{(5)}(\varphi) = \sum_{j=-5}^{5} X_j^{(0)} \varphi^{j},
\]
we obtain, by symmetry, $S_{-}^{(5)}(\varphi)=S_{+}^{(5)}(\varphi)=:S^{(5)}(\varphi)$ with
\begin{equation}\label{S-rho-s5}
S^{(5)}(\varphi)
=
1 + \sum_{m=1}^{5} \frac{(\varphi^{2m}+1)^{4}}{8\,\varphi^{4m}}.
\end{equation}

Thus the effective $2\times2$ transition matrix on the spin-$\tfrac12$ states
$\Psi=\{-\tfrac12,\tfrac12\}$ can be written as
\[
\mathbb{R}^{(5)}(\varphi)
=
\begin{pmatrix}
R^{(5)}_{-\frac12,-\frac12}(\varphi) & R^{(5)}_{-\frac12,\frac12}(\varphi)\\[2pt]
R^{(5)}_{\frac12,-\frac12}(\varphi)  & R^{(5)}_{\frac12,\frac12}(\varphi)
\end{pmatrix},
\]
where
\begin{align*}
R^{(5)}_{-\frac12,-\frac12}(\varphi)
&=
\sum_{j=-5}^{5} Q^{(5)}_{-\frac12,j}(\varphi)\,P^{(5)}_{j,-\frac12}(\varphi)
=
\frac{N^{(5)}(\varphi)}{S^{(5)}(\varphi)},\\[4pt]
R^{(5)}_{\frac12,\frac12}(\varphi)
&=
\sum_{j=-5}^{5} Q^{(5)}_{\frac12,j}(\varphi)\,P^{(5)}_{j,\frac12}(\varphi)
=
\frac{N^{(5)}(\varphi)}{S^{(5)}(\varphi)},
\end{align*}
and, by symmetry,
\[
R^{(5)}_{-\frac12,\frac12}(\varphi)
=
R^{(5)}_{\frac12,-\frac12}(\varphi)
=
1-\frac{N^{(5)}(\varphi)}{S^{(5)}(\varphi)}.
\]
Here the numerator $N^{(5)}(\varphi)$ is given explicitly by
\begin{equation}\label{N-rho-s5}
N^{(5)}(\varphi)
=
\frac{1}{2}
+
\sum_{m=1}^{5}
\frac{(\varphi^{2m}+1)^{3}\,\bigl(1+\varphi^{4m}\bigr)}
     {8\,\varphi^{4m}\bigl(1+\varphi^{2m}\bigr)}.
\end{equation}

Consequently, at the fixed point \eqref{IP-FP-1a}, the matrix
$\mathbb{R}^{(5)}(\varphi)$ takes the compact symmetric form
\begin{equation}\label{R-rho-s5-final}
\mathbb{R}^{(5)}(\varphi)
=
\begin{pmatrix}
p^{(5)}(\varphi) & 1-p^{(5)}(\varphi)\\[2pt]
1-p^{(5)}(\varphi) & p^{(5)}(\varphi)
\end{pmatrix},
\qquad
p^{(5)}(\varphi) = \frac{N^{(5)}(\varphi)}{S^{(5)}(\varphi)},
\end{equation}
with $S^{(5)}(\varphi)$ and $N^{(5)}(\varphi)$ defined in
\eqref{S-rho-s5}–\eqref{N-rho-s5}.



\begin{thebibliography}{99}

\bibitem{Akin2024}
Ak\i n, H.
New disordered phases of the $(s,1/2)$-mixed spin Ising model for arbitrary spin $s$.
\emph{Chaos, Solitons \& Fractals}, \textbf{189}, 115733 (2024).
\url{https://doi.org/10.1016/j.chaos.2024.115733}

\bibitem{Mossel2001}
Mossel, E.
Reconstruction on trees: Beating the second eigenvalue.
\emph{Annals of Applied Probability}, \textbf{11}, 285--300 (2001).
\url{https://doi.org/10.1214/aoap/998926994}


\bibitem{Mossel2004} Mossel, E. Survey: Information Flow on Trees. Graphs, Morphisms and Statistical Physics. In DIMACS Series Discrete Mathematics and Theoretical Computer Science 63; American Mathematical Society: Providence, RI, USA, 2004; pp. 155–170.

\bibitem{MosselPeres2003}
Mossel, E., and Peres, Y.
Information flow on trees.
\emph{Annals of Applied Probability}, \textbf{13}, 817--844 (2003).
\url{https://doi.org/10.1214/aoap/1060202828}

\bibitem{CoverThomas2006}
Cover, T.~M., and Thomas, J.~A.
\emph{Elements of Information Theory}, 2nd ed., Wiley, Hoboken, NJ, 2006. \url{https://doi.org/10.1002/047174882X}


\bibitem{Cavender1978} Cavender, J. Taxonomy with condence. Math. BioSci. 1978, 40, 271–280. \url{https://doi.org/10.1016/0025-5564(78)90089-5}


\bibitem{Preston1974} Preston, C. Gibbs States on Countable Sets; Cambridge University Press: London, UK, 1974. \url{https://doi.org/10.1017/CBO9780511897122}

\bibitem{Spitzer1975} Spitzer, F. Markov random fields on an infinite tree. Ann. Probab. 1975, 3, 387–394. \url{https://doi.org/10.1214/aop/1176996347}


\bibitem{Muffato2023} Muffato, M., Louis, A., Nguyen, N.T.T. et al. Reconstruction of hundreds of reference ancestral genomes across the eukaryotic kingdom. Nat Ecol Evol 7, 355–366 (2023). \url{https://doi.org/10.1038/s41559-022-01956-z}

\bibitem{Reza1994} Reza, F. M. An introduction to information theory (Courier Corporation, 1994).

\bibitem{Georgii1988}
Georgii, H.-O.
\emph{Gibbs Measures and Phase Transitions}.
De Gruyter, Berlin/New York, 2011.
\url{https://doi.org/10.1515/9783110250329}

\bibitem{Rozikov2013}
Rozikov, U.~A.
\emph{Gibbs Measures on Cayley Trees}.
World Scientific, Singapore; 2013.
\url{https://doi.org/10.1142/8841}

\bibitem{DemboMontanari2010}
Dembo, A., and Montanari, A.
Ising models on locally tree-like graphs.
\emph{Annals of Applied Probability}, \textbf{20}(2), 565--592 (2010).
\url{https://doi.org/10.1214/09-AAP627}
\bibitem{Gandolfo-Haydar-2013}
Gandolfo, D., Haydarov, F.~H., Rozikov, U.~A., and Ruiz, J.
New phase transitions of the Ising model on Cayley trees.
\emph{J. Stat. Phys.}, \textbf{153}, 400--411 (2013).
\url{https://doi.org/10.1007/s10955-013-0836-3}

\bibitem{Gandolfo-Rahmet-2013}
Gandolfo, D., Rakhmatullaev, M.~M., Rozikov, U.~A., and Ruiz, J.
On free energies of the Ising model on the Cayley tree.
\emph{J. Stat. Phys.}, \textbf{150}, 1201--1217 (2013).
\url{https://doi.org/10.1007/s10955-013-0713-0}

\bibitem{EAP16}
De La Espriella, N., Arenas, A.~J., and Paez Meza, M.~S.
Critical and compensation points of a mixed spin-2-spin-5/2 Ising ferrimagnetic system with crystal field and nearest and next-nearest neighbors interactions.
\emph{J. Magn. Magn. Mater.}, \textbf{417}, 434--441 (2016).
\url{https://doi.org/10.1016/j.jmmm.2016.05.046}

\bibitem{EBM18}
De La Espriella, N., Buendia, G.~M., and Madera, J.~C.
Mixed spin-1 and spin-2 Ising model: Study of the ground states.
\emph{J. Phys. Commun.}, \textbf{2}, 025006 (2018).
\url{https://doi.org/10.1088/2399-6528/aaa39b}

\bibitem{strecka2006}
Stre\v{c}ka, J., Canov\'a, L., and Dely, J.
Weak universal critical behaviour of the mixed spin-$(1/2,S)$ Ising model on the Union Jack (centered square) lattice: Integer versus half-odd-integer spin-$S$ case.
\emph{Physica Status Solidi (b)}, \textbf{243}(1), 194--206 (2006).
\url{https://doi.org/10.1002/pssb.200642018}

\bibitem{akin2024-CJP}
Ak\i n, H.
Exploring the phase transition challenge by analyzing stability in a 5-D dynamical system linked to (2,1/2)-MSIM.
\emph{Chinese Journal of Physics}, \textbf{91}, 494--504 (2024).
\url{https://doi.org/10.1016/j.cjph.2024.08.008}


\bibitem{Akin-Muk-24JSTAT}
Ak\i n, H., and Mukhamedov, F.
The extremality of disordered phases for the mixed spin-(1,1/2) Ising model on a Cayley tree of arbitrary order.
\emph{J. Stat. Mech.} (2024) 013207.
\url{https://doi.org/10.1088/1742-5468/ad1be2}

\bibitem{Akin-Ulusoy-2023}
Ak\i n, H., and Ulusoy, S.
A new approach to studying the thermodynamic properties of the $q$-state Potts model on a Cayley tree.
\emph{Chaos, Solitons \& Fractals}, \textbf{174}, 113811 (2023).
\url{https://doi.org/10.1016/j.chaos.2023.113811}

\bibitem{Akin-Ulusoy-2022}
Ak\i n, H., and Ulusoy, S.
Limiting Gibbs measures of the $q$-state Potts model with competing interactions.
\emph{Physica B: Condensed Matter}, \textbf{640}, 413944 (2022).
\url{https://doi.org/10.1016/j.physb.2022.413944} \bibitem{Akin2025Chaos}
Ak\i n, H.
$q$-state modified Potts model on a Cayley tree and its phase transition in antiferromagnetic region.
\emph{Chaos, Solitons \& Fractals}, \textbf{199} (Part 2), 116746 (2025).
\url{https://doi.org/10.1016/j.chaos.2025.116746}

\bibitem{Akin2023Chaos}
Ak\i n, H.
The classification of disordered phases of mixed spin (2,1/2) Ising model and the chaoticity of the corresponding dynamical system.
\emph{Chaos, Solitons \& Fractals}, \textbf{167}, 113086 (2023).
\url{https://doi.org/10.1016/j.chaos.2022.113086}

\bibitem{Has-Far-2021}
Ak\i n, H., and Mukhamedov, F.
Phase transition for the Ising model with mixed spins on a Cayley tree.
\emph{J. Stat. Mech.} (2022) 053204.
\url{https://doi.org/10.1088/1742-5468/ac68e4}

\bibitem{Külske2009a}
K\"ulske, C., and Formentin, M.
A symmetric entropy bound on the non-reconstruction regime of Markov chains on Galton--Watson trees.
\emph{Electronic Communications in Probability}, \textbf{14}, 587--596 (2009).
\url{https://doi.org/10.1214/ECP.v14-1516}

\bibitem{ioffe1996}
Ioffe, D.
On the extremality of the disordered state for the Ising model on the Bethe lattice.
\emph{Lett. Math. Phys.}, \textbf{37}(2), 137--143 (1996).
\url{https://doi.org/10.1007/BF00416016}

\bibitem{MartinelliSinclairWeitz2004}
Martinelli, F., Sinclair, A., and Weitz, D.
Glauber dynamics on trees: Boundary conditions and mixing time.
\emph{Communications in Mathematical Physics}, \textbf{250}(2), 301--334 (2004).
\url{https://doi.org/10.1007/s00220-004-1147-y}

\bibitem{KuelskeRozikov2017}
K\"ulske, C., and Rozikov, U.~A.
Fuzzy transformations and extremality of Gibbs measures for the Potts model on a Cayley tree.
\emph{Random Structures \& Algorithms}, \textbf{50}, 636--678 (2017).
\url{https://doi.org/10.1002/rsa.20671}

\bibitem{Martin2003}
Martin, J.~B.
Reconstruction thresholds on regular trees.
In \emph{Discrete Random Walks} (Paris 2003), Discrete Math. Theor. Comput. Sci. Proc., Vol.~AC,
Assoc.\ Discrete Math.\ Theor.\ Comput.\ Sci., Nancy, pp.~191--204 (2003).

\bibitem{MartinelliSinclairWeitz2007}
Martinelli, F., Sinclair, A., and Weitz, D.
Fast mixing for independent sets, coloring and other models on trees.
\emph{Random Structures \& Algorithms}, \textbf{31}, 134--172 (2007).
\url{https://doi.org/10.1002/rsa.20132}

\bibitem{Moraal1978}
Moraal, H.
Ising spin systems on Cayley tree-like lattices: Spontaneous magnetization and correlation functions far from the boundary.
\emph{Physica A}, \textbf{92}, 305--314 (1978).
\url{https://doi.org/10.1016/0378-4371(78)90037-7}


\bibitem{RR-2021}
Rahmatullaev, M.~M., and Rasulova, M.~A.
Extremality of translation-invariant Gibbs measures for the Potts-SOS model on the Cayley tree.
\emph{Journal of Statistical Mechanics: Theory and Experiment}, \textbf{2021}(7), 073201 (2021).
\url{https://doi.org/10.1088/1742-5468/ac08ff}

\bibitem{Rozikov2018}
Rozikov, U.~A., Khakimov, R.~M., and Khaidarov, F.~K.
Extremality of the translation-invariant Gibbs measures for the Potts model on the Cayley tree.
\emph{Theor. Math. Phys.}, \textbf{196}, 1043--1058 (2018).
\url{https://doi.org/10.1134/S0040577918070103}

\bibitem{KuelskeRozikov2015}
K\"ulske, C., and Rozikov, U.~A.
Extremality of translation-invariant phases for a three-state SOS-model on the binary tree.
\emph{J. Stat. Phys.}, \textbf{160}(3), 659--680 (2015).
\url{https://doi.org/10.1007/s10955-015-1279-9}

\bibitem{Ak-Mukh-2025-PS} Ak\i n, H., \& Mukhamedov, F.  3-state hybrid Potts–SOS model with different coupling constants and its phase transition phenomenon. Physica Scripta, 100(8), 085201 (2025).
\url{https://doi.org/10.1088/1402-4896/adefb2}

\bibitem{Qawasmeh-Ak-Mukh-2025-PS} Qawasmeh, A., Mukhamedov, F., \& Ak\i n, H. Analysis of Gibbs measures and stability of dynamical system linked to (1,1/2)-mixed Ising model on $(m,k)$-ary trees. Math Phys Anal Geom 28, 10 (2025). \url{https://doi.org/10.1007/s11040-025-09504-4}

\bibitem{Seino1992}
Seino, M.
The free energy of the random Ising model on the Bethe lattice.
\emph{Physica A: Statistical Mechanics and its Applications}, \textbf{181}(3--4), 233--242 (1992).
\url{https://doi.org/10.1016/0378-4371(92)90087-7}

\bibitem{Petersen-2020}
Petersen, K., and Salama, I.
Entropy on regular trees.
\emph{Discrete \& Continuous Dynamical Systems}, \textbf{40}(7), 4453 (2020).
\url{https://doi.org/10.3934/dcds.2020186}

\bibitem{Kuelske2004}
K\"ulske, C., Le Ny, A., and Redig, F.
Relative entropy and variational properties of generalized Gibbsian measures.
\emph{Ann. Probab.}, \textbf{32}(2), 1691--1726 (2004).
\url{https://doi.org/10.1214/009117904000000342}
\bibitem{Feutrill2021}
Feutrill, A., and Roughan, M.
A review of Shannon and differential entropy rate estimation.
\emph{Entropy}, \textbf{23}, 1046 (2021).
\url{https://doi.org/10.3390/e23081046}

\bibitem{Brandani2013}
Brandani, G.~B., Schor, M., MacPhee, C.~E., Grubm\"uller, H., and Zachariae, U.
Quantifying disorder through conditional entropy: An application to fluid mixing.
\emph{PLoS ONE}, \textbf{8}(6), e65617 (2013).
\url{https://doi.org/10.1371/journal.pone.0065617}

\bibitem{Akin2025IJMPB}
Ak\i n, H.  
Thermodynamic analysis of the $q$-state Akin model on a Cayley tree.
\emph{International Journal of Modern Physics B}, \textbf{39} (32),  2550285 (2025).
\url{https://doi.org/10.1142/S0217979225502856}

\bibitem{Rahmatullaev2025FreeEnergyEntropy}
Rahmatullaev, M.M., and Burxonova, Z.A.
Free energy and entropy for the constructive Gibbs measures of the Ising model on the Cayley tree of order three,
\emph{Nanosystems: Physics, Chemistry, Mathematics} 16(3) (2025) 261--273. \url{https://doi.org/10.17586/2220-8054-2025-16-3-261-273}

\bibitem{Akin2022}
Ak\i n, H.
Calculation of the free energy of the Ising model on a Cayley tree via the self-similarity method.
\emph{Axioms}, \textbf{11}(12), 703 (2022).
\url{https://doi.org/10.3390/axioms11120703}

\bibitem{Akin2022CJP}
Ak\i n, H.
A novel computational method of the free energy for an Ising model on Cayley tree of order three.
\emph{Chinese Journal of Physics}, \textbf{77}, 2276--2287 (2022).
\url{https://doi.org/10.1016/j.cjph.2022.01.016}


\bibitem{Akin23CJP}
Ak\i n, H.
Quantitative behavior of (1,1/2)-MSIM on a Cayley tree.
\emph{Chinese Journal of Physics}, \textbf{83}, 501--514 (2023).
\url{https://doi.org/10.1016/j.cjph.2023.04.014}

\bibitem{Shiryaev-1980}
Shiryaev, A.~N.
\emph{Probability}.
Springer, Moscow, 1980.
\url{https://doi.org/10.1007/978-0-387-72206-1}

\bibitem{Mukhamedov-2022}
Mukhamedov, F.
Extremality of disordered phase of $\lambda$-model on Cayley trees.
\emph{Algorithms}, \textbf{15}, 18 (2022).
\url{https://doi.org/10.3390/a15010018}

\bibitem{KestenStigum1966}
Kesten, H., and Stigum, B.~P.
Additional limit theorems for indecomposable multidimensional Galton--Watson processes.
\emph{Ann. Math. Statist.}, \textbf{37}(6), 1463--1481 (1966).
\url{https://doi.org/10.1214/aoms/1177699139}

\bibitem{Bleher1990}
Bleher, P.~M.
Extremity of the disordered phase in the Ising model on the Bethe lattice.
\emph{Communications in Mathematical Physics}, \textbf{128}(2), 411--419 (1990).
\url{https://doi.org/10.1007/BF02108787}








































\bibitem{Evans2000} Evans, W., Kenyon, C. and Peres,  Y. Broadcasting on trees and the Ising model,
The Annals of Applied Probability \textbf{10} (2000), no.~2, 410--433. \url{https://doi.org/10.1214/aoap/1019487349}

\bibitem{Mossel2004PT-Phylogeny-entropy}
Mossel, E.  Phase transitions in phylogeny,
\emph{Transactions of the American Mathematical Society}
\textbf{356}(6) (2004), 2379--2404.
\url{https://doi.org/10.1090/S0002-9947-03-03382-8}.


\bibitem{Mathematica}
Wolfram Research, Inc.
\emph{Mathematica}, Version 13.3, Champaign, IL (2023).

\end{thebibliography}
\end{document}